\documentclass{lmcs}
\pdfoutput=1

\usepackage{lastpage}
\lmcsdoi{20}{4}{12}
\lmcsheading{}{\pageref{LastPage}}{}{}%
{Dec.~18,~2023}{Nov.~12,~2024}{}

\usepackage[utf8]{inputenc}


\pdfoutput=1 

\usepackage{hyperref}
\usepackage[english]{babel}
\usepackage{microtype}
\usepackage{bm}
\usepackage[noadjust]{cite}
\usepackage{xspace}
\usepackage{latexsym}
\usepackage{amsmath,amssymb,amstext}
\usepackage{wasysym}
\usepackage{centernot}
\usepackage{esvect}
\usepackage{quiver}
\usepackage{cite}
\usepackage{graphicx}
\usepackage{stmaryrd}
\usepackage{tikz}

\usetikzlibrary{arrows,arrows.meta,automata,backgrounds,calc,chains,decorations.fractals,decorations.pathreplacing,fadings,fit,folding,mindmap,patterns,plotmarks,positioning,shadows,shapes.geometric,shapes.symbols,through,trees,cd}
\colorlet{dblue}{blue!40!black}

\usepackage{textcomp}
\usepackage{enumitem}
\usepackage{verbatim}
\usepackage{float}
\usepackage{pdfpages}
\usepackage{wrapfig}
\usepackage{xspace}
\usepackage{mathtools}
\usepackage[all,cmtip]{xy}
\usepackage{subfig}
\usepackage{adjustbox}
\usepackage{stackrel}
\usepackage{changepage}
\usepackage[newfloat, frozencache,cachedir=minted-cache]{minted}
\usepackage[pdf]{graphviz}
\usepackage{xpatch}
\usepackage{fancyvrb}

\newcommand{\II}{\catname{I}}

\newcommand{\scala}[1]{\mintinline{scala}{#1}}
\newcommand{\java}[1]{\mintinline{java}{#1}}

\newenvironment{code}{\captionsetup{type=listing}}{}
\SetupFloatingEnvironment{listing}{name=Listing}

\newcommand{\graphIndex}{%
	\,\begin{tikzpicture}[baseline=-.6ex]
		\node (x) [circle,inner sep=0,outer sep=1mm,minimum size=0.7mm,fill=black] {};
		\node (y) [circle,inner sep=0,outer sep=1mm,minimum size=0.7mm,fill=black, right of=x, xshift=-3mm] {};
		\draw [->] ([yshift=.7mm]x.east) to ([yshift=.7mm]y.west);
		\draw [->] ([yshift=-.7mm]x.east) to ([yshift=-.7mm]y.west);
	\end{tikzpicture}\,%
}

\newcommand{\functorcat}[2]{[#1, #2]}
\newcommand{\dpoLL}[1]{DPO$^{{#1}}$}

\theoremstyle{definition}
\theoremstyle{definition}\newtheorem{notation}[thm]{Notation}

\newcommand{\pbpostrong}{PBPO$^{+}$\xspace}

\newcommand{\pbpopluspdf}{\texorpdfstring{\pbpostrong{}}{PBPO+}}
\newcommand{\catname}[1]{{\normalfont\textbf{#1}}\xspace}
\newcommand{\FinSet}{\catname{FinSet}}
\newcommand{\Set}{\catname{Set}}

\newcommand{\Graph}{\catname{Graph}}
\newcommand{\FinGraph}{\catname{FinGraph}}

\newcommand{\labels}{\mathcal{L}}

\newcommand{\CC}{\catname{C}}

\newcommand{\AR}{\mathcal{A}}
\newcommand{\ARset}[2]{\AR(#1, #2)}

\newcommand{\sys}{\mathcal{T}}

\newcommand{\Iso}{\mathrm{iso}}
\newcommand{\Mono}{\mathrm{mono}}
\newcommand{\Reg}{\mathrm{rm}}
\newcommand{\Hom}{\mathrm{Hom}}

\newcommand{\ClassC}[2]{{#1(#2)}}
\newcommand{\HomC}[1]{\ClassC{\Hom}{#1}}
\newcommand{\MonoC}[1]{\ClassC{\Mono}{#1}}
\newcommand{\IsoC}[1]{\ClassC{\Iso}{#1}}
\newcommand{\RegC}[1]{\ClassC{\Reg}{#1}}
\newcommand{\ARC}[1]{\ClassC{\AR}{#1}}

\newcommand{\id}[1]{\mathrm{1}_{#1}}

\newcommand{\inverse}[1]{{{#1}^{{-1}}}}

\newcommand{\obj}[1]{\mathrm{Ob}(#1)}
\newcommand{\homset}[2]{\Hom(#1,#2)}
\newcommand{\homsetmono}[2]{\Mono(#1,#2)}

\let\emptyset\varnothing

\newcommand{\card}[1]{|#1|}

\newcommand{\mono}{\rightarrowtail}

\newcommand{\epi}{\twoheadrightarrow}
\newcommand{\leftmono}{\leftarrowtail}
\newcommand{\iso}{\cong}
\newcommand{\leftto}{\leftarrow}

\newcommand{\dom}{\mathit{dom}}

\tikzset{default/.style={
  thick,
  every node/.style={circle},
  level distance=12mm, 
  inner sep=.5mm}}
\tikzset{smallCircle/.style={circle,fill=black,inner sep=0mm,outer sep=1mm,minimum size=1mm}}
\tikzset{paint/.style={very thick,draw=#1!50!black,fill=#1,opacity=.4}}
\tikzset{paintopaque/.style={very thick,draw=#1!50!black!60,fill=#1!60}}
\tikzset{loop/.style={out=-30+#1,in=30+#1,distance=2.7em,pos=0.5}}
\tikzset{thinloop/.style={out=-20+#1,in=20+#1,distance=2.7em,pos=0.5}}

\tikzset{emptystep/.style={-,dotted,line cap=round,dash pattern=on 0 off 3.00000}}

\tikzset{b/.style={anchor=north,at=(#1.south)}}
\tikzset{br/.style={anchor=north west,at=(#1.south east)}}
\tikzset{bl/.style={anchor=north east,at=(#1.south west)}}
\tikzset{bw/.style={anchor=north west,at=(#1.south west)}}
\tikzset{be/.style={anchor=north east,at=(#1.south east)}}
\tikzset{a/.style={anchor=south,at=(#1.north)}}
\tikzset{ar/.style={anchor=south west,at=(#1.north east)}}
\tikzset{al/.style={anchor=south east,at=(#1.north west)}}
\tikzset{aw/.style={anchor=south west,at=(#1.north west)}}
\tikzset{ae/.style={anchor=south east,at=(#1.north east)}}
\tikzset{r/.style={anchor=west,at=(#1.east)}}
\tikzset{l/.style={anchor=east,at=(#1.west)}}
\tikzset{rn/.style={anchor=north west,at=(#1.north east)}}

\tikzset{eta/.style={very thick,->,cblue!90!black}}
\tikzset{beta/.style={very thick,->,corange!80!white!90!black}}

\tikzset{devcirc/.style={circle,draw,fill=white,inner sep=0,minimum size=4\pgflinewidth}}
\tikzset{dev/.style={postaction={decorate},decoration={
  markings,
  mark=at position .5 with \node [devcirc] {};}}
}
    
\tikzset{medium tree/.style={
    level 1/.style={sibling distance=17mm},
    level 2/.style={sibling distance=9mm},
    level 3/.style={sibling distance=5mm},
    level 4/.style={sibling distance=4mm},
  }}

\tikzset{startAt/.style={inner sep=0mm,r=#1,xshift=-1.2mm,yshift=.8mm}}

\tikzset{graphNode/.style={circle,draw=black,inner sep=.5mm,outer sep=.5mm}}
\tikzset{morph/.style={thick,-{>[width=1.7mm]}}}
\tikzset{mono/.style={thick,{>[width=1.7mm]}-{>[width=1.7mm]}}}
\tikzset{edge/.style={line width=0.2mm,-{Triangle[width=1.4mm]}}}

\newcommand{\arrowTriangle}[2]{
  \pgfpathmoveto{\pgfpoint{-0.01#1+#2}{.6#1}}
  \pgfpathlineto{\pgfpoint{#1+#2}{0}}
  \pgfpathlineto{\pgfpoint{-0.01#1+#2}{-.6#1}}
  \pgfusepathqfill
}

\newdimen\prearrowsize
\newdimen\arrowsize
\newdimen\temparrowsize
\newcommand{\arrowscale}{5}
\newcommand{\setarrowsize}{
  \arrowsize=0.000000001pt
  \prearrowsize=\arrowscale\pgflinewidth
  \normalizearrowsize
}
\newcommand{\normalizearrowsize}{
  \ifdim\prearrowsize>2mm
    \addtolength{\arrowsize}{2mm}
    \addtolength{\prearrowsize}{-2mm}
    \temparrowsize=0.5\prearrowsize
    \prearrowsize=\temparrowsize
  \else
  \fi

  \addtolength{\arrowsize}{\prearrowsize}
}

\pgfarrowsdeclarealias{share>}{share>}{stealth'}{stealth'}%

\pgfarrowsdeclare{red>}{red>}
{
  \setarrowsize
  \pgfarrowsleftextend{1.5\pgflinewidth} 
  \pgfarrowsrightextend{\arrowsize-0.1\arrowsize}
}{
  \setarrowsize
  \pgfsetdash{}{0pt} 
  \pgfsetroundjoin 
  \pgfsetroundcap 
  \arrowTriangle{\arrowsize}{-0.1\arrowsize}
}

\pgfarrowsdeclare{red>>}{red>>}
{
  \setarrowsize
  \pgfarrowsleftextend{1.5\pgflinewidth}
  \pgfarrowsrightextend{\arrowsize-0.1\arrowsize+.8\arrowsize}
}{
  \setarrowsize
  \pgfsetdash{}{0pt} 
  \pgfsetroundjoin 
  \pgfsetroundcap 
  \arrowTriangle{\arrowsize}{-0.1\arrowsize}
  \arrowTriangle{\arrowsize}{-0.1\arrowsize+.8\arrowsize}
}

\pgfarrowsdeclare{red>>>}{red>>>}
{
  \setarrowsize
  \pgfarrowsleftextend{1.5\pgflinewidth}
  \pgfarrowsrightextend{\arrowsize-0.1\arrowsize+1.6\arrowsize}
}{
  \setarrowsize
  \pgfsetdash{}{0pt} 
  \pgfsetroundjoin 
  \pgfsetroundcap 
  \arrowTriangle{\arrowsize}{-0.1\arrowsize}
  \arrowTriangle{\arrowsize}{-0.1\arrowsize+.8\arrowsize}
  \arrowTriangle{\arrowsize}{-0.1\arrowsize+1.6\arrowsize}
}

\tikzstyle{gyellow}=[draw=black!80,top color=white!50,bottom color=black!20]
\tikzstyle{gblue}=[draw=blue!50,top color=white,bottom color=blue!60]
\tikzstyle{gred}=[draw=red!50,top color=white,bottom color=red!60]
\tikzstyle{ggreen}=[draw=blue!80!green!90!black,top color=white,bottom color=blue!80!green!60]

\tikzstyle{roundNode}=[gyellow,thick,circle,minimum size=4mm,inner sep=0.5mm]

\definecolor{cblue}{rgb}{0,0.4,0.7}
\definecolor{clighterblue}{rgb}{0,0.6,1.0}
\colorlet{cred}{red}
\colorlet{cgreen}{green!80!black}
\colorlet{corange}{orange!70!red}
\colorlet{cpureorange}{orange}
\colorlet{cpurple}{clighterblue!50!cred}

\colorlet{clightblue}{clighterblue!50!cblue!40}
\colorlet{clightred}{cred!40}
\colorlet{clightgreen}{cgreen!80!cblue!40}
\colorlet{clightyellow}{corange!40!yellow!50}
\colorlet{clightorange}{cred!50!orange!40}
\colorlet{clightpurple}{clighterblue!50!cred!50}

\colorlet{cdarkred}{cred!70!black}
\colorlet{cdarkgreen}{cgreen!60!black}
\colorlet{cdarkblue}{cblue!60!black}

\colorlet{chighlight}{orange!50!yellow!60}

\tikzset{pgnode/.style={smallCircle,fill=white,draw=black,minimum size=1.3mm,outer sep=0.5mm}}
\tikzset{pgnodecolor/.style={pgnode,minimum size=4mm,scale=0.9}}
\tikzset{pgnodebig/.style={roundNode,gyellow,outer sep=1mm}}
\tikzset{pgrelation/.style={ultra thick,cblue!80!black,decorate,decoration={snake,amplitude=.4mm,segment length=4mm}}}

\tikzset{exi/.style={densely dotted}}
\tikzset{decweak/.style={-,green!50!black,very thick,opacity=0.7}}
\tikzset{decstrict/.style={->,orange,very thick,opacity=0.7}}

\tikzset{graphNode/.style={circle,draw=black,inner sep=.5mm,outer sep=1mm}}

\tikzset{sloop/.style={looseness=7}}

\tikzset{interconnect/.style={dotted,cred}}
\tikzset{match/.style={cdarkgreen,very thick}}
\tikzset{jigsaw/.style={circle,draw=black,minimum size=2mm,fill=white,minimum size=3.5mm}}

\colorlet{myblue}{blue!80!black}
\colorlet{mygreen}{cdarkgreen}
\colorlet{myred}{cred}
\colorlet{myorange}{corange}
\colorlet{mypurple}{blue!40!cred}

\tikzset{smalljigsaw/.style={rectangle,rounded corners=3mm,inner sep=1mm}}

\tikzset{epat/.style={thick}}
\tikzset{eset/.style={densely dotted}}
\tikzset{npattern/.style={rectangle,rounded corners=2mm,draw=black,inner sep=0.5mm,outer sep=.5mm,minimum size=4.5mm}}
\tikzset{nset/.style={npattern,draw=black,fill=white,densely dotted}}
\tikzset{label/.style={scale=0.85,inner sep=0,outer sep=0.5mm}}

\tikzset{short/.style={node distance=10mm}}

\newcommand{\graphbox}[8]{
  \begin{scope}[xshift=#2,yshift=#3]
    \draw [rounded corners=2mm] (0,0) rectangle (#4,-#5);
    \node at (0,0mm) [anchor=north west,inner sep=1mm] {#1};
    \begin{scope}[xshift=#4/2+#6,yshift=#7] 
    #8
    \end{scope}
  \end{scope}
}

\newcommand{\graphboxx}[8]{
  \begin{scope}[xshift=#2,yshift=#3]
    \draw [rounded corners=2mm] (0,0) rectangle (#4,-#5);
    \node (l) at (0,0mm) [anchor=north west,inner sep=1mm] {$t_{#1}$};
    \begin{scope}[xshift=#4/2+#6,yshift=#7] 
    #8
    \end{scope}
  \end{scope}
}

\newcommand{\graphboxy}[8]{
  \begin{scope}[xshift=#2,yshift=#3]
    \draw [rounded corners=2mm] (0,0) rectangle (#4,-#5);
    \node (l) at (0,0mm) [anchor=north west,inner sep=1mm] {$t_{#1}$};
    \begin{scope}[xshift=#4/2+#6,yshift=#7] 
    #8
    \end{scope}
  \end{scope}
}

\newcommand{\vertex}[2]{%
  \begin{tikzpicture}[baseline=-1ex]%
    \node [rectangle,rounded corners=2mm,inner sep=0.5mm,fill=#2] {$#1$};%
  \end{tikzpicture}%
}

\newcommand{\ssrc}{\mathit{s}}

\newcommand{\stgt}{\mathit{t}}

\tikzset{styleNode/.style={graphNode,rectangle,rounded corners=2mm,fill=white,inner sep=1.5mm,outer sep=0}}

\tikzset{styleEdge/.style={->,shorten <= 2mm,shorten >= 2mm}}
\tikzset{styleEdgeNode/.style={rectangle,rounded corners=2mm,fill=white,draw=none,outer sep=0,inner sep=1mm,pos=0.45}}

\tikzcdset{scale cd/.style={every label/.append style={scale=#1},
    cells={nodes={scale=#1}}}}

\newcommand{\lbl}{\ell}

\newcommand{\lble}{\lbl^E}

\newcommand{\nat}{\mathbb{N}}

\newcommand{\MM}{\mathcal{M}}

\newcommand{\regmono}{\hookrightarrow}

\newcommand{\pb}[2]{\langle#1\mid#2\rangle}

\makeatletter
\providecommand{\leftsquigarrow}{%
  \mathrel{\mathpalette\reflect@squig\relax}%
}
\newcommand{\reflect@squig}[2]{%
  \reflectbox{$\m@th#1\rightsquigarrow$}%
}
\makeatother

\tikzset{AR/.style={line join=round,
decorate, decoration={
    zigzag,
    segment length=4,
    amplitude=.9,post=lineto,
    post length=2pt
}}}

\newcommand{\mysidepicture}[5]{
    \nointerlineskip\noindent
    \begin{tikzpicture}[overlay]%
      \node at (\textwidth,0) [anchor=north east,rectangle,inner sep=0,outer sep=0,xshift=#2,yshift=#3] {#4};%
    \end{tikzpicture}%
    \begin{adjustwidth}{0cm}{#1}%
    #5
    \end{adjustwidth}%
}

\newcommand{\pbpostrongsteprule}[3]{
	{#1} \Rightarrow_{\mathrm{PBPO}^{+}}^{#3} {#2}
}

\colorlet{linkcolor}{red!60!black}
\hypersetup{
    colorlinks=true,
    linkcolor=linkcolor,
    citecolor=linkcolor,
    filecolor=linkcolor,      
    urlcolor=linkcolor,
    pdftitle={Termination of Graph Transformation Systems Using Weighted Subgraph Counting},
    }

\begin{document}

\author[R.~Overbeek]{Roy Overbeek\lmcsorcid{0000-0003-0569-0947}}
\author[J.~Endrullis]{J\"org Endrullis\lmcsorcid{0000-0002-2554-8270}}

\address{Vrije Universiteit Amsterdam, De Boelelaan 1105, 1081 HV Amsterdam}

\email{r.overbeek@vu.nl, j.endrullis@vu.nl}

\title[Termination of Graph Transformation Using Weighted Subgraph Counting]{Termination of Graph Transformation Systems\texorpdfstring{\\}{} Using Weighted Subgraph Counting\rsuper*}
\titlecomment{{\lsuper*} This paper extends and improves the identically titled paper~\cite{overbeek2023termination} published at ICGT2023.}

\begin{abstract}
    We introduce a termination method for the algebraic graph transformation framework \pbpostrong{}, in which we weigh objects by summing a class of weighted morphisms targeting them.
    The method is well-defined in rm-adhesive quasitoposes (which include toposes and therefore many graph categories of interest), and is applicable to non-linear rules.
    The method is also defined for other frameworks, including SqPO and left-linear DPO, because we have previously shown that they are naturally encodable into \pbpostrong{} in the quasitopos setting. We have implemented our method, and the implementation includes a REPL that can be used for guiding relative termination proofs.
\end{abstract}

\maketitle

\section{Introduction}
\label{section:introduction}

Many fields of study related to computation have mature termination theories. See, for example, the corpus for (first-order) term rewriting systems~\cite[Chapter 6]{terese} (for a more recent (but less systematic) discussion, see, e.g., \cite{yamada2022tuple}).

For the study of graph transformation, by contrast, not many termination methods exist, and the ones that do exist are usually defined for rather specific notions of graphs. Although the techniques themselves can be interesting, the latter observation fits somewhat uneasily with the general philosophy of the algebraic graph transformation tradition, in which graph transformations are defined and studied in a graph-agnostic manner.

In this chapter, we introduce a termination method for \pbpostrong{}. We weigh objects $G$ by summing a class of weighted elements (i.e., morphisms of the form $T \to G$), and construct a decreasing measure. Our method enjoys generality across two dimensions:
\begin{enumerate}
	\item The method is formulated completely in categorical terms, and is well-defined in (locally finite) rm-adhesive quasitoposes.
	\item The method is also defined for SqPO~\cite{corradini2006sesqui}, AGREE~\cite{corradini2020algebraic}, PBPO~\cite{corradini2019pbpo}, and left-linear DPO~\cite{ehrig1973graph}. This is because we have recently shown that, in the quasitopos setting, each rule of these formalisms can be straightforwardly encoded as a \pbpostrong{} rule that generates the same rewrite relation~\cite[Theorem 73]{overbeek2023quasitoposes}.
\end{enumerate}

To the best of our knowledge, this is the first termination method applicable in such a broad setting; and the first method that is automatically defined for a variety of well-known algebraic graph transformation frameworks. In addition, the termination method can be applied to non-linear (duplicating) rules.

This paper is structured as follows. We summarize some basic categorical and termination preliminaries (Section~\ref{section:preliminaries}). We then provide some useful background on quasitoposes, and required background on \pbpostrong{} (Section~\ref{section:background}). Next, we explain and prove our termination method (Section~\ref{section:decreasingness}). After, we illustrate our method with a variety of examples (Section~\ref{section:examples}), and compare our approach to related work (Section~\ref{section:related:work}). We then present an implementation of our method, which includes a REPL (read-eval print loop) that can be used for guiding relative termination proofs (Section~\ref{section:implementation}). We close with some concluding remarks and pointers for future work (Section~\ref{section:conclusion}).

\begin{rem}
	This paper extends and improves the identically titled conference paper~\cite{overbeek2023termination}. The most noteworthy differences are as follows:
	\begin{enumerate}
		\item All previously elided proofs, Example~\ref{example:plump:string}, and Example~\ref{ex:plump:jungle:eval:hypergraph} are now included. Previously they were available on arXiv~\cite{overbeek2023termination-arxiv} only due to space restrictions.
		\item The background survey on quasitoposes (Section~\ref{sec:background:quasitoposes}) is newly included, in order to make the paper more self-contained, and to improve the motivation.
		\item Instead of requiring the existence of (essentially unique) $\AR$-factorizations~\cite[Definition 9]{overbeek2023termination}, we use a simpler and weaker condition (Definition~\ref{def:morphism:preserves:factorization}), which we then relate to standard factorization systems (Proposition~\ref{prop:preserving:a:factorization:and:factorization:systems}).
		\item Section~\ref{section:implementation} contains completely new material (and the implementation described therein is also new).
	\end{enumerate}
	The new additions include previously unpublished material taken from the PhD thesis of the first author~\cite{overbeek2024phdthesis}.
\end{rem}

\section{Preliminaries}
\label{section:preliminaries}

\newcommand{\ARar}{\rightsquigarrow}
\newcommand{\leftARar}{\leftsquigarrow}

The preliminaries for this paper include basic categorical and graph notions (Section~\ref{section:general:notions}), and a basic understanding of termination (Section~\ref{section:termination}).  

\subsection{Basic Notions}
\label{section:general:notions}

We assume familiarity with basic categorical notions such as (regular) monomorphisms, pullbacks and pushouts~\cite{barr1990category,pierce1991basic}. We write $\mono$ for monos; and $\HomC{\CC}$, $\MonoC{\CC}$, $\RegC{\CC}$ and $\IsoC{\CC}$ for the classes of morphisms, monomorphisms, regular monomorphisms and isomorphisms in $\CC$, respectively.

\begin{notation}[Nonstandard Notation]
	\label{notation:nonstandard:notation}
	Given a class of morphisms $\ARC{\CC}$, we write $\ARset{A}{B}$ to denote the collection of $\AR$-morphisms from $A$ to $B$, leaving $\CC$ implicit. For sets of objects $S$, we overload $\ARset{S}{A}$ to denote $\bigcup_{X \in S} \ARset{X}{A}$. If $\ARC{\CC}$ is a generic class in lemmas, we use $\ARar$ to denote $\AR$-morphisms.
	
	For cospans $A \stackrel{f}{\to} C \stackrel{g}{\leftarrow} D$, we write $\pb{f}{g}$ to denote the arrow $B \to D$ obtained by pulling $f$ back along $g$.
\end{notation}

\begin{defi}[$\AR$-Local Finiteness]
	Let $\ARC{\CC}$ be a class of morphisms. A category $\CC$ is \emph{$\AR$-locally finite} if $\ARset{A}{B}$ is finite for all $A,B \in \obj{\CC}$.
\end{defi}

\begin{lem}[{Pullback Lemma~\cite[Proposition 2.5.9]{borceux1994handbook1}}]
	\label{lemma:pullback:lemma}
	Assume the right square of
	\begin{center}
		\begin{tikzcd}[row sep=3mm, column sep=5mm,ampersand replacement=\&]
			A \arrow[r] \arrow[d] \& B \arrow[d] \arrow[r]                        \& C \arrow[d] \\
			D \arrow[r]           \& E \arrow[r] \arrow[ru, "\mathrm{PB}", phantom] \& F          
		\end{tikzcd}%
	\end{center}
	is a pullback and the left square commutes. Then the outer rectangle (obtained by composing the horizontal morphisms) is a pullback iff the left square is a pullback. \qed
\end{lem}%

\newcommand{\rightinv}[1]{{#1^\leftarrow}}

\begin{defi}[Split Epimorphism]
	An epimorphism $e : A \epi B$ is \emph{split} if it has a right inverse, i.e., if there exists an $f : B \to A$ such that $e \circ f = \id{B}$.
	
	For split epimorphisms $e$, we let $\rightinv{e}$ denote an arbitrary right inverse of $e$.
\end{defi}

Observe that right inverses need not be unique. In \Set, for instance, the constant function $k : \{ 0, 1 \} \to \{ 0 \}$ has two right inverses.

\begin{propC}[{\cite[Prop.\ 7.59]{adamek2009joy}}]
	\label{prop:split:epi:facts}
	If $e$ is a split epi, then $\rightinv{e} \in \RegC{\CC}$.
	\qed
\end{propC}

Our method is defined fully in categorical terms.
For examples, and to guide intuition, we will use the category of edge-labeled multigraphs.

\begin{defi}[Graph Notions]
	\label{def:graph}
	Let a finite label set $\labels$ be fixed. An (edge-labeled) \emph{(multi)graph} $G$ consists of a set of vertices $V$, a set of edges $E$, source and target functions $\ssrc,\stgt : E \to V$, and an edge label function $\lble : E \to \labels$.
	A graph is \emph{unlabeled} if $\labels$ is a singleton.
	
	A \emph{homomorphism} between graphs $G$ and $G'$ is a pair of maps
	$
	\phi = (\phi_V : V_G \to V_{G'}, \phi_E : E_G \to E_{G'})
	$
	satisfying $(s_{G'}, t_{G'}) \circ \phi_E = \phi_V \circ (s_G, t_G)$ and $\lble_{G'} \circ \phi_E = \lble_{G}$.
\end{defi}

\begin{defiC}[{\cite{ehrig2006fundamentals}}]
	\label{definition:category:graph}
	The category $\Graph$ has graphs as objects, parameterized over some global (and usually implicit) label set $\labels$, and homomorphisms as arrows. The subcategory $\FinGraph$ restricts to graphs with finite $V$ and $E$.
\end{defiC}

\newcommand{\step}[5]{{#1} \mathrel{\Rightarrow_{#5}^{{#3}, {#4}}} {#2}}
\newcommand{\dpostep}[4]{\step{#1}{#2}{#3}{#4}{\mathrm{DPO}}}

Although our termination method is defined for \pbpostrong{}~(Section~\ref{sec:background:pbpoplus}), we will state a key result~(Theorem~\ref{thm:pbpostrong:models:others}) that involves (left-linear) DPO.

\begin{defi}[DPO Rewriting~\cite{ehrig1973graph}]
	\label{def:dpo}
	In its most general formulation, a \emph{DPO rewrite rule} $\rho$ is a span $L \stackrel{l}{\leftto} K \stackrel{r}{\to} R$, and a diagram 
	\begin{center}
		\begin{tikzcd}
			L \arrow[d, "m" description] & K \arrow[l, "l" description] \arrow[r, "r" description] \arrow[d] \arrow[ld, "\mathrm{PO}", phantom] \arrow[rd, "\mathrm{PO}", phantom] & R \arrow[d] \\
			G_L                                & G_K \arrow[l] \arrow[r]                                                                                                                       & G_R        
		\end{tikzcd}%
	\end{center}
	defines a \emph{DPO rewrite step} $\dpostep{G_L}{G_R}{\rho}{m}$, i.e., a step from $G_L$ to $G_R$ using rule $\rho$ and match $m : L \to G_L$. 
	
	In some versions of DPO, morphism $l$ is required to belong to some class of monomorphisms $\MM$, to ensure that pushout complements exist uniquely. We will refer to this restricted version of DPO as \emph{left-linear DPO}, and denote it by \dpoLL{\MM}. Additionally, morphism $m$ is often required to be in $\MM$ because it increases expressivity~\cite{habel2001doublerevisited}. 
\end{defi}

\subsection{Termination}
\label{section:termination}

The topic of termination dates back at least to Turing, and is studied in many different settings. For a systematic overview for term rewriting systems (not yet existent for graph transformation systems), see~\cite[Chapter 6]{terese} (for a more recent discussion, see, e.g., \cite{yamada2022tuple}). Plump has shown that termination of graph rewriting is undecidable~\cite{plump98termination}. 

\begin{defi}
	\label{def:sequence}
	Define $I_n = \{ m \in \nat \mid m < n \}$ for $n \in \nat \cup \{\omega\}$, where $\omega$ extends the natural ordering of $\nat$ with $k < \omega$ for all $k \in \nat$.
	
	Let $R \subseteq A \times A$ be given.  An \emph{$R$-sequence (of length $n$)} is a function $f: I_n \to A$ such that for all $i < n$, $(f(i), f(i+1)) \in R$. The pair $(f(i), f(i+1))$ is said to be the \emph{$(i+1)$-th step} in this sequence. The sequence is \emph{infinite} if $n = \omega$.
\end{defi}

\begin{defi}
	A binary relation $R$ is \emph{terminating} if there does not exist an infinite $R$-sequence.
\end{defi}

\begin{defiC}[{\cite{bachmair1986commutation,klop1987term}}]
	Let $R, S \subseteq A \times A$ be binary relations.
	Then $R$ is \emph{terminating relative to} $S$ if every infinite $R \cup S$-sequence contains a finite number of $R$ steps.
\end{defiC}

\newcommand{\tileset}[1]{\mathbb{T}}
\newcommand{\TT}{\mathbb{T}}
\newcommand{\ww}{\mathbf{wt}}
\newcommand{\wwTAR}[2]{\ww_{#1}^{#2}}

For our purposes, it suffices to measure objects as natural numbers     (instead of a general well-founded order).

\begin{defi}[Measure]
	A \emph{measure} is a function $\ww : A \to \mathbb{N}$. 
	The measure $\ww$ is \emph{decreasing} for a binary relation $R \subseteq A \times A$ if for all $(x, y) \in R$, $\ww(x) > \ww(y)$, and it is \emph{non-increasing} for $R$ if for all $(x, y) \in R$, $\ww(x) \ge \ww(y)$.
\end{defi}

\begin{propC}[{\cite{bachmair1986commutation,klop1987term}}]
	\label{prop:measure}
	Let $R, S \subseteq A \times A$ be binary relations.
	Assume that there exists a measure $\ww$ that is decreasing for $R$ and non-increasing for~$S$.
	Then $R$ is terminating relative to $S$.
	Consequently $R \cup S$ is terminating iff $S$ is. \qed
\end{propC}

In a framework agnostic setting, a rule $\rho$ is a mathematical object that induces a binary relation ${\Rightarrow_\rho} \subseteq A \times A$.
We say that a rule (or a system of rules) is terminating, decreasing or non-increasing if the induced rewrite relations have the respective property (and analogously for relative termination). Note that also Proposition~\ref{prop:measure} can then be applied to systems of rules in place of relations.

\section{Background}
\label{section:background}

\newcommand{\stepRule}[3]{{#1} \Rightarrow_{#3} {#2}}
\newcommand{\stepRuleAlpha}[4]{{#1} \Rightarrow_{#3}^{#4} {#2}}

\subsection{Quasitoposes}
\label{sec:background:quasitoposes}

Our termination method makes a number of assumptions about the underlying category and about certain morphism classes. One natural and relevant setting in which these assumptions are satisfied is an rm-adhesive quasitopos. For motivational reasons, we therefore provide a brief background on quasitoposes. 

\begin{defi}[{Quasitopos~\cite{wyler1991lecture, adamek2009joy,johnstone2002sketches}}]
	A category $\CC$ is a \emph{quasitopos} if it has all finite limits and colimits, it is locally cartesian closed, and it has a regular-subobject classifier.
\end{defi}

The definition of a quasitopos itself need not be understood: 	
we are interested in specific properties of quasitoposes, and in examples.
Quasitoposes can also be understood in terms of the more familiar toposes.

\begin{prop}[{\cite[Proposition 19.6]{wyler1991lecture}, \cite[Proposition 28.6(2)]{adamek2009joy}}]
	\label{prop:a:topos:is:a:quasitopos:iff}
	$\CC$ is a topos iff $\CC$ is a quasitopos and $\RegC{\CC} = \MonoC{\CC}$. \qed
\end{prop}

In the graph transformation literature there is a variety of adhesivity properties, which depend on the notion of a Van Kampen square.

\begin{defi}[Van Kampen Square~\cite{lack2004adhesive}]
	\label{def:van:kampen}
	A pushout square is \emph{Van Kampen} (\emph{VK}) if, whenever it lies at the bottom of a commutative cube 
	
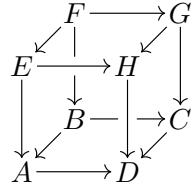
\begin{center}
		\begin{tikzcd}[row sep={20,between origins},column sep={20,between origins},ampersand replacement=\&,nodes={rectangle,inner sep=0.5mm,outer sep=0}]
			\& F  \ar[rr] \ar[dd] \ar[dl] \& \&  G \ar[dd] \ar[dl] \\
			E \ar[rr, crossing over] \ar[dd] \& \& H \\
			\& B \ar[rr] \ar[dl] \& \&  C \ar[dl] \\
			A \ar[rr] \&\& D \ar[from=uu,crossing over]
		\end{tikzcd}
	\end{center}
	of which the back faces $FBAE$ and $FBCG$ are pullbacks, this implies that 
	\[ 
	\text{the top face is a pushout $\iff$ the front faces $EHDA$ and $GHDC$ are pullbacks.}
	\]
\end{defi}

\begin{defi}[{Rm-Adhesive Category~\cite{lack2005adhesive}}]
	A category is $\emph{rm-adhesive}$ (a.k.a.\ \emph{quasiadhesive}) if pushouts along regular monomorphisms exist and are VK.
\end{defi}

Figure~\ref{figure:toposes:and:adhesivity:properties} situates quasitoposes among toposes and various adhesivity properties. Observe in particular that not all quasitoposes are rm-adhesive. A counterexample is found in the category of simple graphs~\cite[Corollary 20]{johnstone2007quasitoposes}.

The following properties of quasitoposes are relevant for this paper (either as requirements for the termination method or for exposition). Extended (and overlapping) summaries relevant for graph transformation are provided by Behr et al.~\cite[Corollary 5.15]{behr2023fundamentals} and by the authors~\cite[Proposition 36]{overbeek2023quasitoposes}.

\begin{prop}[Relevant Quasitopos Properties]
	\label{prop:quasitopos:properties:summary}
	A quasitopos $\CC$
	\begin{enumerate}
		\item has (by definition) all pushouts and pullbacks;
		\item satisfies
		\begin{enumerate}
			\item ${g,f \in \RegC{\CC}} \implies {g \circ f \in \RegC{\CC}}$; and
			\item $	{{g\circ f} \in \RegC{\CC}} \implies {f \in \RegC{\CC}}$;
		\end{enumerate}
		for all morphisms $g : B \to C$ and $f : A \to B$~\cite[Proposition 7.62(2), Corollary 28.6(2-3)]{adamek2009joy};
		\item pushouts along regular monomorphisms are pullbacks~\cite[Lemma A2.6.2]{johnstone2002sketches};
		\item regular monomorphisms are stable under pushout~\cite[Lemma A.2.6.2]{johnstone2002sketches};
		\item has essentially unique (epi, regular mono) factorizations for every morphism~\cite[Proposition 28.10]{adamek2009joy}, i.e., every morphism $f : A \to B$ factors uniquely as $A \stackrel{e}{\epi} C \stackrel{m}{\regmono} B$ (up to isomorphism in $C$), where $e$ is epic and $m$ regular monic;
		\item has essentially unique (regular epi, mono) factorizations for every morphism~\cite[Proposition 11]{johnstone2007quasitoposes}, i.e., every morphism $f : A \to B$ factors uniquely as $A \stackrel{e}{\epi} C \stackrel{m}{\mono} B$ (up to isomorphism in $C$), where $e$ is regular epic and $m$ is monic; and
		\item is stable under slicing, i.e., every slice category $\CC/X$ (for $X \in \obj{\CC}$) is a quasitopos~\cite[Theorem 19.4]{wyler1991lecture}.
	\end{enumerate}
\end{prop}

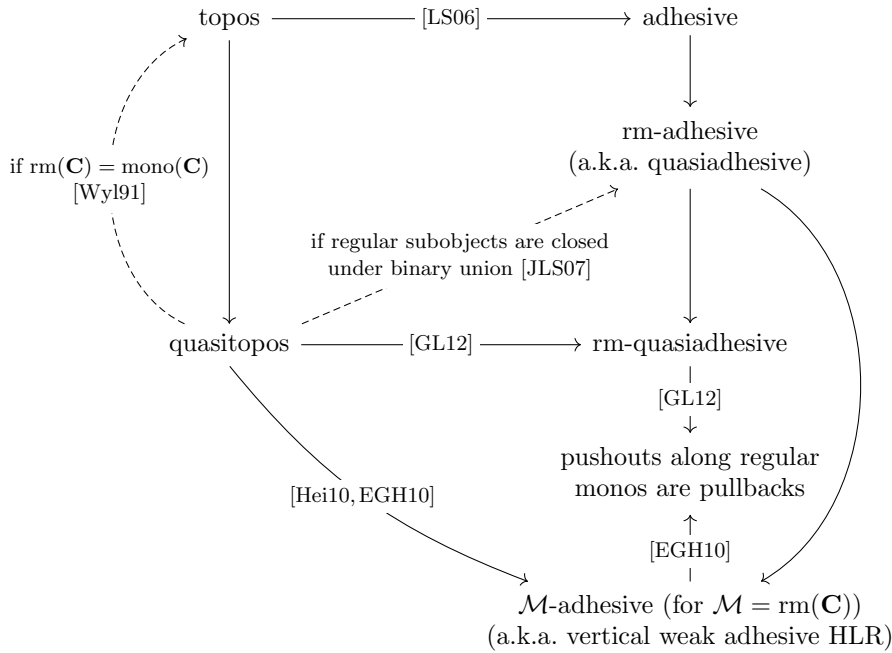
\begin{figure}
	\begin{center}
		\hspace*{4.5mm}%
		\scalebox{0.9}{
			\begin{tikzcd}[ampersand replacement=\&, row sep=10mm, column sep=12mm]
				\text{topos} \arrow[rr, "\text{\footnotesize \cite{lack2006toposes}}" description] \arrow[ddd] \& \& {\text{adhesive}} \arrow[d] \& \& \\
				\&  \& {\begin{tabular}{c}rm-adhesive \\ (a.k.a.\ quasiadhesive)\end{tabular}} \arrow[dd] \arrow[dddd, bend left=60] \&  \& \\
				\\
				{\text{quasitopos}} \ar[uuu, bend left=65, "\text{\footnotesize \begin{tabular}{c}if $\RegC{\CC} = \MonoC{\CC}$ \\ \cite{wyler1991lecture}\end{tabular}}" description, dashed] \arrow[rr, "\text{\footnotesize{\cite{garner2012axioms}}}" description] \arrow[ddrr, bend right=15,"\text{\footnotesize \cite{heindel2010hereditary, ehrig2010categorical}}" description,xshift=-3mm] \arrow[rruu, "\text{\footnotesize \begin{tabular}{c}if regular subobjects are closed \\ under binary union~\cite{johnstone2007quasitoposes}\end{tabular}}" description, dashed, xshift=3.5mm] \& \& {\text{rm-quasiadhesive}} \ar[d, "\text{\footnotesize{\cite{garner2012axioms}}}" description] \\ 
				\&  \& {\begin{tabular}{c}pushouts along regular \\ monos are pullbacks\end{tabular}} \\
				\& \& {{\begin{tabular}{c}$\mathcal{M}$-adhesive (for $\MM = \mathrm{rm}(\CC)$) \\ (a.k.a.\ vertical weak adhesive HLR)\end{tabular}}} \ar[u, "\text{\footnotesize \cite{ehrig2010categorical}}" description, pos=0.45]
			\end{tikzcd}
		}
	\end{center}
	\caption{Implications between (quasi)toposes and adhesivity properties.}
	\label{figure:toposes:and:adhesivity:properties}
\end{figure}

Many structures of interest are (quasi)toposes. We give some motivating examples of rm-adhesive quasitoposes.

\newcommand{\op}[1]{{#1}^\mathrm{op}}

\begin{exa}[(Co)presheaf Toposes]
	If $\CC$ is small, then the functor category $\functorcat{\CC}{\Set}$ is a topos~\cite[Theorem 26.2]{wyler1991lecture}, and hence an rm-adhesive quasitopos. Such categories are known as copresheaf toposes, $\CC$-sets~\cite{brown2023computationalJLAMP}, or graph structures~\cite{lowe1993algebraic} (and $\functorcat{\op{\CC}}{\Set}$ is known as a presheaf category).
	Many structures that are of interest to the graph transformation community can be defined in this manner. For example, if the label set $\labels$ is a singleton (i.e., graphs are unlabeled), then $\Graph \iso \functorcat{\graphIndex}{\Set}$).
\end{exa}

The following proposition assures us that such toposes are closed under finite restrictions. We are not aware of a similar principle for quasitoposes.

\begin{prop}
	\label{prop:finite:presheaf}
	If $\catname{I}$ is finite and $\CC \cong [\catname{I}, \FinSet]$, then $\CC$ is a $\Hom$-locally finite topos, and so an rm-adhesive quasitopos that is  $\AR$-locally finite for any $\ARC{\CC}$. 
\end{prop}
\begin{proof}
	$\CC$ is a topos~\cite[Example 5.2.7]{borceux1994handbook}, and it is locally finite because $\FinSet$ is $\Hom$-locally finite. And any topos is an rm-adhesive quasitopos (Figure~\ref{figure:toposes:and:adhesivity:properties}).
\end{proof}

\begin{exa}[Typed Graphs]
	Given the category of unlabeled graphs \Graph, a category of graphs typed over $X \in \obj{\Graph}$ can be constructed through a slice construction $\Graph/X$, which yields a quasitopos (Proposition~\ref{prop:quasitopos:properties:summary}). Typed graphs are subsumed by $\CC$-sets~\cite[Section 2.3]{brown2023computationalJLAMP}.
\end{exa}

The following is an example of an rm-adhesive quasitopos that is not a topos.

\begin{exa}[Fuzzy Presheaves]
	\label{ex:fuzzy:presheaves}
	If $(\labels, \leq)$ is a complete Heyting algebra, then the category of $(\labels, \leq)$-fuzzy presheaves is an rm-adhesive quasitopos~\cite[Theorems~15 and 41]{rosset2023fuzzy}. This includes the category of fuzzy graphs over a complete Heyting algebra $(\labels, \leq)$~\cite[Definition 76]{overbeek2023quasitoposes}. We have proposed fuzzy graphs as a mechanism for relabeling labeled graphs~\cite[Section 6]{overbeek2023quasitoposes}. In such a category, $\MonoC{\CC}$ contains all injective graph homomorphisms such that labels are non-decreasing ($\leq$) w.r.t.\ the given order, and $\RegC{\CC}$ contains only the injective graph homomorphisms that preserve labels ($=$).
\end{exa}

\subsection{\pbpopluspdf{}}
\label{sec:background:pbpoplus}

\pbpostrong{} is short for \emph{Pullback-Pushout with strong matching}. It is obtained by strengthening the matching mechanism of PBPO~\cite{corradini2019pbpo} by Corradini et al.

We provide the necessary definitions and results on \pbpostrong{}. See Section~\ref{section:examples} for many examples of rules. For a gentler introduction to \pbpostrong{}, with examples of rewrite steps, see the tutorial~\cite{overbeek2023tutorial} or the PhD thesis of the first author~\cite{overbeek2024phdthesis}.

\begin{defi}[{\pbpostrong{} Rewriting~\cite{corradini2019pbpo, overbeek2023quasitoposes}}]\label{def:pbpostrong:rewrite:step}%
	A \emph{\pbpostrong rule} $\rho$  is a diagram as shown on the left of:
	\begin{center}
		\begin{tikzcd}[ampersand replacement=\&,column sep=48,row sep=25,nodes={rectangle,inner sep=1mm,outer sep=0}]
			L \arrow[d, "t_L" description] \& K \arrow[d, "t_K" description] \arrow[l, "l" description] \arrow[r, "r" description] \arrow[ld, "\mathrm{PB}", phantom] \& R \\
			L'                                   \& K' \arrow[l, "l'" description]                                                                                                \&  
		\end{tikzcd}     
		\hspace{1cm}
		\begin{tikzcd}[ampersand replacement=\&,column sep=48,row sep=25, nodes={rectangle,inner sep=1mm,outer sep=0}]
			\&                                     \& K \arrow[d, "!u" description, densely dotted] \arrow[r, "r" description] \arrow[rd, "\mathrm{PO}", phantom] \& R \arrow[d, "w" description] \\
			L \arrow[r, "m" description] \arrow[d, equals] \arrow[rd, "\mathrm{PB}", phantom] \& G_L \arrow[d, "\alpha" description] \& G_K \arrow[l, "g_L" description] \arrow[r, "g_R" description] \arrow[d, "u'" description] \arrow[ld, "\mathrm{PB}", phantom]                            \& G_R                          \\
			L \arrow[r, "t_L" description]                                                                 \& L'                                  \& K' \arrow[from=uu, "t_K" description, bend left=38, pos=0.7, crossing over] \arrow[l, "l'" description]                                                                                                                          \&                             
		\end{tikzcd}
	\end{center}
	$L$ is the \emph{lhs pattern} of the rule, $L'$ its \emph{(context) type} and $t_L$ the \emph{(context) typing} of $L$. Likewise for the \emph{interface} $K$. $R$ is the \emph{rhs pattern} or \emph{replacement for $L$}.    
	
	Rule $\rho$, \emph{match morphism} $m : L \to G_L$ and \emph{adherence morphism} $\alpha : G_L \to L'$ induce a \emph{rewrite step} $\stepRuleAlpha{G_L}{G_R}{\rho}{\alpha}$ on arbitrary objects $G_L$ and $G_R$ if the properties indicated by the commuting diagram
	on the right
	hold, where $u : K \to G_K$ is the unique morphism satisfying $t_K = {u' \circ u}$~\cite[Lemma 15]{overbeek2023quasitoposes}.
\end{defi}

Lemma~\ref{lemma:computing:a:pbpoplus:step} suggests how to compute a rewrite step: a search for $u$ is in fact not required.

\begin{asm}
	\label{assumption:rule:po}
	We assume that for any rule $\rho$ analyzed for termination, the pushout of $t_K$ along $r$ exists.
\end{asm}

\begin{rem}
	\label{remark:pbpoplus:diagrams}
	If Assumption~\ref{assumption:rule:po} holds for a rule $\rho$,
	then any $\rho$ step defines (up to isomorphism) a diagram 
	\begin{center}
		\begin{tikzcd}[ampersand replacement=\&,column sep=60,row sep=20,nodes={rectangle,inner sep=1mm,outer sep=0},
			execute at end picture={
				\node at ($(\tikzcdmatrixname-1-1)!.5!(\tikzcdmatrixname-2-2)$) [anchor=center] {PB};
				\node at ($(\tikzcdmatrixname-2-1)!.5!(\tikzcdmatrixname-3-2)$) [anchor=center] {PB};
				\node at ($(\tikzcdmatrixname-1-2)!.5!(\tikzcdmatrixname-2-3)$) [anchor=center] {PO};
				\node at ($(\tikzcdmatrixname-2-2)!.5!(\tikzcdmatrixname-3-3)$) [anchor=center] {PO};
			}]
			\bm{L} \arrow[d, "m" description]  \& \bm{K} \arrow[d, "u" description] \arrow[r, "\bm{r}" description, thick] \arrow[l, "\bm{l}" description, thick] \arrow[ld, "", phantom] \arrow[rd, "", phantom] \& \bm{R} \arrow[d, "w" description] \arrow[dd, "\bm{t_R}" description, bend left=55, pos=0.7, thick] \\
			G_L \arrow[d, "\alpha" description]                                   \& G_K \arrow[d, "u'" description] \arrow[l, "g_L" description] \arrow[r, "g_R" description] \arrow[ld, "", phantom] \arrow[rd, "", phantom]                                   \& G_R \arrow[d, "w'" description]                                       \\
			\bm{L'}  \arrow[from=uu, "\bm{t_L}" description, bend left=55, crossing over, pos=0.7, thick]                                                                  \& \bm{K'} \arrow[l, "\bm{l'}" description, thick] \arrow[r, "\bm{r'}" description, thick] \arrow[from=uu, "\bm{t_K}" description, bend left=55, crossing over, pos=0.7, thick]                                                                                                                               \& \bm{R'}                                                            
		\end{tikzcd}
	\end{center}
	where the bold diagram is $\rho$ ($t_L \circ l = l' \circ t_K$ a pullback and $t_R \circ r = r' \circ t_K$ a pushout)~\cite[Section 3]{overbeek2023quasitoposes}. Our method uses the extra pushout to analyze how rewritten objects (the middle span) relate to the context types (the bottom span).
\end{rem}

\begin{rem}
	 This assumption is not restrictive. All colimits (and thus all pushouts) exist in quasitoposes, our main setting. Moreover, pushouts are guaranteed to exist along $\MM$-morphisms in $\MM$-adhesive categories, and morphism $t_K$ is usually in $\MM$.
\end{rem}

\newcommand{\framework}[1]{\mathcal{#1}}
\newcommand{\FF}{\framework{F}}
\newcommand{\GG}{\framework{G}}
\renewcommand{\models}{\prec}

\begin{thmC}[{\cite[Theorem 73]{overbeek2023quasitoposes}}]
	\label{thm:pbpostrong:models:others}
	Let $\CC$ be a quasitopos, and let matches $m \in \RegC{\CC}$.
	For rewriting formalisms $\FF$ and $\GG$, let $\FF \prec \GG$ express that in $\CC$, for any $\FF$ rule $\rho$, there exists a $\GG$ rule $\tau$ such that ${\Rightarrow^\rho_\FF} = {\Rightarrow^\tau_\GG}$. We have:
	\pushQED{\qed}
	\[
	\begin{tikzpicture}[default,nodes={rectangle,inner sep=3mm},baseline=(l.base)]
		\node (t) {\pbpostrong};
		\node (l) at (t.west) [anchor=east,yshift=-3.5mm]  
		{$\text{SqPO} \models \text{AGREE}$};
		\node (r) at (t.east) [anchor=west,xshift=-2mm,yshift=-3.5mm]  
		{$\text{\dpoLL{\RegC{\CC}}}$};
		\node (b) at ($(t)+(0,-7mm)$) {PBPO};
		\node at ($(l.north east)!0.5!(t.south west)$) [rotate=35] {$\models$};
		\node at ($(r.north west)!0.5!(t.south east) + (-.5mm,0mm)$) [rotate=180-35] {$\models$};
		\node at ($(t)!0.5!(b)$) [rotate=90] {$\models$};
	\end{tikzpicture} \qedhere
	\]
	\popQED
\end{thmC}

Observe that $\models$ is transitive. As the constructive proofs in~\cite{overbeek2023quasitoposes} show, the procedures to encode the mentioned formalisms into \pbpostrong{} are straightforward.\footnote{More precise statements and proofs can be found in the first author's PhD thesis~\cite[Section 4.5]{overbeek2024phdthesis}.} We moreover conjecture $\text{SPO} \prec \text{\pbpostrong{}}$~\cite[Remark 26]{overbeek2023quasitoposes}.

\section{Decreasingness by Counting Weighted Elements}
\label{section:decreasingness}

\newcommand{\edgeGraph}{
	\,\begin{tikzpicture}[baseline=-.6ex]
		\node (x) [circle,inner sep=0,outer sep=1mm,minimum size=0.7mm,fill=black] {};
		\node (y) [circle,inner sep=0,outer sep=1mm,minimum size=0.7mm,fill=black, right of=x, xshift=-3mm] {};
		\draw [->] (x) to (y);
	\end{tikzpicture}\,
}

\newcommand{\nodeGraph}{
	\,\begin{tikzpicture}[baseline=-.6ex]
		\node (x) [circle,inner sep=0,outer sep=1mm,minimum size=0.7mm,fill=black] {};
	\end{tikzpicture}\,
}

\newcommand{\aloop}[1]{%
	\begin{tikzpicture}[baseline=-.6ex]
		\useasboundingbox (-0.7ex,0) rectangle ++(4.5ex,1ex);
		\node (x) [circle,inner sep=0,outer sep=1mm,minimum size=0.7mm,fill=black] {};
		\draw [->] (x) to[loop=0,distance=1em] node [right,label] {$#1$} (x);
	\end{tikzpicture}
}

\newcommand{\aloopUnlabeled}{%
	\begin{tikzpicture}[baseline=-.6ex]
		\useasboundingbox (-0.7ex,0) rectangle ++(3.5ex,1ex);
		\node (x) [circle,inner sep=0,outer sep=1mm,minimum size=0.7mm,fill=black] {};
		\draw [->] (x) to[loop=0,distance=1em] (x);
	\end{tikzpicture}
}

We start with an explanation of the general idea behind our termination approach. Given a set of rules $\sys$, we seek to construct a measure $\ww$ such that for all \pbpostrong{} steps $\stepRule{G_L}{G_R}{\rho}$ generated by a rule $\rho \in \sys$, $\ww(G_L) > \ww(G_R)$. Then $\ww$ is a decreasing measure for the rewrite relation generated by $\sys$, such that $\sys$ is terminating.
We construct such a measure $\ww$ by weighing objects as follows.

\begin{defi}[Weight Functions]
	Given a set of objects $\TT$, \emph{weight function} $\ww : \TT \to \mathbb{N}$,  and class of morphisms $\ARC{\CC}$, we define the \emph{tiling weight function}
	\[
	\wwTAR{\TT}{\AR}(X) \qquad = \qquad \sum_{t \in \ARset{\TT}{X}} \ww(\dom(t))
	\]
	for objects $X \in \obj{\CC}$. In this context, we refer to the objects of $\TT$ as \emph{tiles}. 
\end{defi}

\begin{asm}
	We assume that $\TT$ is finite and $\CC$ is $\AR$-locally finite, such that $\wwTAR{\TT}{\AR}$ is well defined.
\end{asm}

\begin{exa}
	\label{exa:weight:functions}
	Let $\CC = \FinGraph$ with singleton label set $\labels$, and $G$ an arbitrary graph. Some basic examples of tile sets and parameters are as follows.
	\begin{itemize}
		\item Let $\nodeGraph$ represent the graph consisting of a single node. If $\TT = \{ \nodeGraph \}$, $\ww(\nodeGraph) = 1$, and $\ARC{\CC} \in \{ \HomC{\CC}, \MonoC{\CC}, \RegC{\CC}\}$, then $\wwTAR{\TT}{\AR}(G) = \card{V_G}$.
		\item Let $\edgeGraph$ represent the graph consisting of a single edge with distinct endpoints. If $\TT = \{ \edgeGraph \}$, $\ww(\edgeGraph) = 1$ and $\ARC{\CC} = \HomC{\CC}$, then $\wwTAR{\TT}{\AR}(G) = \card{E_G}$. If instead $\ARC{\CC} = \MonoC{\CC}$, then $\wwTAR{\TT}{\AR}(G)$ counts the number of subgraph occurrences isomorphic to $\edgeGraph$ in $G$ (loops are not counted). (See also Example~\ref{example:edge:graph:decreasing:system}.)
		\item If $\TT = \{ \nodeGraph, \edgeGraph \}$, $\ww(\nodeGraph) = 2$, $\ww(\edgeGraph) = 1$ and $\ARC{\CC} = \HomC{\CC}$, then $\wwTAR{\TT}{\AR}(G) = 2 \cdot \card{V_G} + \card{E_G}$. 
	\end{itemize}
\end{exa}

Our goal is to use $\wwTAR{\TT}{\AR}(\cdot)$ as a decreasing measure.
This gives rise to two main challenges: finding a suitable $\TT$ (if it exists), and determining whether $\wwTAR{\TT}{\AR}(\cdot)$ is decreasing. In this paper, we focus exclusively on the second problem, and show that the matter can be decided through a finite rule analysis. 

Certain assumptions on $\ARC{\CC}$ will be needed. To prevent clutter and to help intuition, we state them now, valid for the remainder of this paper. In the individual proofs, we clarify which assumptions on $\ARC{\CC}$ are used.

\begin{asm}
	\label{assumption:ar}
	We assume the following about $\ARC{\CC}$:
	\begin{itemize}
		\item $\RegC{\CC} \subseteq \AR(\CC)$; and
		\item $\ARC{\CC}$ is stable under:
		\begin{itemize}
			\item pullback;
			\item composition ($g,f \in \ARC{\CC} \implies g \circ f \in \ARC{\CC}$); and
			\item decomposition ($g \circ f \in \ARC{\CC}  \implies f \in \ARC{\CC}$).\footnote{More precisely, the notion used here is stability under decomposition for a right class $\MM$ of morphisms of some factorization system $(\mathcal{E}, \MM)$.  Stability under decomposition for the left class would mean ${{g \circ f} \in \mathcal{E}} \implies {g \in \mathcal{E}}$. This would hold, for instance, for $\mathcal{E}$ the class of epimorphisms. We will only use stability under decomposition in the sense defined here, and so we do not make the distinction in the definition.}
		\end{itemize}
	\end{itemize}
\end{asm}

Note that $\IsoC{\CC} \subseteq \ARC{\CC}$, because $\IsoC{\CC} \subseteq \RegC{\CC}$ (see, e.g., \cite[Proposition~4.29]{overbeek2024phdthesis}).

\begin{prop}
	\label{prop:assumption:regular:monos:stable:class:quasitopos}
	In any category, the classes $\HomC{\CC}$ and $\MonoC{\CC}$ satisfy Assumption~\ref{assumption:ar}. Likewise for $\RegC{\CC}$ if $\CC$ is a quasitopos.
\end{prop}
\begin{proof}
	The statement is vacuously true in any category for $\HomC{\CC}$.
	
	The properties for $\MonoC{\CC}$ are well known. Some references are \cite[Proposition~3.16]{awodey2010category} for $\RegC{\CC} \subseteq \MonoC{\CC}$, \cite[Proposition~2.5.2]{borceux1994handbook1} for pullback stability, and \cite[Exercise 2.8.4(a,b)]{awodey2010category} for stability under composition and decomposition.
	
	Class $\RegC{\CC}$ is stable under pullback in any category~\cite[Proposition~4.3.8(2)]{borceux1994handbook1}. See the provided summary (Proposition~\ref{prop:quasitopos:properties:summary}) for the remaining properties if $\CC$ is a quasitopos.
\end{proof}

Now suppose that a rule $\rho$ generates a rewrite step diagram. This defines a factorization $t_R = {R \stackrel{w}{\to} G_R \stackrel{w'}{\to} R'}$  (Remark~\ref{remark:pbpoplus:diagrams}). Any tiling of $G_R$ can be partitioned into two using the following definition.

\newcommand{\isosub}[2]{{#1}_{\cong}^{#2}}
\newcommand{\nisosub}[2]{{#1}_{\not\cong}^{#2}}

\begin{defi}
	\label{def:iso:partitions}
	For arrows $f : A \to B$ and sets $S$ of arrows with codomain $B$ we define the partitioning
	$S = \isosub{S}{f} \uplus \nisosub{S}{f}$ where $\isosub{S}{f} = \{g \in S \mid \pb{f}{g} \in \IsoC{\CC} \}$ and
	$\nisosub{S}{f} = \{ g \in S \mid \pb{f}{g} \not\in \IsoC{\CC} \}$.
\end{defi}

Intuitively,
\[ 
{\isosub{\ARset{\TT}{G_R}}{w}} = {\{ t \in \ARset{\TT}{G_R} \mid \pb{w}{t} \in \IsoC{\CC} \}}
\]
contains all tilings that lie isomorphically in the pattern $w(R)$, and 
\[
{\nisosub{\ARset{\TT}{G_R}}{w}} = {\{ t \in \ARset{\TT}{G_R} \mid \pb{w}{t} \not\in \IsoC{\CC} \}}
\]
the remaining tilings, which overlap partially or fully with the context. The remainder of this section is structured as follows.

We will start by centrally identifying some key assumptions and properties that we need in order to reason on the level of the rule (Section \ref{sec:key:requirements}). 

We then prove that there exists a domain-preserving bijection between $\isosub{\ARset{\TT}{G_R}}{w}$ and $\isosub{\ARset{\TT}{R'}}{t_R}$, allowing us to determine $\ww(\isosub{\ARset{\TT}{G_R}}{w})$ on the level of the rule (Section~\ref{sec:bijection:pattern:G_R:and:R'}).

Determining $\ww(\nisosub{\ARset{\TT}{G_R}}{w})$ on the level of the rule is in general impossible, because usually $G_R$ can have an arbitrary size. Instead, we give precise conditions, formulated on the level of the rule, that ensure that
there exists a domain-preserving injection $\xi : \nisosub{\ARset{\TT}{G_R}}{w} \mono \ARset{\TT}{G_L}$ across the rewrite step diagram, such that $\ww(\xi \circ \nisosub{\ARset{\TT}{G_R}}{w}) = \ww(\nisosub{\ARset{\TT}{G_R}}{w})$ (Section~\ref{sec:sliding:context}). Such injections often exist in the usual categories of interest, in which the context of $G_R$ is roughly inherited from the left.

\newcommand{\tilingL}{\Delta}

The two results are then combined as follows.
If we additionally find a tiling $\tilingL \subseteq \ARset{\TT}{L}$ such that for the given match $m : L \to G_L$, 
\begin{enumerate}
	\item ${m \circ \tilingL} \; \subseteq \; {\ARset{\TT}{G_L}}$;
	\item $\ww(m \circ \tilingL) \; > \; \ww(\isosub{\ARset{\TT}{G_R}}{w})$; and
	\item ${(m \circ \tilingL)} \cap (\xi \circ \nisosub{\ARset{\TT}{G_R}}{w}) \; = \; \emptyset$;
\end{enumerate}
then
\begin{align*}
	\wwTAR{\TT}{\AR}(G_L) &\geq \ww(m \circ \tilingL) + \ww(\xi \circ \nisosub{\ARset{\TT}{G_R}}{w}) \\
	&> \ww(\isosub{\ARset{\TT}{G_R}}{w}) + \ww(\nisosub{\ARset{\TT}{G_R}}{w}) \\
	&= \wwTAR{\TT}{\AR}(G_R)
\end{align*}
and we will have successfully proven that $\wwTAR{\TT}{\AR}(\cdot)$ is a decreasing measure. This is the main result of this section (Section~\ref{sec:termination:theorem}).

\subsection{Relating Rule and Step}
\label{sec:key:requirements}

\newcommand{\pbpoAdhesive}{\pbpostrong{}-Adhesive}
\newcommand{\pbpoadhesive}{\pbpostrong{}-adhesive}
\newcommand{\pbpoadhesivity}{\pbpostrong{}-adhesivity}

In order to reason about steps on the level of rules, the following variant of adhesivity is needed. It does not yet occur in the literature.
\medskip

\begin{defi}[\pbpoAdhesive]
	A pushout square ${r' \circ t_K} = {t_R \circ r}$ is \emph{\pbpoadhesive} if, whenever it lies at the bottom of a commutative cube 
	
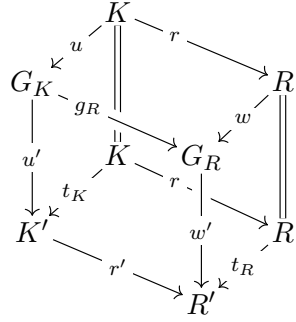
\begin{center}
		\begin{tikzcd}[column sep=16, row sep=16,ampersand replacement=\&,nodes={rectangle,inner sep=0.5mm,outer sep=0}]
			\& K \arrow[ld, "u" description] \arrow[rrd, "r" description, pos=0.3] \arrow[dd, equals] \&                                  \&                                                               \\
			G_K \arrow[dd, "u'" description] \&                                                                                            \&                                  \& R \arrow[ld, "w" description] \arrow[dd, equals] \\
			\& K \arrow[ld, "t_K" description] \arrow[rrd, "r" description, pos=0.3]                               \& G_R \arrow[from=llu, "g_R" description, crossing over, pos=0.25] \&                                                               \\
			K' \arrow[rrd, "r'" description]                               \&                                                                                            \&                                  \& R \arrow[ld, "t_R" description]                               \\
			\&                                                                                            \& R'  \arrow[from=uu, "w'" description, crossing over]                               \&                                                              
		\end{tikzcd}
	\end{center}
	where the top face is a pushout and the back faces are pullbacks, we have that the front faces are pullbacks.
\end{defi}

\begin{cor}
	\label{cor:rm:adhesive:pbpo:adhesive}
	If $\CC$ is rm-adhesive, pushouts $r' \circ t_K = {t_R \circ r}$ with $t_K \in \RegC{\CC}$ are \pbpoadhesive. \qed
\end{cor}

\begin{rem}
	\label{remark:pbpo:adhesive:question}
	Not all quasitoposes are \pbpoadhesive{}: the counterexample by Johnstone et al.~\cite[Fig.\ 1]{johnstone2007quasitoposes}, which shows that the category of simple graphs is not rm-adhesive, is also a counterexample for \pbpoadhesivity{}. An open question is whether there are categories (of interest to the graph transformation community) that are \pbpoadhesive{}, but not rm-adhesive.
\end{rem}

The following equalities will prove crucial. Recall Notation~\ref{notation:nonstandard:notation}.

\begin{lem}
	\label{lemma:VK:step}
	Assume $\CC$ has pullbacks. Let a rewrite step for a \pbpostrong{} rule $\rho$ be given. If square $r' \circ t_K = t_R \circ r$ is \pbpoadhesive, then for any $\rho$-rewrite step and any $t : T \to G_R$
	\begin{enumerate}
		\item\label{lemma:VK:step:along:t} $\pb{g_R}{t} = \pb{r'}{w' \circ t}$;
		\item\label{lemma:VK:step:along:g_R} $u' \circ \pb{t}{g_R} = \pb{w' \circ t}{r'}$; 
		\item\label{lemma:VK:step:along:t_R} $\pb{w}{t} = \pb{t_R}{w' \circ t}$; and
		\item\label{lemma:VK:step:AR} $\pb{t}{w} =  \pb{w' \circ t}{t_R}$.
	\end{enumerate}
\end{lem}

\mysidepicture{4.8cm}{0.5ex}{15ex}{%
	\adjustbox{scale=1}{%
		\begin{tikzcd}[column sep=18, row sep=15,ampersand replacement=\&,nodes={rectangle,inner sep=0.5mm,outer sep=0}]
			\& 
			\&                                     \&                                                                           \\
			S \arrow[rrd, "\pb{g_R}{t}" description, pos=0.3] \arrow[dd, "\pb{t}{g_R}" description] \&                                                                                                  \&                                     \& Y \arrow[ld, "\pb{w}{t}" description] \arrow[dd, "\pb{t}{w}" description] \\
			\& K \arrow[ld, "u" description] \arrow[rrd, "r" description, pos=0.3] \arrow[dd, equals] \& T \&                                                                           \\
			G_K \arrow[dd, "u'" description]                 \&                                                                                                  \&                                     \& R \arrow[ld, "w" description] \arrow[dd, equals]             \\
			\& K \arrow[ld, "t_K" description] \arrow[rrd, "r" description, pos=0.3]                               \& G_R \arrow[from=llu, "g_R" description, crossing over, pos=0.25] \arrow[from=uu, "t" description, crossing over]   \&                                                                           \\
			K' \arrow[rrd, "r'" description, pos=0.25]                                               \&                                                                                                  \&                                     \& R \arrow[ld, "t_R" description]                                           \\
			\&                                                                                                  \& R' \arrow[from=uu, "w'" description, crossing over]                                  \&                                                                          
		\end{tikzcd}
	}
}{%
	\begin{proof}
		In the diagram on the right,
		the bottom face of the bottom cube is \pbpoadhesive{} by assumption, its top face is a pushout, and its back faces are pullbacks in any category~\cite[Lemma 15]{overbeek2023quasitoposes}. Hence its front faces are pullbacks by \pbpoadhesivity. Then all claims follow by composing pullback squares, using the pullback lemma. 
	\end{proof}
}

\begin{rem}
	Because every $t \in \ARset{T}{G_R}$ defines an arrow ${w' \circ t} \in \homset{T}{R'}$, we can overapproximate $\ARset{T}{G_R}$ using $\homset{T}{R'}$.
	The equalities of Lemma~\ref{lemma:VK:step} will then be used as follows.
	\begin{enumerate}
		\item We will slide morphisms $t \in \nisosub{\ARset{\TT}{G_R}}{w}$ to the left. If $\pb{g_R}{t}$ is invertible, then $g_L \circ \pb{t}{g_R} \circ \rightinv{\pb{g_R}{t}} : T \to G_L$ is an arrow towards the left. Lemma~\ref{lemma:VK:step}.\ref{lemma:VK:step:along:t} implies that invertibility of $\pb{g_R}{t}$ can be verified on the level of the rule.
		\item Although we cannot deduce $\pb{t}{g_R}$, Lemma~\ref{lemma:VK:step}.\ref{lemma:VK:step:along:g_R} implies that we can at least deduce how it is mapped into $K'$.
		\item Lemma~\ref{lemma:VK:step}.\ref{lemma:VK:step:along:t_R} implies that it suffices to restrict the overapproximation of $\nisosub{\ARset{\TT}{G_R}}{w}$ to $\nisosub{\homset{\TT}{R'}}{t_R}$.
		\item If $t \in \ARC{\CC}$, then $\pb{t}{w} \in \ARC{\CC}$ by the pullback stability assumption. Thus, Lemma~\ref{lemma:VK:step}.\ref{lemma:VK:step:AR} implies that it suffices to restrict the overapproximation even further to $\{ f \in \nisosub{\homset{\TT}{R'}}{t_R} \mid \pb{f}{t_R} \in \ARC{\CC} \}$.
	\end{enumerate}
\end{rem}

\subsection{Determining \texorpdfstring{$\ww(\isosub{\ARset{\TT}{G_R}}{w})$}{the Weight of the Pattern}}
\label{sec:bijection:pattern:G_R:and:R'}

In this section we show that the weight of $\isosub{\ARset{\TT}{G_R}}{w}$ can be determined under minimal assumptions. 

\newcommand{\pbeq}[2]{\langle #1 \mid #2 \rangle^{=}}

\begin{prop}
	\label{prop:iso:unique:pb:morphism}
	If $f$ and $g$ have a pullback of the form on the left of:
	\begin{center}
		\begin{tikzcd}[column sep=40, ampersand replacement=\&]
			C \arrow[d, "g" description] \& A' \arrow[l, "f'" description] \arrow[d, "i" description, two heads, tail] \arrow[ld, "\mathrm{PB}", phantom] \&  \& C \arrow[d, "g" description] \& A \arrow[l, "f'i^{{-1}}" description] \arrow[d, equals] \arrow[ld, "\mathrm{PB}", phantom] \\
			B                            \& A \arrow[l, "f" description]                                                                                  \&  \& B                            \& A \arrow[l, "f" description]                                                                    
		\end{tikzcd}
	\end{center}
	where $i \in \IsoC{\CC}$, then the square on the right is also a pullback. Moreover, $f'i^{{-1}}$ is a unique solution for $h$ in $f = g \circ h$.
\end{prop}
\begin{proof}
	First, ${gf' = fi} \implies {gf'i^{{-1}} = fii^{{-1}} \implies gf'i^{{-1}} = f}$. So the right square commutes.
	
	For the pullback property, let $gx = fy$ commute for a span $C \stackrel{x}{\leftto} X \stackrel{y}{\to} A$. Then using the pullback property of the left square, there exists a unique $z$ such that
	\begin{center}
		\begin{tikzcd}[column sep=40]
			C \arrow[d, "g" description] & A \arrow[l, "f' i^{-1}" description] \arrow[d, equals] \arrow[ld, "=", phantom] & A' \arrow[l, "i" description, two heads, tail] \arrow[ll, "f'" description, bend right] \arrow[ld, "i" description, two heads, tail] & X \arrow[l, "z" description, dotted] \arrow[lll, "x" description, bend right=49] \arrow[lld, "y" description] \\
			B                            & A \arrow[l, "f" description]                                                                 &                                                                                                                                     &                                                                                                              
		\end{tikzcd}
	\end{center}
	commutes. Then $iz$ is a solution for $w$ in $f'i^{{-1}} w = x$ and $1_A w = y$, and it is in fact a unique solution by the second equation. So the commuting square is indeed a pullback.
	
	Finally, that $f'i^{{-1}}$ is a unique solution for $h$ in $f = g \circ h$ follows easily from the established pullback property~\cite[Proposition 13]{overbeek2023quasitoposes}.
\end{proof}

\begin{defi}
	\label{def:pbeq}
	For cospans $A \stackrel{f}{\to} B \stackrel{g}{\leftto} C$ with $\pb{g}{f} \in \IsoC{\CC}$ iso, we let $\pbeq{f}{g}$ denote the unique $h$ satisfying $f = g \circ h$ (well-defined by Proposition~\ref{prop:iso:unique:pb:morphism}).
\end{defi}

\begin{lem}
	\label{lemma:bijection:iso:pattern}
	Let the pullback 
	\begin{center}
		\begin{tikzcd}[ampersand replacement=\&,row sep=26, column sep=60]
			R \arrow[d, "t_R" description, AR] \& R \arrow[d, "w" description] \arrow[l, equals] \arrow[ld, "\mathrm{PB}", phantom] \\
			R'                             \& G_R \arrow[l, "w'" description]                                                               
		\end{tikzcd}
	\end{center}
	be given with
	$t_R \in \ARC{\CC}$. Let $\CC$ be $\AR$-locally finite and $\TT$ a set of objects.
	Then $\chi : \isosub{\ARset{\TT}{G_R}}{w} \to  \isosub{\ARset{\TT}{R'}}{t_R}$ defined by $\chi(t) = w' \circ t$ is a domain-preserving bijection.
\end{lem}
\begin{proof}
	That $\chi$ is domain-preserving is immediate from the definition.
	
	We show that $\chi$ is well-typed. Let $t \in \isosub{\ARset{\TT}{G_R}}{w}$ and let $\dom(t) = T \in \TT$. We have
	\begin{equation}
		\label{eq:bijection:iso:pattern:first:diagram}
		\begin{tikzcd}[column sep=60, row sep=26]
			R \arrow[d, "t_R" description, AR] & R \arrow[d, "w" description] \arrow[l, equals] \arrow[ld, "\mathrm{PB}", phantom] & T \arrow[l, "\pbeq{t}{w}" description, AR] \arrow[d, equals] \arrow[ld, "\mathrm{PB}", phantom] \\
			R'                             & G_R \arrow[l, "w'" description]                                                                & T \arrow[l, "t" description, AR]                                                                            
		\end{tikzcd}
	\end{equation}
	using Proposition~\ref{prop:iso:unique:pb:morphism}. By the pullback lemma, the outer square is a pullback, and so $\pb{t_R}{w' \circ t} \in \IsoC{\CC}$. Moreover,
	${w' \circ t} \in \ARset{T}{R'}$, because ${w' \circ t} = {t_R \circ \pbeq{t}{w}}$, and ${t_R \circ \pbeq{t}{w}} \in \ARset{T}{R'}$ by stability under composition (Assumption~\ref{assumption:ar}).
	Thus, the image of $\chi$ lies in $\isosub{\ARset{\TT}{R'}}{t_R}$, and so $\isosub{\ARset{\TT}{R'}}{t_R}$ is a valid codomain.
	
	For injectivity of $\chi$, assume $w' \circ s = w' \circ t$. We then have the diagram
	\begin{equation}
		\begin{tikzcd}[column sep=60, row sep=26]
			R \arrow[d, "t_R" description, AR] & R \arrow[d, "w" description] \arrow[l, equals] \arrow[ld, "\mathrm{PB}", phantom] & T \arrow[l, "\pbeq{s}{w}" description, AR] \arrow[d, equals] \arrow[ld, "\mathrm{PB}", phantom] & T \arrow[lll, "\pbeq{t}{w}" description, bend right=15, AR] \arrow[ld, equals] \arrow[l, "!x" description, dotted] \\
			R'                             & G_R \arrow[l, "w'" description]                                                                & T \arrow[l, "s" description, AR]                                                                             &                                                                                                                         
		\end{tikzcd}
	\end{equation}
	where the outer diagram commutes by ${w' \circ s} = {w' \circ t}$ and Diagram~\eqref{eq:bijection:iso:pattern:first:diagram}. Hence there exists a unique $x : T \to T$ such that $\pbeq{t}{w} = {\id{R} \circ \pbeq{s}{w} \circ x}$ and $\id{T} = {\id{T} \circ x}$. Hence $x = \id{T}$, and by canceling identities, $\pbeq{t}{w} = \pbeq{s}{w}$. Thus $s = {w \circ \pbeq{s}{w}} = {w \circ \pbeq{t}{w}} = t$.
	
	For surjectivity of $\chi$, let $f \in \isosub{\ARset{\TT}{R'}}{t_R}$. We have the diagram
	\begin{equation}
		\begin{tikzcd}[column sep=80, row sep=25]
			R \arrow[d, "t_R" description, AR] & R \arrow[d, "w" description] \arrow[l, equals] \arrow[ld, "\mathrm{PB}", phantom] & T \arrow[ll, "\pbeq{f}{t_R}" description, AR, bend right=15] \arrow[d, equals] \arrow[l, "\pbeq{f}{t_R}" description, AR] \arrow[ld, "=", phantom] \\
			R'                                     & G_R \arrow[l, "w'" description]                                                                & T \arrow[ll, "f" description, AR, bend left=15] \arrow[l, "w \circ \pbeq{f}{t_R}" description, AR]                                                                     
		\end{tikzcd}
	\end{equation}
	where the outer square is a pullback. We have ${t_R \circ \pbeq{f}{t_R}} \in \ARC{\CC}$ by stability under composition, and hence $w \circ \pbeq{f}{t_R} \in \ARC{\CC}$ by stability under decomposition, using $t_R = {w' \circ w}$.
	The right square is moreover a pullback by the pullback lemma. Hence ${w \circ \pbeq{f}{t_R}} \in \isosub{\ARset{\TT}{G_R}}{w}$. It then follows that $f = {w' \circ w \circ \pbeq{f}{t_R}} = \chi(w \circ \pbeq{f}{t_R})$ lies in the image of $\chi$.
\end{proof}

\begin{cor}
	\label{corollary:determining:isosub:weight}
	If the conditions of Lemma~\ref{lemma:bijection:iso:pattern} are met, then 
	\[ 
	\pushQED{\qed}
	\ww(\isosub{\ARset{\TT}{G_R}}{w}) =  \ww(\isosub{\ARset{\TT}{R'}}{t_R}) \text{.} \qedhere
	\popQED
	\] 
\end{cor}

\subsection{Sliding Tiles Injectively}
\label{sec:sliding:context}

In this section we establish conditions for the existence of a domain-preserving injection 
$\xi : \nisosub{\ARset{\TT}{G_R}}{w} \mono \ARset{\TT}{G_L}$.
Intuitively, one can think of $\xi$ as sliding tiles from right to left across the rewrite step diagram.

If $l' \in \RegC{\CC}$, then $\xi$ will be seen to exist rather straightforwardly. However, in general it suffices to require more weakly that $l'$ preserves any tiles to be slid (and distinctly so). Definitions~\ref{def:morphism:preserves:factorization} and \ref{def:monic:for} help capture such a weaker requirement.
With these definitions, $\xi$ can be shown to exist even for non-trivial rules with non-monic $l'$. 

\begin{defi}\label{def:morphism:preserves:factorization}
	A morphism $g : B \to C$ \emph{preserves an $\AR$-factorization} of $f : A \to B$ if there exists a factorization $f = f'' \circ f'$ of $f$ satisfying
	\begin{itemize}
		\item $f'' \in \ARC{\CC}$; and
		\item $g \circ f'' \in \ARC{\CC}$.
	\end{itemize}
\end{defi}

\begin{lem}
	\label{lemma:monic:on:image}
	Assume $\CC$ has pullbacks. Let the diagram 
	\begin{center}
		\begin{tikzcd}[ampersand replacement=\&,column sep=30,row sep=20,nodes={rectangle,inner sep=1mm,outer sep=0}]
			A \arrow[d, "g" description] \& B \arrow[d, "g'" description] \arrow[l, "f'" description] \arrow[ld, "\mathrm{PB}", phantom] \& X \arrow[l, "x" description, AR] \\
			C                            \& D \arrow[l, "f" description]                                                                 \&                             
		\end{tikzcd}
	\end{center}
	be given, with $x \in \ARC{\CC}$. If $f$ preserves an $\AR$-factorization of $g' \circ x$, then $f' \circ x \in \ARC{\CC}$.
\end{lem}
\begin{proof}
	By assumption, there exists an $\AR$-factorization $k' \circ k$ of $g' \circ x$ with $k', (f \circ k') \in \ARC{\CC}$. Construct the pullback of $B \stackrel{g'}{\to} D \stackrel{k'}{\leftARar} K$ as depicted in diagram
	\begin{center}
		\begin{tikzcd}
			A \arrow[d, "g" description] & B \arrow[d, "g'" description] \arrow[l, "f'" description] \arrow[ld, "\mathrm{PB}", phantom] & E \arrow[l, "k''" description, AR] \arrow[d, "g''" description] \arrow[ld, "\mathrm{PB}", phantom] & X \arrow[ll, "x" description, AR, bend right] \arrow[ld, "k" description, bend left] \arrow[l, "y" description, dotted, AR] \\
			C                            & D \arrow[l, "f" description]                                                                 & K \arrow[l, "k'" description, AR] \arrow[ll, "f \circ k'" description, AR, bend left]        &                                                                                                                                 
		\end{tikzcd}
	\end{center}
	where $k'' \in \ARC{\CC}$ by pullback stability.
	By the pullback property, there exists a morphism $y : X \to E$ such that $x = k'' \circ y$ and $k = g'' \circ y$. By $x \in \ARC{\CC}$ and the decomposition property, $y \in \ARC{\CC}$. By the pullback lemma, the two squares compose to form a larger pullback square so that $f' \circ k'' \in \ARC{\CC}$ by pullback stability. Finally, by stability under composition, $f' \circ x = (f' \circ k'') \circ y  \in \ARC{\CC}$. 
\end{proof}

\newcommand{\EE}{\mathcal{E}}

\begin{prop}
	\label{prop:preserving:a:factorization:and:factorization:systems}
	If $\CC$ has (up to isomorphism) unique $(\EE, \AR)$-factorizations (for some class of morphisms $\EE$), then the following statements are equivalent:
	\begin{enumerate}
		\item $g : B \to C$ preserves an $\AR$-factorization of $f:  A \to B$.
		\item $g \circ a \in \ARC{\CC}$ for $f = {a \circ e}$ an $(\EE, \AR)$-factorization of $f$, i.e., $a \in \AR$ and $e \in \EE$.
	\end{enumerate}
\end{prop}
\begin{proof}
	$1 \implies 2$: let $f = {f'' \circ f'}$ be given with $f'', (g \circ f'') \in \ARC{\CC}$. $(\EE, \AR)$-factorize both $f'$ and $f$ as in the following diagram:
	\begin{center}
		\begin{tikzcd}[column sep=40]
			A \arrow[rd, dashed, "e'" description] \arrow[rr, "f'" description] \arrow[rdd, dashed, bend right, "e" description] &                                      & D \arrow[r, "f''" description, AR] \arrow[rr, "g \circ f''" description, bend left, AR] & B \arrow[r, "g" description] \arrow[from=lll, "f" description, bend left, crossing over] & C \\
			& Y \arrow[ru, AR, "a'" description]                         &                                                                                 &                              &   \\
			& X \arrow[rruu, bend right, AR, "a" description] \arrow[u, AR, "i" description] &                                                                                 &                              &  
		\end{tikzcd}
	\end{center}
	The dashed arrows are in $\EE$. Then $(e', f'' \circ a')$ is an $(\EE, \AR)$-factorization of $f$ (using that $\ARC{\CC}$ is stable under composition), so that by essential uniqueness of such factorizations, there exists an $(i : X \to Y) \in \IsoC{\CC} \subseteq \ARC{\CC}$ with $a = {(f'' \circ  a') \circ i}$. Then $g \circ a \in \ARC{\CC}$ by $g \circ a = {(g \circ f'') \circ a' \circ i}$ and the assumption that $\ARC{\CC}$ is stable under composition.
	
	$2 \implies 1$: immediate by definition.
\end{proof}

\begin{rem}
	In any category $\CC$, morphisms $f: A \to B$ trivially admit a unique (identity, morphism)-factorization, namely $(\id{\dom(f)}, f)$. So if $\ARC{\CC} = \HomC{\CC}$, verifying whether $g: B \to C$ preserves an $\AR$-factorization of $f$ reduces to verifying whether $g \circ f \in \HomC{\CC}$, which is immediate. Moreover, a quasitopos admits both (essentially) unique (epi, regular mono)-factorizations and (essentially) unique (regular epi, mono)-factorizations (Proposition~\ref{prop:quasitopos:properties:summary}). Thus, if $\CC$ is a quasitopos and $\ARC{\CC} = \RegC{\CC}$ or $\ARC{\CC} = \MonoC{\CC}$, verifying whether any morphism $g: B \to C$ preserves an $\AR$-factorization of $f$ reduces to verifying a property for the respective factorizations.
\end{rem}

\begin{defi}[Monic For]
	\label{def:monic:for}
	Morphism $h : B \to C$ is \emph{monic for} morphisms $f,g: A \to B$ if $h \circ f = h \circ g$ implies $f = g$.
\end{defi}

\begin{lem}
	\label{lemma:monic:for:pb:stable}
	Let the diagram 
	\begin{center}
		\begin{tikzcd}[ampersand replacement=\&,column sep=30,row sep=25,nodes={rectangle,inner sep=1mm,outer sep=0}]
			A \arrow[d, "g" description] \& B \arrow[d, "g'" description] \arrow[l, "f'" description] \arrow[ld, "\mathrm{PB}", phantom] \& X \arrow[l, "y" description, shift left=2] \arrow[l, "x" description, shift right=2] \\
			C                            \& D \arrow[l, "f" description]                                                                 \&                                                                                   
		\end{tikzcd}
	\end{center}
	be given. If $f$ is monic for $g' \circ x$ and $g' \circ y$, then $f'$ is monic for $x$ and $y$.
\end{lem}
\begin{proof}
	Assume ($\dagger_1$) ${f' \circ x} = {f' \circ y}$. We must show $x = y$. We have:
	\begin{align*}
		{f' \circ x} &= {f' \circ y} \\
		{g \circ f' \circ x} &= {g \circ f' \circ y} \\
		{f \circ g' \circ x} &= {f \circ g' \circ y} \\
		{g' \circ x} &= {g' \circ y}  \tag{$\dagger_2$}
	\end{align*}
	using ${g \circ f'} = {f \circ g'}$ and the assumption that $f$ is monic for $g' \circ x$ and $g' \circ y$. By ${g \circ (f' \circ x)} = {f \circ (g' \circ x)}$ and the pullback property, there exists a unique $z$ such that in diagram
	\begin{center}
		\begin{tikzcd}
			A \arrow[d, "g" description] & B \arrow[d, "g'" description] \arrow[l, "f'" description] \arrow[ld, "\mathrm{PB}", phantom] & X \arrow[ll, "f' \circ x" description, bend right] \arrow[ld, "g' \circ x" description] \arrow[l, "!z" description, dotted] \\
			C                            & D \arrow[l, "f" description]                                                                 &                                                                                                                            
		\end{tikzcd}
	\end{center}
	$f' \circ x = f' \circ z$ and $g' \circ x = g' \circ z$. Solution $z = x$ is trivial, and $z = y$ is a solution by ($\dagger_1$) and ($\dagger_2$). Hence $x = y$.
\end{proof}

The morphism $g_R : G_K \to G_R$ of the rewrite step may identify elements. So for the injection $\xi$ from right to left to exist, we must be able to go into the inverse direction, without identifying tiles. To this end, the following lemma will prove useful.

\begin{lem}
	\label{lemma:split:epi:squares:injective}
	If epimorphisms $e$ and $e'$ in diagram
	\begin{center}
		\begin{tikzcd}[ampersand replacement=\&]
			A' \arrow[r, "f'" description] \arrow[d, "e" description, two heads] \arrow[rd, "=", phantom] \& B' \arrow[d, "h" description] \& C' \arrow[l, "g'" description] \arrow[d, "e'" description, two heads] \arrow[ld, "=", phantom] \\
			A \arrow[r, "f" description]                                                                  \& B                             \& C \arrow[l, "g" description]                                                                  
		\end{tikzcd}
	\end{center}
	are split, then for right inverses $\rightinv{e}$ and $\rightinv{e'}$, $f'\rightinv{e} = g'\rightinv{e'} \implies f = g$.
\end{lem}
\begin{proof}
	$f'\rightinv{e} = g'\rightinv{e'}  
	\hspace{-0.22mm} \implies \hspace{-0.22mm}
	hf'\rightinv{e} = hg'\rightinv{e'} 
	\hspace{-0.22mm} \implies \hspace{-0.22mm}
	fe\rightinv{e} = ge'\rightinv{e'} 
	\hspace{-0.22mm} \implies  \hspace{-0.22mm}
	f = g$.
\end{proof}

We are now ready for the main theorem of this subsection. We recommend keeping the diagram for Lemma~\ref{lemma:VK:step} alongside it.

\begin{thm}
	\label{thm:xi:is:injection}
	Assume $\CC$ has pullbacks.
	Let a rewrite rule $\rho$ be given with square $r' \circ t_K = t_R \circ r$ \pbpoadhesive{}. Fix a class $\ARC{\CC}$ and a set of objects $\TT$. 
	Define
	\[ 
	\Phi = \{ f \in  \nisosub{\homset{\TT}{R'}}{t_R} \mid  \pb{f}{t_R} \in \ARC{\CC} \} \text{.}
	\]
	If
	\begin{enumerate}
		\item\label{ass:phi:split:epi} for all $f \in \Phi$,   $\pb{r'}{f}$  is a split epimorphism; and
		\item for some right inverse choice function $\rightinv{(\cdot)}$ and for all $f,g \in \Phi$, 
		\begin{enumerate}
			\item\label{ass:phi:l':factorization} $l'$ preserves an $\AR$-factorization of $\pb{f}{r'} \circ \rightinv{\pb{r'}{f}}$; and
			\item\label{ass:phi:l':monic} $l'$ is monic for $\pb{f}{r'} \circ \rightinv{\pb{r'}{f}}$ and $\pb{g}{r'} \circ \rightinv{\pb{r'}{g}}$;
		\end{enumerate}
	\end{enumerate}
	then for any rewrite step diagram induced by $\rho$, the function 
	\[
	\xi(t) = {g_L \circ \pb{t}{g_R} \circ \rightinv{\pb{g_R}{t}}}
	\]
	defines an injection $\nisosub{\ARset{\TT}{G_R}}{w} \mono \ARset{\TT}{G_L}$.
\end{thm}
\begin{proof}
	We first argue that the use of $\rightinv{(\cdot)}$ in $\xi(t)$ is well-defined.
	Because $t \in \ARC{\CC}$, we have $\pb{t}{w} \in \ARC{\CC}$ by pullback stability. And $\pb{t}{w} = \pb{w' \circ t}{t_R}$ by Lemma~\ref{lemma:VK:step}.\ref{lemma:VK:step:AR}. Moreover, $\pb{w}{t} \notin \IsoC{\CC}$, and $\pb{w}{t} = \pb{t_R}{w' \circ t}$ by Lemma~\ref{lemma:VK:step}.\ref{lemma:VK:step:along:t_R}. So $\pb{w' \circ t}{t_R} \in \Phi$. Then by local assumption~\ref{ass:phi:split:epi}, $\pb{r'}{w' \circ t}$ is a split epimorphism. And $\pb{r'}{w' \circ t} = \pb{g_R}{t}$ by Lemma~\ref{lemma:VK:step}.\ref{lemma:VK:step:along:t}. So $\rightinv{\pb{g_R}{t}}$ is well-defined.
	
	We next argue that $\xi(t) \in \ARset{\TT}{G_L}$.
	As established, $\pb{r'}{w' \circ t} = \pb{g_R}{t}$, and by Lemma~\ref{lemma:VK:step}.\ref{lemma:VK:step:along:g_R}, $u' \circ \pb{t}{g_R} = \pb{w' \circ t}{r'}$. So the diagram
	\begin{center}
		\begin{tikzcd}[column sep=120]
			G_L \arrow[d, "\alpha" description] & G_K \arrow[d, "u'" description] \arrow[l, "g_L" description] \arrow[ld, "\mathrm{PB}", phantom] & T \arrow[l, "{\pb{t}{g_R} \circ \rightinv{\pb{g_R}{t}}}" description, AR] \arrow[ld, "\pb{w' \circ t}{r'} \circ \rightinv{\pb{r'}{w' \circ t}}" description, bend left=10] \\
			L'                                & K' \arrow[l, "l'" description]                                                                &                                                             
		\end{tikzcd}
	\end{center}
	commutes, where the pullback square is given by the rewrite step. Moreover, $\pb{t}{g_R} \circ \rightinv{\pb{g_R}{t}} \in \ARC{\CC}$ (as indicated in the diagram) by stability under composition, using $\pb{t}{g_R} \in \ARC{\CC}$ (by pullback stability and $t \in \ARC{\CC}$) and $\rightinv{\pb{g_R}{t}} \in \RegC{\CC} \subseteq \ARC{\CC}$ (using Proposition~\ref{prop:split:epi:facts} and Assumption~\ref{assumption:ar}). By local assumption~\ref{ass:phi:l':factorization} and the commuting triangle of the diagram, $l'$ preserves an $\AR$-factorization of 
	$u' \circ \pb{t}{g_R} \circ \rightinv{\pb{g_R}{t}}$. So by Lemma~\ref{lemma:monic:on:image}, $\xi(t) \in \ARC{\CC}$ and consequently $\xi(t) \in \ARset{\TT}{G_L}$.
	
	For injectivity of $\xi$, assume $\xi(t) = \xi(s)$ for $t,s \in \nisosub{\ARset{\TT}{G_R}}{w}$. 
	By local assumption~\ref{ass:phi:l':monic}, Lemma~\ref{lemma:VK:step}.\ref{lemma:VK:step:along:t}, and Lemma~\ref{lemma:VK:step}.\ref{lemma:VK:step:along:g_R}, $l'$ is monic for
	\[
	{\pb{w' \circ t}{r'} \circ \rightinv{\pb{r'}{w' \circ t}}} = {u' \circ \pb{t}{g_R} \circ \rightinv{\pb{g_R}{t}}}
	\]
	and 
	\[
	{\pb{w' \circ s}{r'} \circ \rightinv{\pb{r'}{w' \circ s}}} = 
	{u' \circ \pb{s}{g_R} \circ \rightinv{\pb{g_R}{s}}} \text{.}
	\]
	So by Lemma~\ref{lemma:monic:for:pb:stable}, $g_L$ is monic for $\pb{t}{g_R} \circ \rightinv{\pb{g_R}{t}}$ and $\pb{s}{g_R} \circ \rightinv{\pb{g_R}{s}}$. Then because $\xi(t) = \xi(s)$, $\pb{t}{g_R} \circ \rightinv{\pb{g_R}{t}} = \pb{s}{g_R} \circ \rightinv{\pb{g_R}{s}}$. Then finally, $t = s$ by Lemma~\ref{lemma:split:epi:squares:injective}.
\end{proof}

\subsection{The Main Result}
\label{sec:termination:theorem}

We are now ready to prove the main result of this paper (Theorem~\ref{thm:termination:by:element:counting}) and its corollary (Corollary~\ref{corr:termination:by:element:counting}). We also show that in rather common settings, many technical conditions of the theorem are met automatically (Lemma \ref{lemma:sliding:does:not:collide} and Propositions~\ref{proposition:rmadhesive:quasitopos:consequences} and~\ref{prop:finite:presheaf}). We close with a complementary lemma that establishes decreasingness for deleting rules (Lemma~\ref{lemma:deleting:rules:terminate}). Examples of applications will be given in Section~\ref{section:examples}.

\begin{thm}[Decreasingness by Element Counting]
	\label{thm:termination:by:element:counting}
	Let $\sys$ and $\sys'$ be disjoint sets of \pbpostrong{} rules.
	Assume $\CC$ has pullbacks and let $\ARC{\CC}$ be a class such that $\CC$ is $\AR$-locally finite.
	Let $\TT$ be a set of objects and $\ww: \TT \to \mathbb{N}$ a weight function such that, for every $\rho \in {\sys \uplus \sys'}$, the following conditions hold:
	\begin{itemize}
		\item $\rho$'s pushout square $r' \circ t_K = t_R \circ r$  is \pbpoadhesive{}; and
		\item $t_R \in \ARC{\CC}$; and 
		\item set $\Phi_{\rho} = \{ f \in \nisosub{\homset{\TT}{R'}}{t_R} \mid  \pb{f}{t_R} \in \ARC{\CC} \}$ meets the conditions of Theorem~\ref{thm:xi:is:injection} for some right inverse choice function $\rightinv{(\cdot)}$; and 
		\item
		there exists a set $\tilingL_{\rho} \subseteq \ARset{\TT}{L}$ such that
		\begin{itemize}
			\item for all $f \in \Phi_\rho$ and $t \in \tilingL_{\rho}$,  $l' \circ \pb{f}{r'} \circ \rightinv{\pb{r'}{f}} \neq t_L \circ t$;
			\item $t_L$ is monic for all $t,t' \in \tilingL_\rho$; and
			\item 
			$\ww(\tilingL_\rho) > \ww(\isosub{\ARset{\TT}{R'}}{t_R})$ if $\rho \in \sys$ and $\ww(\tilingL_\rho) \ge \ww(\isosub{\ARset{\TT}{R'}}{t_R})$ if $\rho \in \sys'$.
		\end{itemize}
	\end{itemize}
	Then for any rewrite step with match $m \in \ARC{\CC}$, induced by a rule $\rho \in \sys \uplus \sys'$, we have 
	$\wwTAR{\TT}{\AR}(G_L) > \wwTAR{\TT}{\AR}(G_R)$ if $\rho \in \sys$ and $\wwTAR{\TT}{\AR}(G_L) \ge \wwTAR{\TT}{\AR}(G_R)$ if $\rho \in \sys'$.
\end{thm}
\begin{proof}
	Let a step induced by a $\rho \in {\sys \uplus \sys'}$ be given.
	
	By Corollary~\ref{corollary:determining:isosub:weight}, $\ww(\isosub{\ARset{\TT}{G_R}}{w}) = \ww(\isosub{\ARset{\TT}{R'}}{t_R})$.
	
	By Theorem~\ref{thm:xi:is:injection}, we obtain an injection $\xi : \nisosub{\ARset{\TT}{G_R}}{w} \mono \ARset{\TT}{G_L}$
	with $\dom(\xi(t)) = \dom(t)$, using the assumption on $\Phi_\rho$. So ${\ww(\xi \circ \nisosub{\ARset{\TT}{G_R}}{w})} = {\ww(\nisosub{\ARset{\TT}{G_R}}{w})}$.
	
	Moreover, by $m \in \ARC{\CC}$ and stability under composition, we have
	$(m \circ \tilingL_\rho) \subseteq \ARset{\TT}{G_L}$. And by $t_L$ monic for all $t, t' \in \tilingL_\rho$, we have $m$ monic for all $ t,t' \in \tilingL_\rho$, and so $\ww(m \circ \tilingL_\rho) = \ww(\tilingL_\rho)$. It remains to show that $(m \circ \tilingL_\rho)$ and 
	$(\xi \circ \nisosub{\ARset{\TT}{G_R}}{w})$ are disjoint.  If for a $t' \in \tilingL_\rho$ and $t \in \nisosub{\ARset{\TT}{G_R}}{w}$, $m \circ t' = \xi(t)$, then  ${t_L \circ t'} = {\alpha \circ m \circ t'} = {\alpha \circ \xi(t)} = {\alpha \circ g_L \circ \pb{t}{g_R} \circ \rightinv{\pb{g_R}{t}}} = {l' \circ u' \circ \pb{t}{g_R} \circ \rightinv{\pb{g_R}{t}}} = {l' \circ \pb{w' \circ t}{r'} \circ \rightinv{\pb{r'}{w' \circ t}}}$,
	using Lemma~\ref{lemma:VK:step}.(\ref{lemma:VK:step:along:t}--\ref{lemma:VK:step:along:g_R}) and ${\alpha \circ g_L} = {l' \circ u}$,
	which contradictingly implies $t' \notin \tilingL_\rho$ by the definition of $\tilingL_\rho$ and $w' \circ t \in \Phi$. Thus $\xi(t) \neq m \circ t'$.
	
	In summary,
	\begin{align*}
		\wwTAR{\TT}{\AR}(G_L) &\geq \ww(m \circ \tilingL_\rho) + \ww(\xi \circ \nisosub{\ARset{\TT}{G_R}}{w}) \\
		&= 
		\ww(\tilingL_\rho) + \ww(\nisosub{\ARset{\TT}{G_R}}{w}) \\
		& \succ
		\ww(\isosub{\ARset{\TT}{R'}}{t_R}) + \ww(\nisosub{\ARset{\TT}{G_R}}{w})
		\\
		&=
		\ww(\isosub{\ARset{\TT}{G_R}}{w}) + \ww(\nisosub{\ARset{\TT}{G_R}}{w})
		\\
		&= \wwTAR{\TT}{\AR}(G_R)
	\end{align*}
	for ${\succ} = {>}$ if $\rho \in \sys$ and ${\succ} = {\geq}$ if $\rho \in \sys'$, completing the proof.
\end{proof}

\begin{rem}
	\label{remark:matches:lower:bound:A}
	The requirement $m \in \ARC{\CC}$ puts a lower bound on what one can choose for $\ARC{\CC}$ in a termination proof. Usually two factors are relevant: the class of $t_L$, and match restrictions imposed by the setting. More precisely, let $\ClassC{X}{\CC}$ and $\ClassC{Y}{\CC}$ be classes of morphisms. If $t_L \in \ClassC{X}{\CC}$, where $\ClassC{X}{\CC}$ satisfies the decomposition property (meaning $m \in \ClassC{X}{\CC}$ by $t_L = {\alpha \circ m}$), and the setting imposes $m \in \ClassC{Y}{\CC}$, then the choice of $\ARC{\CC}$ must satisfy ${\ClassC{X}{\CC} \cap \ClassC{Y}{\CC}} \subseteq \ARC{\CC}$.
\end{rem}

From Theorem~\ref{thm:termination:by:element:counting} and Remark~\ref{remark:matches:lower:bound:A} the following is immediate.

\begin{cor}[Termination by Element Counting]
	\label{corr:termination:by:element:counting}
	Let $\sys$ and $\sys'$ be disjoint sets of \pbpostrong{} rules. Let $\ARC{\CC}$ be a class such that for all rules $\rho \in (\sys \uplus \sys')$,  $t_L(\rho) \in \ARC{\CC}$ or matching of $\rho$ is restricted to a class $X \subseteq \ARC{\CC}$.
	If the conditions  of Theorem~\ref{thm:termination:by:element:counting} are met, then $\sys$ terminates relative to $\sys'$.
	Hence $\sys \uplus \sys'$ is terminating iff $\sys'$ is. \qed
\end{cor}

\begin{rem}[Choosing $\ARC{\CC}$ Per Tile]
	\label{remark:choosing:class:per:tile}
	Theorem~\ref{thm:termination:by:element:counting} and the results it depends on still hold if in instead of having $\ARC{\CC}$ globally fixed, a class $\ARC{\CC}$ is fixed for each individual $T \in \TT$, and match morphism $m$ is in the intersection of every class. This for instance allows counting some tiles monically, and others non-monically.
\end{rem}

The following lemma implies that in many categories of interest, tilings of $L$ and slid tiles never collide, in which case one condition on $\Delta_{\rho}$ from Theorem~\ref{thm:termination:by:element:counting} is vacuously satisfied. For instance, in quasitoposes, pushouts along $t_K$ are pullbacks if $t_K \in \RegC{\CC}$~\cite[Lemma~A2.6.2]{johnstone2002sketches}.

\begin{lem}
	\label{lemma:sliding:does:not:collide}
	Let $\CC$ have pullbacks, and let a rule $\rho$ be given.
	If $t_K \in \MonoC{\CC}$ and pushout ${r' \circ t_K} = {t_R \circ r}$ is a pullback, then for all $t \in \nisosub{\homset{T}{R'}}{t_R}$ with
	$\pb{r'}{t}$ a split epi, we have
	\[
	{l' \circ \pb{t}{r'} \circ \rightinv{\pb{r'}{t}}} \neq {t_L \circ t'}
	\]
	for all $T' \in \obj{\CC}$ and $t' \in \homset{T'}{L}$.
\end{lem}
\begin{proof}
	By contradiction. Assume that for some $t \in \nisosub{\homset{T}{R'}}{t_R}$ with
	$\pb{r'}{t}$ a split epi and some $t' \in \homset{T'}{L}$, ${l' \circ \pb{t}{r'} \circ \rightinv{\pb{r'}{t}}} = {t_L \circ t'}$. Then we have the commuting diagram
	\begin{center}
		\begin{tikzcd}[column sep=50, row sep=25]
			& T' \arrow[d, "t'" description]                                                                                                        & S' \arrow[d, "\pb{t'}{l}" description] \arrow[l, "\pb{l}{t'}" description] \arrow[ld, "\mathrm{PB}(1)", phantom]                                                                    &                                \\
			& L \arrow[d, "t_L" description]                                                                                                        & K \arrow[l, "l" description] \arrow[r, "r" description] \arrow[d, "t_K" description, tail] \arrow[ld, "\mathrm{PB}(2)", phantom] \arrow[rd, "\mathrm{PO + PB}", phantom]            & R \arrow[d, "t_R" description] \\
			& L'                                                                                                                                    & K' \arrow[l, "l'" description] \arrow[r, "r'" description]                                                                                                                          & R'                             \\
			T' \arrow[ru, "t_L \circ t'" description] & S' \arrow[ru, "t_K \circ \pb{t'}{l}" description] \arrow[l, "\pb{l}{t'}" description] \arrow[u, "\mathrm{PB}(3)", phantom] & S \arrow[u, "\pb{t}{r'}" description] \arrow[r, "\pb{r'}{t}" description, two heads] \arrow[ru, "\mathrm{PB}(4)", phantom]                                                          & T \arrow[u, "t" description]   \\
			&                                                                                                                                       & T \arrow[u, "\rightinv{\pb{r'}{t}}" description, hook'] \arrow[ru, equals] \arrow[lu, "x" description, dotted] \arrow[llu, equals] &                               
		\end{tikzcd}
	\end{center}
	where
	\begin{itemize}
		\item squares $\mathrm{PB}(2)$ and $\mathrm{PO}+\mathrm{PB}$ are given by $\rho$;
		\item squares $\mathrm{PB}(1)$ and $\mathrm{PB}(4)$ are constructed;
		\item square $\mathrm{PB}(3)$ follows from $\mathrm{PB}(1)$ and $\mathrm{PB}(2)$, using the pullback lemma;
		\item morphism $x$ exists by the hypothesis and the pullback property; and
		\item ${\pb{r'}{t} \circ \rightinv{\pb{r'}{t}}} = \id{T}$ by the right inverse property.
	\end{itemize}
	By a diagram chase we thus have $t = {r' \circ t_K \circ \pb{t'}{l} \circ x}$. Then from
	\begin{center}
		\begin{tikzcd}[column sep=50]
			R \arrow[d, "t_R" description] & K \arrow[d, "t_K" description, tail] \arrow[l, "r" description] \arrow[ld, "\mathrm{PB}", phantom] & K \arrow[l, equals] \arrow[d, equals] \arrow[ld, "\mathrm{PB}", phantom] & T \arrow[l, "\pb{t'}{l} \circ x" description] \arrow[d, equals] \arrow[ld, "\mathrm{PB}", phantom] \\
			R'                             & K' \arrow[l, "r'" description]                                                                     & K \arrow[l, "t_K" description, tail]                                                               & T \arrow[l, "\pb{t'}{l} \circ x" description] \arrow[lll, "t" description, bend left=15]                          
		\end{tikzcd}
	\end{center}
	and two applications of the pullback lemma, we have $\pb{t_R}{t} \in \IsoC{\CC}$, contradicting $t \in \nisosub{\homset{T}{R'}}{t_R}$.
\end{proof}

The proposition below states a strong sufficient condition for satisfying the termination method's preconditions. Many graph categories of interest meet these conditions.

\begin{prop}
	\label{proposition:rmadhesive:quasitopos:consequences}
	If $\CC$ is an $\Reg$-locally finite, rm-adhesive quasitopos, then $\RegC{\CC}$ satisfies Assumption~\ref{assumption:ar}. $\CC$ also has all pullbacks and all pushouts, and so in particular the required pushouts described in Assumption~\ref{assumption:rule:po}. If moreover $t_L(\rho) \in \RegC{\CC}$, then $m, t_K, t_R \in \RegC{\CC}$, $\rho$'s pushout square is \pbpostrong{}-adhesive, and $t_L$ is monic for $\ARset{\TT}{L(\rho)}$.
\end{prop}
\begin{proof}
	See Proposition~\ref{prop:quasitopos:properties:summary} for a summary of quasitopos properties. That $\RegC{\CC}$ satisfies Assumption~\ref{assumption:ar} was stated in Proposition~\ref{prop:assumption:regular:monos:stable:class:quasitopos}. If $t_L \in \RegC{\CC}$, and $m,t_K \in \RegC{\CC}$ by stability under decomposition and pullback, respectively. That $t_R \in \RegC{\CC}$ subsequently follows from pushout stability in quasitoposes~\cite[Lemma A.2.6.2]{johnstone2002sketches}. Because of the assumed rm-adhesivity and $t_K \in \RegC{\CC}$, $\rho$'s pushout square is \pbpoadhesive{}    (Corollary~\ref{cor:rm:adhesive:pbpo:adhesive}). Finally, that $t_L$ is monic for $\ARset{\TT}{L(\rho)}$ follows trivially from the fact that $t_L$ is monic.
\end{proof}

Finally, we have the following general principle, which does not require any assumptions on $t_K$ and $t_R$, nor any adhesivity assumptions.

\begin{lem}[Deleting Rules Are Decreasing]
	\label{lemma:deleting:rules:terminate}
	Assume $\CC$ has pullbacks and is $\Mono$-locally finite. Suppose that for a \pbpostrong{} rule $\rho$, $l'$ is monic, $l$ is not epic, and $r$ is iso; and that for any matches $m$ for $\rho$, $m$ is monic. Then $\rho$ is decreasing for $\TT = \{ L \}$, $\ww(L) > 0$, $\tilingL_\rho = \{ \id{L} \}$ and $\ARC{\CC} = \MonoC{\CC}$.
\end{lem}
\begin{proof}
	Let a rewrite step diagram be given. By stability properties, morphism $g_L$ is monic because $l'$ is monic, and $g_R$ is iso because $r$ is iso. Then 
	\[
	\xi : \homsetmono{L}{G_R} \mono \homsetmono{L}{G_L}
	\]
	defined by $\xi(t) = g_L \circ \inverse{g_R} \circ t$ is an injection because $g_L \circ \inverse{g_R}$ is monic. But $\xi$ is not surjective. Because if $m = \xi(t)$ for some $t \in \homsetmono{L}{G_R}$, then
	in diagram
	\begin{center}
		\begin{tikzcd}[column sep=40]
			L \arrow[d, "m" description, tail] & K \arrow[d, "u" description, tail] \arrow[l, "l" description, tail] \arrow[ld, "\mathrm{PB}", phantom] &                                                    & L \arrow[ll, "!x" description, tail, dotted]  \arrow[d, equals] \arrow[lll, equals, bend right=15] \\
			G_L                          & G_K \arrow[l, "g_L" description, tail]                                                           & G_R \arrow[l, "g_R^{{-1}}" description, two heads, tail] & L \arrow[l, "t" description, tail] \arrow[lll, "\xi(t)" description, bend left=15, tail]                                                                    
		\end{tikzcd}
	\end{center}
	a unique $x : L \to K$ is obtained that makes the diagram commute, using the pullback square on the left. Thus $l \circ x$ is an isomorphism, and hence $l$ an epimorphism, which contradicts our assumption about $l$. So $m \in \homsetmono{L}{G_L}$ does not lie in the image of $\xi$.
	It follows that $\wwTAR{\TT}{\AR}(G_L) >
	\wwTAR{\TT}{\AR}(G_R)$.
\end{proof}

\section{Examples}
\label{section:examples}

We give a number of examples of applying Theorem~\ref{thm:termination:by:element:counting} in category $\CC = \text{\FinGraph}$ (Definition~\ref{definition:category:graph}), each demonstrating new features. For each example, we will fix $\TT$, $\ww$, and $\ARC{\CC}$, and usually some properties of the relevant morphism sets (such as cardinalities) or related comments. The remaining details of the proofs are routine.
Note that in \FinGraph{} (and more generally in any topos), $\RegC{\CC} = \MonoC{\CC}$, and because the rules in examples satisfy $t_L \in \MonoC{\CC}$, we are in each case free to choose $\MonoC{\CC}$ or $\HomC{\CC}$ for $\ARC{\CC}$ (Remark~\ref{remark:matches:lower:bound:A}).

\tikzset{epat/.style={thick}}
\tikzset{eset/.style={densely dotted}}
\tikzset{npattern/.style={rectangle,rounded corners=2mm,draw=black,inner sep=0.5mm,outer sep=.5mm,minimum size=4.5mm}}
\tikzset{nset/.style={npattern,draw=black,fill=white,densely dotted}}
\tikzset{label/.style={scale=0.85,inner sep=0,outer sep=0.5mm}}

\begin{notation}[Visual Notation]
	\label{notation:visual}
	In our examples of rules, the morphisms $t_X : X \hookrightarrow X'$ ($X \in \{L, K, R\}$) of rules are regular monos (embeddings). We depict $t_X$ by depicting the graph $X'$, and then let solid, colored vertices and solid edges denote $t_X(X)$, with dotted blank vertices and dotted edges the remainder of $X'$. For example, in Example~\ref{example:fold:edge} below, subgraph $L$ of $t_L : L \regmono L'$ is
	{
		\hspace{-3.4mm}
		\newcommand{\nodex}{\vertex{x}{cblue!20}}
		\newcommand{\nodey}{\vertex{y}{cgreen!20}}
		\begin{tikzpicture}[->,node distance=12mm,n/.style={},baseline=-0.8ex]
			\node [npattern] (x) at (-6mm,0mm) {\nodex};
			\node [npattern] (y) at (6mm,0mm) {\nodey};
			\draw [epat] (x) to (y);
		\end{tikzpicture}%
	}.
	
	The vertices of graphs are non-empty sets $\{ x_1,\ldots,x_n \}$ depicted by boxes
	{
		\hspace{-3.4mm}
		\newcommand{\nodexa}{\vertex{x_1}{cblue!20}}
		\newcommand{\nodexb}{\vertex{x_n}{cblue!20}}
		\begin{tikzpicture}[->,node distance=12mm,n/.style={},baseline=-0.8ex]
			\node [npattern] (xaxb)
			{\nodexa \ \raisebox{1mm}{$\cdots$} \  \nodexb};
		\end{tikzpicture}%
	}. When depicting the homomorphism $(r' = (r'_V, r'_E)): K' \to R'$ we will choose the vertices of $K'$ and $R'$ in such a way that component $r'_V$ is fully determined by $S \subseteq r'_V(S)$ for all $S \in V_{K'}$. For example, for nodes $\{ x \}, \{ y \} \in V_{K'}$ of Example~\ref{example:fold:edge} below (in which morphism $r'$ is implicit), $r'(\{x\}) = r'(\{ y \}) = \{ x,y \} \in V_{R'}$. If component $r'_E$ is not uniquely determined by $r'_V$, then let $r'_E$ preserve the relative positioning of the edges (although normally, this choice will be inconsequential). Morphism $l' : K' \to L'$ is depicted similarly.
\end{notation}

\begin{exa}[Folding an Edge]
	\label{example:fold:edge}
	The rule
	\begin{center}
		$\rho = $ \scalebox{0.9}{
%
{
\newcommand{\context}{
          \node [nset] (c) at (0mm,-10mm) {\nodectx};
          \draw [eset] (c) to[bend left=15] (x);
          \draw [eset] (x) to[bend left=5] (c);
          \draw [eset] (c) to[bend left=5] (y);
          \draw [eset] (y) to[bend left=15] (c);
          \draw [eset] (c) to[thinloop=0,distance=1.5em] (c);
          \draw [eset] (x) to[thinloop=-135,distance=1.5em] (x);
          \draw [eset] (y) to[thinloop=-45,distance=1.5em] (y);
          \draw [eset] (x.105) to[out=60,in=120,distance=5mm] (y.75);
          \draw [eset] (y.105) to[out=135,in=45,distance=3.5mm] (x.75);
}%
\newcommand{\nodex}{\vertex{x}{cblue!20}}%
\newcommand{\nodey}{\vertex{y}{cgreen!20}}%
\newcommand{\nodez}{\vertex{z}{corange!20}}%
\newcommand{\nodew}{\vertex{w}{cred!20}}%
\newcommand{\nodectx}{\vertex{c}{white}}%
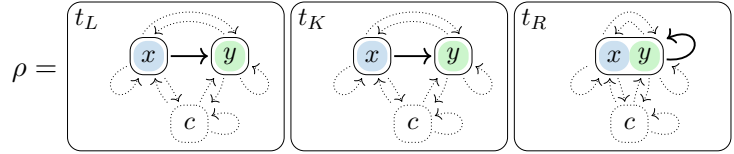
\begin{tikzpicture}[baseline=-22mm,->,node distance=12mm,n/.style={},epattern/.style={very thick},mred/.style={cred,very thick},mblue/.style={cblue,very thick},mgreen/.style={cdarkgreen,very thick},mpurple/.style={cpurple,very thick,densely dotted}]
        \graphboxx{L}{0mm}{-11mm}{32mm}{22mm}{2mm}{-8mm}{
          \node [npattern] (x) at (-6mm,0mm) {\nodex};
          \node [npattern] (y) at (6mm,0mm) {\nodey};
          \draw [epat] (x) to (y);
          \context
        }
        \graphboxx{K}{33mm}{-11mm}{32mm}{22mm}{2mm}{-8mm}{
          \node [npattern] (x) at (-6mm,0mm) {\nodex};
          \node [npattern] (y) at (6mm,0mm) {\nodey};
          \draw [epat] (x) to (y);
          \context
        }
        \begin{scope}[opacity=1]        
        \graphboxx{R}{66mm}{-11mm}{32mm}{22mm}{1mm}{-8mm}{
          \node [npattern,opacity=0] (x) at (-3mm,0mm) {\nodex};
          \node [npattern,opacity=0] (y) at (3mm,0mm) {\nodey};
          \context
          \node [npattern] (x) at (0mm,0mm) {\nodex\nodey};
          \draw [epat] (x) to [thinloop=10,out=-5,in=25,distance=1.5em] (x);
        }
        \end{scope}
      \end{tikzpicture}
}
}
	\end{center}
	folds a non-loop edge $\edgeGraph$ into a loop $\aloopUnlabeled{}$ (in any context). Define tile set $\TT = \{ \edgeGraph \}$ with $\ww(\edgeGraph) = 1$ and fix $\ARC{\CC} = \MonoC{\CC}$. Then $\card{\Phi_{\rho}} = 5$ (every non-loop dotted edge in $R'$, and the loop on $c$). For every $f \in \Phi_{\rho}$, $\pb{r'}{f}$ is iso and hence a regular epi with a unique right inverse. Because $l'$ is monic (even iso), the remaining details of Theorem~\ref{thm:xi:is:injection} are immediate. Finally, for the only choice of $\Delta_\rho$, $\ww(\Delta_\rho) = 1 > \ww(\isosub{\homsetmono{\TT}{R'}}{t_R}) = \ww(\emptyset) = 0$. So $\rho$ is terminating.
\end{exa}

\begin{exa}
	\label{example:edge:graph:decreasing:system}
	Let the rewrite system $\sys$ consist of rules
	\begin{center}
		$\rho = $ \scalebox{0.9}{
%
{
\newcommand{\context}{
          \node [nset] (c) at (0mm,-10mm) {\nodectx};
          \draw [eset] (c) to[bend left=15] (x);
          \draw [eset] (x) to[bend left=15] (c);
          \draw [eset] (c) to[thinloop=0,distance=1.5em] (c);
          \draw [eset] (x) to[thinloop=-180,distance=1.5em] (x);
}%
\newcommand{\nodex}{\vertex{x}{cblue!20}}%
\newcommand{\nodey}{\vertex{y}{cgreen!20}}%
\newcommand{\nodez}{\vertex{z}{corange!20}}%
\newcommand{\nodew}{\vertex{w}{cred!20}}%
\newcommand{\nodectx}{\vertex{c}{white}}%
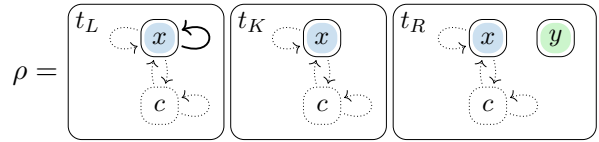
      \begin{tikzpicture}[->,node distance=12mm,n/.style={},epattern/.style={very thick},mred/.style={cred,very thick},mblue/.style={cblue,very thick},mgreen/.style={cdarkgreen,very thick},mpurple/.style={cpurple,very thick,densely dotted},baseline=-22mm]
        \graphboxx{L}{0mm}{-11mm}{23mm}{20mm}{2mm}{-5mm}{
          \node [npattern] (x) at (0mm,0mm) {\nodex};
          \draw [epat] (x) to[thinloop=0,distance=1.5em] (x);
          \context
        }
        \graphboxx{K}{24mm}{-11mm}{23mm}{20mm}{2mm}{-5mm}{
          \node [npattern] (x) at (0mm,0mm) {\nodex};
          \context
        }
        \begin{scope}[opacity=1]        
        \graphboxx{R}{48mm}{-11mm}{30mm}{20mm}{-1mm}{-5mm}{
          \node [npattern] (x) at (0mm,0mm) {\nodex};
          \node [npattern] (y) at (10mm,0mm) {\nodey};
          \context
        }
        \end{scope}
      \end{tikzpicture}
}
}
	\end{center}
	and
	\begin{center}
		$\tau = $ \scalebox{0.9}{
%
{
\newcommand{\context}{
          \node [nset] (c) at (0mm,-10mm) {\nodectx};
          \draw [eset] (c) to[bend left=15] (x);
          \draw [eset] (x) to[bend left=15] (c);
          \draw [eset] (c) to[bend left=15] (y);
          \draw [eset] (y) to[bend left=15] (c);
          \draw [eset] (c) to[thinloop=0,distance=1.5em] (c);
          \draw [eset] (x) to[thinloop=-135,distance=1.5em] (x);
          \draw [eset] (y) to[thinloop=-45,distance=1.5em] (y);
          \draw [eset] (x.105) to[out=60,in=120,distance=5mm] (y.75);
          \draw [eset] (y.105) to[out=135,in=45,distance=3.5mm] (x.75);
}%
\newcommand{\nodex}{\vertex{x}{cblue!20}}%
\newcommand{\nodey}{\vertex{y}{cgreen!20}}%
\newcommand{\nodez}{\vertex{z}{corange!20}}%
\newcommand{\nodew}{\vertex{w}{cred!20}}%
\newcommand{\nodectx}{\vertex{c}{white}}%
      
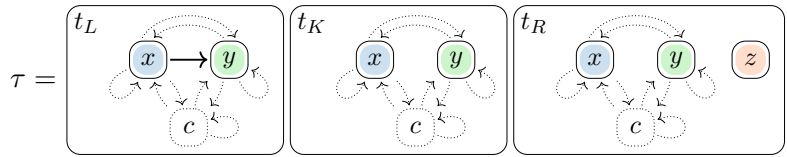
\begin{tikzpicture}[->,node distance=12mm,n/.style={},epattern/.style={very thick},mred/.style={cred,very thick},mblue/.style={cblue,very thick},mgreen/.style={cdarkgreen,very thick},mpurple/.style={cpurple,very thick,densely dotted},baseline=-23mm]
        \graphboxx{L}{0mm}{-11mm}{32mm}{22mm}{2mm}{-8mm}{
          \node [npattern] (x) at (-6mm,0mm) {\nodex};
          \node [npattern] (y) at (6mm,0mm) {\nodey};
          \draw [epat] (x) to (y);
          \context
        }
        \graphboxx{K}{33mm}{-11mm}{32mm}{22mm}{2.5mm}{-8mm}{
          \node [npattern] (x) at (-6mm,0mm) {\nodex};
          \node [npattern] (y) at (6mm,0mm) {\nodey};
          \context
        }
        \begin{scope}[opacity=1]        
        \graphboxx{R}{66mm}{-11mm}{40mm}{22mm}{-2mm}{-8mm}{
          \node [npattern] (x) at (-6mm,0mm) {\nodex};
          \node [npattern] (y) at (6mm,0mm) {\nodey};
          \node [npattern] (z) at (17mm,0mm) {\nodez};
          \context
        }
        \end{scope}
      \end{tikzpicture}
}
} 
	\end{center}
	Rule $\rho$ deletes a loop in any context, and adds a node; and rule $\tau$ deletes a non-loop edge in any context, and adds a node.
	
	Because the $r$ morphisms of $\rho$ and $\tau$ are not iso, Lemma~\ref{lemma:deleting:rules:terminate} cannot be applied. But for $T = \edgeGraph$, $\TT = \{ T \}$, $\ww(T) = 1$, $\ARC{\CC} = \HomC{\CC}$, and $\Delta_{\rho_1}$ and $\Delta_{\rho_2}$ the singleton sets containing the unique morphisms $T \to L(\rho)$ and $T \to L(\tau)$, respectively,
	$\sys$ is proven decreasing and thus terminating by Theorem~\ref{thm:termination:by:element:counting}. This argument captures the natural argument: ``the number of edges is decreasing''.
	
	An alternative argument lets $\TT = \{T, T'\}$, with 
	\smash{$T' = \aloopUnlabeled{}
		$}, $\ww(T) = \ww(T') = 1$, and $\ARC{\CC} = \MonoC{\CC}$. Then $\Delta_{\rho}$ contains the unique mono $T' \mono L(\rho)$ and $\Delta_{\tau}$ the unique mono $T \mono L(\tau)$. This captures the argument: ``the sum of loop and non-loop edges is decreasing''.
\end{exa}

\newcommand{\catName}[1]{{\normalfont\textbf{#1}}}

\begin{rem}[Element Counting in Fuzzy Presheaves]
	In a fuzzy graph category (Example~\ref{ex:fuzzy:presheaves}), rules that change labels (but leave the structure of the graph unchanged) can be proven terminating by using $\ARC{\CC} = \RegC{\CC}$, but not always by using $\ARC{\CC} = \MonoC{\CC}$. For instance, a rule that increases a loop edge label $a$ into label $b > a$, is shown terminating by $\TT = \{ \aloop{a} \}$ and $\ARC{\CC} = \RegC{\CC}$, but no proof exists for $\ARC{\CC} = \MonoC{\CC}$, because $\aloop{b}$ has strictly more monic elements than $\aloop{a}$.
\end{rem}

\newcommand{\twoEdgeGraph}{%
	\,\begin{tikzpicture}[baseline=-.6ex]
		\node (x) [circle,inner sep=0,outer sep=1mm,minimum size=0.7mm,fill=black] {};
		\node (y) [circle,inner sep=0,outer sep=1mm,minimum size=0.7mm,fill=black, right of=x, xshift=-3mm] {};
		\draw [->] ([yshift=.7mm]x.east) to ([yshift=.7mm]y.west);
		\draw [->] ([yshift=-.7mm]y.west) to ([yshift=-.7mm]x.east);
	\end{tikzpicture}\,%
}

\begin{exa}
	\label{example:dpo:as:pbpo}
	The rule
	\begin{center}
		$\rho = $ \scalebox{0.9}{
%
{
\newcommand{\context}{
          \node [nset] (c) at (0mm,-20mm) {\nodectx};
          \draw [eset] (c) to[bend left=15] (x);
          \draw [eset] (x) to[bend left=15] (c);
          \draw [eset] (c) to[bend left=15] (y);
          \draw [eset] (y) to[bend left=15] (c);
          \draw [eset] (c) to[thinloop=0,distance=1.5em] (c);
          \draw [eset] (x) to[thinloop=180,distance=1.5em] (x);
          \draw [eset] (y) to[thinloop=0,distance=1.5em] (y);
          \draw [eset] (x) to[bend left=25] (y);
          \draw [eset] (y) to[bend left=25] (x);
}%
\newcommand{\nodex}{\vertex{x}{cblue!20}}%
\newcommand{\nodey}{\vertex{y}{cgreen!20}}%
\newcommand{\nodez}{\vertex{z}{corange!20}}%
\newcommand{\nodew}{\vertex{w}{cred!20}}%
\newcommand{\nodectx}{\vertex{c}{white}}%
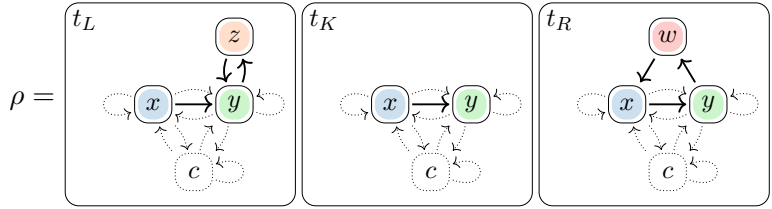
\begin{tikzpicture}[baseline=-36mm,->,node distance=12mm,n/.style={},epattern/.style={very thick},mred/.style={cred,very thick},mblue/.style={cblue,very thick},mgreen/.style={cdarkgreen,very thick},mpurple/.style={cpurple,very thick,densely dotted}]
        \graphboxx{L}{0mm}{-21mm}{34mm}{30mm}{2mm}{-5mm}{
          \node [npattern] (x) at (-6mm,-10mm) {\nodex};
          \node [npattern] (y) at (6mm,-10mm) {\nodey};
          \node [npattern] (z) at (6mm,0mm) {\nodez};
          \draw [epat] (x) to (y);
          \draw [epat] (y) to[bend right=20] (z);
          \draw [epat] (z) to[bend right=20] (y);
          \context
        }
        \graphboxx{K}{35mm}{-21mm}{34mm}{30mm}{2mm}{-5mm}{
          \node [npattern] (x) at (-6mm,-10mm) {\nodex};
          \node [npattern] (y) at (6mm,-10mm) {\nodey};
          \draw [epat] (x) to (y);
          \context
        }
        \begin{scope}[opacity=1]        
        \graphboxx{R}{70mm}{-21mm}{34mm}{30mm}{2mm}{-5mm}{
          \node [npattern] (x) at (-6mm,-10mm) {\nodex};
          \node [npattern] (y) at (6mm,-10mm) {\nodey};
          \node [npattern] (z) at (0mm,0mm) {\nodew};
          \draw [epat] (x) to (y);
          \draw [epat] (y) to[bend left=0] (z);
          \draw [epat] (z) to[bend left=0] (x);
          \context
        }
        \end{scope}

      \end{tikzpicture}
}
}
	\end{center}
	is proven terminating by $\TT = \{\twoEdgeGraph \}$, $\ww( \twoEdgeGraph) = 1$, 	and $\ARC{\CC} = \MonoC{\CC}$.
	
	This example can be used to show that applicability of our technique is not invariant under rule equivalence.
	Let $\rho'$ be obtained from $\rho$ by dropping the solid edge from $x$ to $y$ in $K$ and $K'$. The rules $\rho$ and $\rho'$ are equivalent: they induce the same rewrite relation. However, our termination technique fails for $\rho'$, because not all tiles can be intactly transferred. We therefore pose the following open question. \emph{Does there exist a procedure that maps a rule $\tau$ onto an equivalent ``standard'' representative $\mathit{standard}(\tau)$, such that
		the termination method fails on $\mathit{standard}(\tau)$ iff it fails on all rules in the equivalence class? And if it exists, can it be formulated on the level of (rm-adhesive) quasitoposes?}
\end{exa}

\begin{exa}
	Consider the rules $\rho$ and $\tau$, respectively:
	\begin{center}
		$\rho = $ \scalebox{0.9}{
%
{
\newcommand{\context}{
  \draw [eset] (c) to[thinloop=90,distance=1em] node [above,label,inner sep=0.5mm] {$a$} (c);
  \draw [eset] (c) to[thinloop=35,distance=1em] node [above,label,inner sep=0.5mm,pos=0.7] {$b$} (c);
}%
\newcommand{\nodex}{\vertex{x}{cblue!20}}%
\newcommand{\nodey}{\vertex{y}{cblue!20}}%
\newcommand{\nodez}{\vertex{z}{corange!20}}%
\newcommand{\nodeu}{\vertex{u}{cgreen!20}}%
\newcommand{\nodev}{\vertex{v}{cgreen!20}}%
\newcommand{\nodew}{\vertex{w}{cgreen!20}}%
\newcommand{\nodectx}{\vertex{c}{white}}%
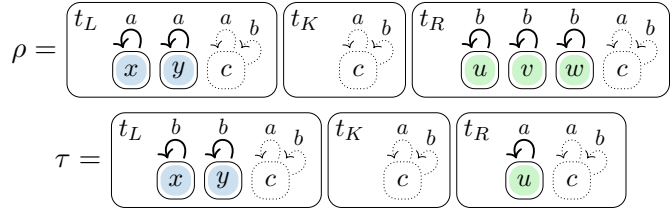
      \begin{tikzpicture}[->,node distance=12mm,n/.style={},epattern/.style={very thick},mred/.style={cred,very thick},mblue/.style={cblue,very thick},mgreen/.style={cdarkgreen,very thick},mpurple/.style={cpurple,very thick,densely dotted},baseline=-23mm]
        \graphboxx{L}{0mm}{-15mm}{31mm}{14mm}{1mm}{-10mm}{
          \node [npattern] (x) at (-7mm,0mm) {\nodex};
          \node [npattern] (y) at (0mm,0mm) {\nodey};
          \draw [epat] (x) to[thinloop=90,distance=1em] node [above,label,inner sep=0.5mm] {$a$} (x);
          \draw [epat] (y) to[thinloop=90,out=90-20,in=90+20,distance=1em] node [above,label,inner sep=0.5mm] {$a$} (y);
          \node [nset] (c) at (7mm,0mm) {\nodectx};
          \context
        }
        \graphboxx{K}{32mm}{-15mm}{18mm}{14mm}{2mm}{-10mm}{
          \node [nset] (c) at (0mm,0mm) {\nodectx};
          \context
        }
        \begin{scope}[opacity=1]        
        \graphboxx{R}{51mm}{-15mm}{38mm}{14mm}{-2mm}{-10mm}{
          \node [npattern] (x) at (-7mm,0mm) {\nodeu};
          \node [npattern] (y) at (0mm,0mm) {\nodev};
          \node [npattern] (z) at (7mm,0mm) {\nodew};
          \draw [epat] (x) to[thinloop=90,distance=1em] node [above,label,inner sep=0.5mm] {$b$} (x);
          \draw [epat] (y) to[thinloop=90,distance=1em] node [above,label,inner sep=0.5mm] {$b$} (y);
          \draw [epat] (z) to[thinloop=90,distance=1em] node [above,label,inner sep=0.5mm] {$b$} (z);
          \node [nset] (c) at (14mm,0mm) {\nodectx};
          \context
        }
        \end{scope}
      \end{tikzpicture}
}
}
	\end{center}
	\begin{center}
		$\tau = $ \scalebox{0.9}{
%
{
\newcommand{\context}{
  \draw [eset] (c) to[thinloop=90,distance=1em] node [above,label,inner sep=0.5mm] {$a$} (c);
  \draw [eset] (c) to[thinloop=35,distance=1em] node [above,label,inner sep=0.5mm,pos=0.7] {$b$} (c);
}%
\newcommand{\nodex}{\vertex{x}{cblue!20}}%
\newcommand{\nodey}{\vertex{y}{cblue!20}}%
\newcommand{\nodeu}{\vertex{u}{cgreen!20}}%
\newcommand{\nodew}{\vertex{w}{cred!20}}%
\newcommand{\nodectx}{\vertex{c}{white}}%
      \begin{tikzpicture}[->,node distance=12mm,n/.style={},epattern/.style={very thick},mred/.style={cred,very thick},mblue/.style={cblue,very thick},mgreen/.style={cdarkgreen,very thick},mpurple/.style={cpurple,very thick,densely dotted},baseline=-23mm]
        \graphboxx{L}{0mm}{-15mm}{31mm}{14mm}{1mm}{-10mm}{
          \node [npattern] (x) at (-7mm,0mm) {\nodex};
          \node [npattern] (y) at (0mm,0mm) {\nodey};
          \draw [epat] (x) to[thinloop=90,distance=1em] node [above,label,inner sep=0.5mm] {$b$} (x);
          \draw [epat] (y) to[thinloop=90,out=90-20,in=90+20,distance=1em] node [above,label,inner sep=0.5mm] {$b$} (y);
          \node [nset] (c) at (7mm,0mm) {\nodectx};
          \context
        }
        \graphboxx{K}{32mm}{-15mm}{18mm}{14mm}{2mm}{-10mm}{
          \node [nset] (c) at (0mm,0mm) {\nodectx};
          \context
        }
        \begin{scope}[opacity=1]        
        \graphboxx{R}{51mm}{-15mm}{25mm}{14mm}{1mm}{-10mm}{
          \node [npattern] (x) at (-3.5mm,0mm) {\nodeu};
          \draw [epat] (x) to[thinloop=90,distance=1em] node [above,label,inner sep=0.5mm] {$a$} (x);
          \node [nset] (c) at (3.5mm,0mm) {\nodectx};
          \context
        }
        \end{scope}
      \end{tikzpicture}
}
}
	\end{center}
	Intuitively, these rules model replacements in multisets over $\{a,b\}$. The elements of the multiset are modeled by nodes with loops that carry the label $a$ or $b$, respectively.
	Rule $\rho$ replaces two $a$'s by three $b$'s,
	and $\tau$ replaces two $b$'s by one $a$.
	
	These rules are terminating.
	To prove this, we use tiles $\TT = \{ \aloop{a}, \aloop{b} \}$ together with the weight assignment $\ww(\aloop{a}) = 5$ and $\ww(\aloop{b}) = 3$, and let $\ARC{\CC} = \MonoC{\CC}$. The context is isomorphically preserved along $l'$ and $r'$, and partial overlaps with the pattern are not possible. So Theorem~\ref{thm:xi:is:injection} is easily verified for $\Phi_\rho$ and $\Phi_\tau$. Then for the obvious largest choices of $\Delta_\rho$ and $\Delta_\tau$, we have $\ww(\Delta_\rho) = 2 \cdot 5 = 10 >  \ww(\isosub{\homsetmono{\TT}{\rho(R')}}{\rho(t_R)}) = 3 \cdot 3 = 9$ for $\rho$ and
	$\ww(\Delta_\tau) = 2 \cdot 3 = 6 >  \ww(\isosub{\homsetmono{\TT}{\tau(R')}}{\tau(t_R)}) = 5$ for $\tau$, completing the proof.
	
	The above termination proof works also for vast generalizations of the rules. For instance, rule $\rho$ can be generalized to rule $\rho'$ as follows:
	\begin{center}
		\scalebox{0.9}{
%
{
\newcommand{\context}{
  \draw [eset] (c) to[thinloop=180,distance=1em] node [left,label,inner sep=0.5mm] {$a,b$} (c);
  \draw [eset] (c) to[bend left=15] node [left,label,inner sep=0.7mm] {$a,b$} (y);
  \draw [eset] (y) to[bend left=15] node [right,label,inner sep=0.7mm,pos=0.4] {$a,b$} (c);
  \draw [eset] (x) to[thinloop=-90,distance=1em] node [left,label,inner sep=0.5mm,pos=0.2] {$a,b$} (x);
}%
\newcommand{\nodex}{\vertex{x}{cblue!20}}%
\newcommand{\nodey}{\vertex{y}{cgreen!20}}%
\newcommand{\nodez}{\vertex{z}{corange!20}}%
\newcommand{\nodew}{\vertex{w}{cred!20}}%
\newcommand{\nodectx}{\vertex{c}{white}}%
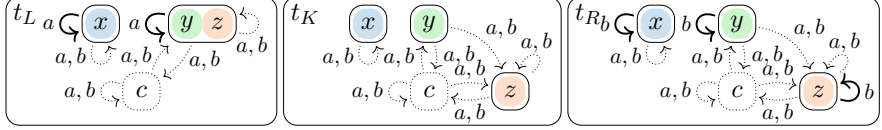
      \begin{tikzpicture}[->,node distance=12mm,n/.style={},epattern/.style={very thick},mred/.style={cred,very thick},mblue/.style={cblue,very thick},mgreen/.style={cdarkgreen,very thick},mpurple/.style={cpurple,very thick,densely dotted}]
        \graphboxx{L}{0mm}{-9mm}{40mm}{19mm}{0mm}{-4mm}{
          \node [npattern] (x) at (-6mm,0mm) {\nodex};
          \node [npattern] (y) at (9mm,0mm) {\nodey\nodez};
          \draw [epat] (x) to[thinloop=180,distance=1em] node [left,label,inner sep=0.5mm] {$a$} (x);
          \draw [epat] (y) to[thinloop=0,out=180-14,in=180+14,distance=1em] node [left,label,inner sep=0.5mm] {$a$} (y);
          \node [nset] (c) at (0mm,-10mm) {\nodectx};
          \context
          \draw [eset] (y) to[thinloop=0,out=0-14,in=0+14,distance=1em] node [below,label,inner sep=0.1mm,pos=0.2] {$a,b$} (y);
        }
        \graphboxx{K}{41mm}{-9mm}{41mm}{19mm}{1mm}{-4mm}{
          \node [npattern] (x) at (-9mm,0mm) {\nodex};
          \node [npattern] (y) at (0mm,0mm) {\nodey};
          \node [npattern] (z) at (12mm,-10mm) {\nodez};
          \node [nset] (c) at (0mm,-10mm) {\nodectx};
          \context
          \draw [eset] (c) to[bend left=15] node [above,label,inner sep=0.1mm] {$a,b$} (z);
          \draw [eset] (z) to[bend left=15] node [below,label,inner sep=0.1mm] {$a,b$} (c);
          \draw [eset] (y) to[bend left=30] node [above right,label,inner sep=0mm] {$a,b$} (z);
          \draw [eset] (z) to[thinloop=0,out=55-14,in=55+14,distance=1em] node [above,label,inner sep=0.1mm,pos=0.5] {$a,b$} (z);
        }
        \begin{scope}[opacity=1]        
        \graphboxx{R}{83mm}{-9mm}{46mm}{19mm}{-2mm}{-4mm}{
          \node [npattern] (x) at (-8mm,0mm) {\nodex};
          \node [npattern] (y) at (4mm,0mm) {\nodey};
          \node [npattern] (z) at (16mm,-10mm) {\nodez};
          \node [nset] (c) at (4mm,-10mm) {\nodectx};
          \context
          \draw [eset] (c) to[bend left=15] node [above,label,inner sep=0.1mm] {$a,b$} (z);
          \draw [eset] (z) to[bend left=15] node [below,label,inner sep=0.1mm] {$a,b$} (c);
          \draw [eset] (y) to[bend left=30] node [above right,label,inner sep=0mm] {$a,b$} (z);
          \draw [epat] (x) to[thinloop=180,distance=1em] node [left,label,inner sep=0.5mm] {$b$} (x);
          \draw [epat] (y) to[thinloop=180,distance=1em] node [left,label,inner sep=0.5mm] {$b$} (y);
          \draw [epat] (z) to[thinloop=0,distance=1em] node [right,label,inner sep=0.5mm] {$b$} (z);
          \draw [eset] (z) to[thinloop=0,out=55-14,in=55+14,distance=1em] node [above,label,inner sep=0.1mm,pos=0.5] {$a,b$} (z);
        }
        \end{scope}
      \end{tikzpicture}
}
}
	\end{center}
	Observe that $L'$ now allows an unbounded number of additional loops on the nodes, and an unbounded number of edges between the nodes and the context. The morphism $l'$ preserves the loops, duplicates a node (including its incident edges from and to the context), and unfolds loops between the duplicated nodes.  The system $\{ \rho', \tau \}$ is still proven terminating using the argument above.
	
	More generally, as long as $l'$ and $r'$ of a rule do not create new loops other than those specified by $l$ and $r$, the rule can be proven terminating using the argument above.
\end{exa}

\begin{exa}
	\label{example:delete:loop:duplicate}
	Consider the following rules:
	\begin{center}
		$\rho = $ \scalebox{0.9}{
%
{
\newcommand{\context}{
          \node [npattern] (n) at (-5mm,0mm) {\noden};
          \node [nset] (c) at (-5mm,-9mm) {\nodec};
          \draw [eset] (n) to[thinloop=180,distance=1em]  (n);
          \draw [eset] (c) to[thinloop=0,distance=1em]  (c);
          \draw [eset] (n) to[bend left=20] (c);
          \draw [eset] (c) to[bend left=20] (n);
}%
\newcommand{\contextb}{
          \node [nset] (w) at (5mm,0mm) {\nodew};
          \node [nset] (x) at (12mm,0mm) {\nodex};
          \node [nset] (y) at (5mm,-9mm) {\nodey};
          \node [nset] (z) at (12mm,-9mm) {\nodez};
          \draw [eset] (x) to[bend left=20] (z);
          \draw [eset] (z) to[bend left=20] (x);
          \draw [eset] (w) to[bend left=20] (y);
          \draw [eset] (y) to[bend left=20] (w);
}%
\newcommand{\nodew}{\vertex{w}{white}}
\newcommand{\nodex}{\vertex{x}{white}}
\newcommand{\nodey}{\vertex{y}{white}}
\newcommand{\nodez}{\vertex{z}{white}}
\newcommand{\noden}{\vertex{n}{cgreen!20}}%
\newcommand{\nodec}{\vertex{c}{white}}%
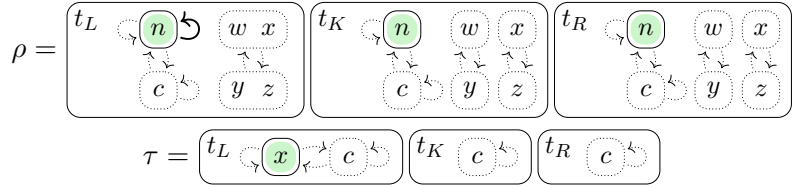
\begin{tikzpicture}[->,node distance=12mm,n/.style={},epattern/.style={very thick},mred/.style={cred,very thick},mblue/.style={cblue,very thick},mgreen/.style={cdarkgreen,very thick},mpurple/.style={cpurple,very thick,densely dotted},baseline=-17mm]
        \graphboxx{L}{0mm}{-9mm}{35mm}{17mm}{1mm}{-4mm}{
          \context
          \draw [epat] (n) to[thinloop=0,distance=1em] (n);
          \node [nset] (x) at (9mm,0mm) {\nodew\nodex};
          \node [nset] (y) at (9mm,-9mm) {\nodey\nodez};
          \draw [eset] (x) to[bend left=20] (y);
          \draw [eset] (y) to[bend left=20] (x);

        }
        \graphboxx{K}{36mm}{-9mm}{35mm}{17mm}{1mm}{-4mm}{
          \context
          \contextb
        }
        \begin{scope}[opacity=1]        
        \graphboxx{R}{72mm}{-9mm}{35mm}{17mm}{1mm}{-4mm}{
          \context
          \contextb
        }
        \end{scope}
      \end{tikzpicture}
}
}
	\end{center}
	\begin{center}
		$\tau = $ \scalebox{0.9}{
%
{
\newcommand{\context}{
}%
\newcommand{\nodew}{\vertex{w}{cblue!20}}%
\newcommand{\nodex}{\vertex{x}{cgreen!20}}%
\newcommand{\nodey}{\vertex{y}{corange!20}}%
\newcommand{\nodez}{\vertex{z}{cred!20}}%
\newcommand{\nodec}{\vertex{c}{white}}%
\begin{tikzpicture}[->,node distance=12mm,n/.style={},epattern/.style={very thick},mred/.style={cred,very thick},mblue/.style={cblue,very thick},mgreen/.style={cdarkgreen,very thick},mpurple/.style={cpurple,very thick,densely dotted},baseline=-14mm]
        \graphboxy{L}{0mm}{-9mm}{30mm}{8mm}{2mm}{-4mm}{
          \node [npattern] (x) at (-5mm,0mm) {\nodex};
          \draw [eset] (x) to[thinloop=180,distance=1em] node [right,label,inner sep=0.5mm] {} (x);
          \node [nset] (c) at (5mm,0mm) {\nodec};
          \draw [eset] (c) to[thinloop=0,distance=1em] node [right,label,inner sep=0.5mm] {} (c);
          \draw [eset] (x) to[bend left=20] (c);
          \draw [eset] (c) to[bend left=20] (x);
        }
        \graphboxy{K}{31mm}{-9mm}{18mm}{8mm}{1mm}{-4mm}{
          \node [nset] (c) at (0mm,0mm) {\nodec};
          \draw [eset] (c) to[thinloop=0,distance=1em] node [right,label,inner sep=0.5mm] {} (c);
        }
        \begin{scope}[opacity=1]        
        \graphboxy{R}{50mm}{-9mm}{18mm}{8mm}{1mm}{-4mm}{
          \node [nset] (c) at (0mm,0mm) {\nodec};
          \draw [eset] (c) to[thinloop=0,distance=1em] node [right,label,inner sep=0.5mm] {} (c);
        }
        \end{scope}
      \end{tikzpicture}
}
}
	\end{center}
	Rule $\rho$ deletes an arbitrary loop, and in doing so, allows arbitrarily many bipartite graph components in the context to duplicate (such components can either be mapped onto node $c$ or onto the right subgraph component). Note that this makes the rule non-deterministic.
	Rule $\tau$ deletes an arbitrary node including incident edges.
	
	Termination of the system can be proven as follows.
	Let $\ARC{\CC} = \MonoC{\CC}$.
	Use the tile set $\TT = \{ \aloopUnlabeled{} \}$ with the weight assignment $\ww(\aloopUnlabeled{}) = 1$. Then $\rho$ is decreasing and $\tau$ is non-increasing, and so it suffices to prove $\tau$ terminating, whose termination is immediate from Lemma~\ref{lemma:deleting:rules:terminate}.
	
	The derivational complexity\footnote{The \emph{derivational complexity} of a terminating term rewrite system is a function $f : \nat \to \nat$ such that $f(n)$ is the length of the longest rewrite sequence for any starting term of size $n$.} of this system is $O(2^mn)$ where $m = {\card{E_G}}$ and $n = {\card{V_G}}$ of the starting graph $G$. This can be seen as follows. In the worst case, $G$ consists of
	\begin{itemize}
		\item one distinguished node $x$ with $m$ loops; and
		\item the remaining $n - 1$ nodes without any incident edges.
	\end{itemize}
	In this case, each of the $m$ edges gives rise to an application of rule $\rho$, and each node forms a bipartite graph component, with the exception of node $x$. Assuming in the worst case that
	\begin{itemize}
		\item all $\rho$-steps precede $\tau$-steps; and
		\item all of the nodes except $x$ are duplicated at every $\rho$-step;
	\end{itemize}
	this scenario gives rise to the following recurrence relation $\mathit{max\_len}(m, n)$ for computing the derivation length:
	\begin{itemize}
		\item $\mathit{max\_len}(0, n) = n$; and
		\item $\mathit{max\_len}(m + 1, n) = {1 + \mathit{max\_len}(m, 2n -1)}$.
	\end{itemize}
	
	Using a routine inductive proof it can be shown that
	\[
	\mathit{max\_len}(m, n) = {2^m(n-1) + m + 1}
	\]
	which is in $O(2^mn)$.
\end{exa}

\begin{exa}
	\label{example:plump:string}
	We consider a system consisting of two DPO rules, given by Plump~{{\cite[Example 3]{plump2018modular}}, \cite[Example 3.8]{plump1995on}}. The equivalent \pbpostrong{} rules (using the standard encoding~\cite[Definition~71]{overbeek2023quasitoposes}) are as follows:
	\begin{center}
		$\rho = $ \scalebox{0.9}{
%
{
\newcommand{\context}{
          \node [nset] (c) at (0mm,-10mm) {\nodectx};
          \draw [eset] (c) to[bend left=15] (x);
          \draw [eset] (x) to[bend left=15] (c);
          \draw [eset] (c) to[bend left=15] (z);
          \draw [eset] (z) to[bend left=15] (c);
          \draw [eset] (c) to[thinloop=-45,distance=1.5em] (c);
          \draw [eset] (x) to[thinloop=-135,distance=1.5em] (x);
          \draw [eset] (z) to[thinloop=-45,distance=1.5em] (z);
          \draw [eset] (x.105) to ++(0mm,3.5mm) -| (z.75);
          \draw [eset] (z.105) to ++(0mm,2mm) -| (x.75);
}%
\newcommand{\nodex}{\vertex{x}{cblue!20}}%
\newcommand{\nodey}{\vertex{y}{cgreen!20}}%
\newcommand{\nodez}{\vertex{z}{corange!20}}%
\newcommand{\nodew}{\vertex{w}{cred!20}}%
\newcommand{\nodectx}{\vertex{c}{white}}%
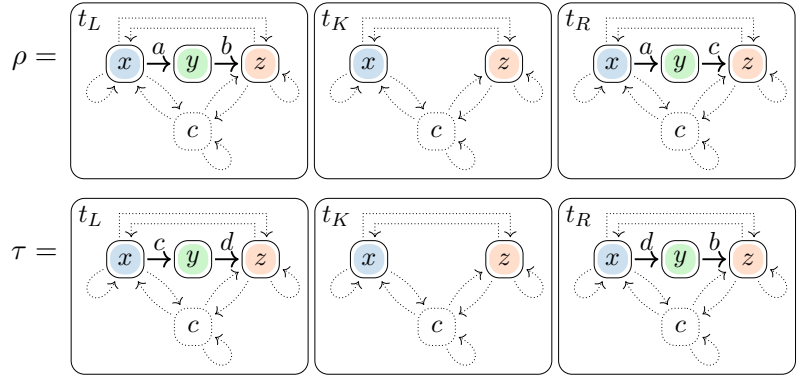
      \begin{tikzpicture}[->,node distance=12mm,n/.style={},epattern/.style={very thick},mred/.style={cred,very thick},mblue/.style={cblue,very thick},mgreen/.style={cdarkgreen,very thick},mpurple/.style={cpurple,very thick,densely dotted},baseline=-20mm]
        \graphboxx{L}{0mm}{-11mm}{35mm}{26mm}{0.5mm}{-9mm}{
          \node [npattern] (x) at (-10mm,0mm) {\nodex};
          \node [npattern] (y) at (0mm,0mm) {\nodey};
          \node [npattern] (z) at (10mm,0mm) {\nodez};
          \draw [epat] (x) to node [above] {$a$} (y);
          \draw [epat] (y) to node [above] {$b$} (z);
          \context
        }
        \graphboxx{K}{36mm}{-11mm}{35mm}{26mm}{0.5mm}{-9mm}{
          \node [npattern] (x) at (-10mm,0mm) {\nodex};
          \node [npattern] (z) at (10mm,0mm) {\nodez};
          \context
        }
        \begin{scope}[opacity=1]        
        \graphboxx{R}{72mm}{-11mm}{35mm}{26mm}{0.5mm}{-9mm}{
          \node [npattern] (x) at (-10mm,0mm) {\nodex};
          \node [npattern] (y) at (0mm,0mm) {\nodey};
          \node [npattern] (z) at (10mm,0mm) {\nodez};
          \draw [epat] (x) to node [above] {$a$} (y);
          \draw [epat] (y) to node [above] {$c$} (z);
          \context
        }
        \end{scope}
      \end{tikzpicture}
}
}
	\end{center}
	\begin{center}
		$\tau = $ \scalebox{0.9}{
%
{
\newcommand{\context}{
          \node [nset] (c) at (0mm,-10mm) {\nodectx};
          \draw [eset] (c) to[bend left=15] (x);
          \draw [eset] (x) to[bend left=15] (c);
          \draw [eset] (c) to[bend left=15] (z);
          \draw [eset] (z) to[bend left=15] (c);
          \draw [eset] (c) to[thinloop=-45,distance=1.5em] (c);
          \draw [eset] (x) to[thinloop=-135,distance=1.5em] (x);
          \draw [eset] (z) to[thinloop=-45,distance=1.5em] (z);
          \draw [eset] (x.105) to ++(0mm,3.5mm) -| (z.75);
          \draw [eset] (z.105) to ++(0mm,2mm) -| (x.75);
}%
\newcommand{\nodex}{\vertex{x}{cblue!20}}%
\newcommand{\nodey}{\vertex{y}{cgreen!20}}%
\newcommand{\nodez}{\vertex{z}{corange!20}}%
\newcommand{\nodew}{\vertex{w}{cred!20}}%
\newcommand{\nodectx}{\vertex{c}{white}}%
      \begin{tikzpicture}[->,node distance=12mm,n/.style={},epattern/.style={very thick},mred/.style={cred,very thick},mblue/.style={cblue,very thick},mgreen/.style={cdarkgreen,very thick},mpurple/.style={cpurple,very thick,densely dotted},baseline=-20mm]
        \graphboxx{L}{0mm}{-11mm}{35mm}{26mm}{0.5mm}{-9mm}{
          \node [npattern] (x) at (-10mm,0mm) {\nodex};
          \node [npattern] (y) at (0mm,0mm) {\nodey};
          \node [npattern] (z) at (10mm,0mm) {\nodez};
          \draw [epat] (x) to node [above] {$c$} (y);
          \draw [epat] (y) to node [above] {$d$} (z);
          \context
        }
        \graphboxx{K}{36mm}{-11mm}{35mm}{26mm}{0.5mm}{-9mm}{
          \node [npattern] (x) at (-10mm,0mm) {\nodex};
          \node [npattern] (z) at (10mm,0mm) {\nodez};
          \context
        }
        \begin{scope}[opacity=1]        
        \graphboxx{R}{72mm}{-11mm}{35mm}{26mm}{0.5mm}{-9mm}{
          \node [npattern] (x) at (-10mm,0mm) {\nodex};
          \node [npattern] (y) at (0mm,0mm) {\nodey};
          \node [npattern] (z) at (10mm,0mm) {\nodez};
          \draw [epat] (x) to node [above] {$d$} (y);
          \draw [epat] (y) to node [above] {$b$} (z);
          \context
        }
        \end{scope}
      \end{tikzpicture}
}
}
	\end{center}
	These are graph transformation versions of the string rewrite rules $ab \to ac$ and $cd \to db$, respectively. Note that both rules specify that node $y$ is not allowed to have any incident edges other than those shown in $L$.
	
	We can use $\TT = \{ L(\rho) \}$ for our tile set, with $\ww(L(\rho)) = 1$ and $\ARC{\CC} = \MonoC{\CC}$. Then rule $\rho$ is clearly decreasing. Rule $\tau$ is non-increasing, because the creation of the $b$-edge cannot create new $L(\rho)$ subobjects, due to $y$ not having any incoming $a$-edges. So $\tau$ is non-increasing and we can remove $\rho$ by relative termination. Next, we can prove $\tau$ terminating using tile 
	{
		\hspace{-3.4mm}
		\newcommand{\nodex}{\vertex{x}{cblue!20}}
		\newcommand{\nodey}{\vertex{y}{cgreen!20}}
		\begin{tikzpicture}[->,node distance=12mm,n/.style={},baseline=-0.8ex]
			\node [npattern] (x) at (-6mm,0mm) {\nodex};
			\node [npattern] (y) at (6mm,0mm) {\nodey};
			\draw [epat] (x) to node [above] {$c$} (y);
		\end{tikzpicture}%
	} with weight $1$ and $\ARC{\CC} = \MonoC{\CC}$, i.e., by counting non-loop $c$-edges.
\end{exa}

We give one example in a hypergraph category that was originally considered by Plump~\cite[Example 6]{plump2018modular}.

\begin{exa}[Hypergraphs]\label{ex:plump:jungle:eval:hypergraph}
	
	Let $\II$ be the category freely generated by:
	\begin{center}
		\begin{tikzcd}[column sep=50, row sep=20]
			& s \arrow[d, "a_s^1" description, shift left=2] \arrow[d, "r_s" description, shift right=2] &                                \\
			+ \arrow[r, "a_+^1" description, shift right=1] \arrow[r, "a_+^2" description, shift right=6] \arrow[r, "r_+" description, shift left=3] & V                                                                                        & 0 \arrow[l, "r_0" description]
		\end{tikzcd}
	\end{center}
	We view $V$ as the type of expression roots. There are three types of hyperedges: $+$, $s$, and $0$, each associated with a root ($r_+$, $r_S$, and $r_0$, respectively). Moreover, $+$ is associated with two (ordered) arguments $a_+^1$ and $a_+^2$, and $s$ with one argument $a_s^1$ ($s$ should be thought of as representing the successor function). The functor category $\functorcat{\II}{\FinSet}$ has hypergraphs as objects and hypergraph homomorphisms as arrows. Moreover, it is a topos (Proposition~\ref{prop:finite:presheaf}).
	
	Not all objects of $\functorcat{\II}{\Set}$ would be considered well-formed under our suggested interpretation of these hypergraphs as expressions (in particular, there could be cycles and parallel hyperedges), but our termination argument below is valid even in the presence of such non-well-formed hypergraphs.
	
	The DPO rules considered by Plump~\cite[Example 6]{plump2018modular} in this category are as follows:
	\begin{center}
		$\rho = $ \scalebox{0.9}{
%
{
\newcommand{\context}{
}%
\newcommand{\nodex}{\vertex{x}{cblue!20}}%
\newcommand{\nodey}{\vertex{y}{cgreen!20}}%
\newcommand{\nodez}{\vertex{z}{corange!20}}%
\newcommand{\nodew}{\vertex{w}{cred!20}}%
\newcommand{\nodectx}{\vertex{c}{white}}%
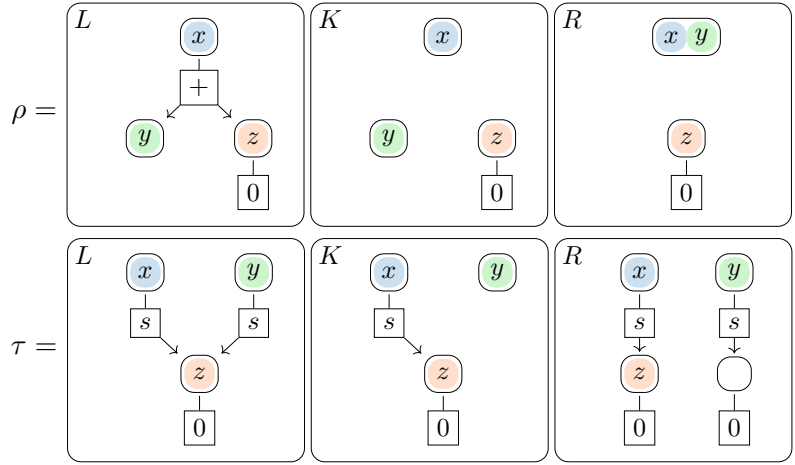
      \begin{tikzpicture}[->,node distance=12mm,n/.style={},epattern/.style={very thick},mred/.style={cred,very thick},mblue/.style={cblue,very thick},mgreen/.style={cdarkgreen,very thick},mpurple/.style={cpurple,very thick,densely dotted},baseline=-28mm]
        \graphbox{$L$}{0mm}{-11mm}{35mm}{33mm}{2mm}{-5mm}{
          \node [npattern] (x) at (0mm,0mm) {\nodex};
          \node [npattern] (y) at (-8mm,-15mm) {\nodey};
          \node [npattern] (z) at (8mm,-15mm) {\nodez};
          \node [rectangle,draw=black,fill=white] (sx) [below of=x,node distance=7.5mm] {$+$};
          \draw [-] (x) to  (sx);
          \draw [->] (sx.south west) to  (y);
          \draw [->] (sx.south east) to  (z);
          \node [rectangle,draw=black,fill=white] (0) [below of=z,node distance=8mm] {$0$}; 
          \draw [-] (z) to  (0);
          \context
        }
        \graphbox{$K$}{36mm}{-11mm}{35mm}{33mm}{2mm}{-5mm}{
          \node [npattern] (x) at (0mm,0mm) {\nodex};
          \node [npattern] (y) at (-8mm,-15mm) {\nodey};
          \node [npattern] (z) at (8mm,-15mm) {\nodez};
          \node [rectangle,draw=black,fill=white] (0) [below of=z,node distance=8mm] {$0$}; 
          \draw [-] (z) to  (0);
          \context
        }
        \begin{scope}[opacity=1]        
        \graphbox{$R$}{72mm}{-11mm}{35mm}{33mm}{2mm}{-5mm}{
          \node [npattern] (x) at (0mm,0mm) {\nodex\nodey};
          \node [npattern] (z) at (0mm,-15mm) {\nodez};
          \node [rectangle,draw=black,fill=white] (0) [below of=z,node distance=8mm] {$0$}; 
          \draw [-] (z) to  (0);
          \context
        }
        \end{scope}
      \end{tikzpicture}
}
}
	\end{center}
	\begin{center}
		$\tau = $ \scalebox{0.9}{
%
{
\newcommand{\context}{
}%
\newcommand{\nodex}{\vertex{x}{cblue!20}}%
\newcommand{\nodey}{\vertex{y}{cgreen!20}}%
\newcommand{\nodez}{\vertex{z}{corange!20}}%
\newcommand{\nodew}{\vertex{w}{cred!20}}%
\newcommand{\nodectx}{\vertex{c}{white}}%
      \begin{tikzpicture}[->,node distance=12mm,n/.style={},epattern/.style={very thick},mred/.style={cred,very thick},mblue/.style={cblue,very thick},mgreen/.style={cdarkgreen,very thick},mpurple/.style={cpurple,very thick,densely dotted},baseline=-28mm]
        \graphbox{$L$}{0mm}{-11mm}{35mm}{33mm}{2mm}{-5mm}{
          \node [npattern] (x) at (-8mm,0mm) {\nodex};
          \node [npattern] (y) at (8mm,0mm) {\nodey};
          \node [npattern] (z) at (0mm,-15mm) {\nodez};
          \node [rectangle,draw=black,fill=white] (sx) [below of=x,node distance=7.5mm] {$s$};
          \draw [-] (x) to  (sx);
          \draw [->] (sx.south east) to  (z);
          \node [rectangle,draw=black,fill=white] (sy) [below of=y,node distance=7.5mm] {$s$};
          \draw [-] (y) to  (sy);
          \draw [->] (sy.south west) to  (z);
          \node [rectangle,draw=black,fill=white] (0) [below of=z,node distance=8mm] {$0$};
          \draw [-] (z) to  (0);
          \context
        }
        \graphbox{$K$}{36mm}{-11mm}{35mm}{33mm}{2mm}{-5mm}{
          \node [npattern] (x) at (-8mm,0mm) {\nodex};
          \node [npattern] (y) at (8mm,0mm) {\nodey};
          \node [npattern] (z) at (0mm,-15mm) {\nodez};
          \node [rectangle,draw=black,fill=white] (sx) [below of=x,node distance=7.5mm] {$s$};
          \draw [-] (x) to  (sx);
          \draw [->] (sx.south east) to  (z);
          \node [rectangle,draw=black,fill=white] (0) [below of=z,node distance=8mm] {$0$};
          \draw [-] (z) to  (0);
          \context
        }
        \begin{scope}[opacity=1]        
        \graphbox{$R$}{72mm}{-11mm}{35mm}{33mm}{2mm}{-5mm}{
          \node [npattern] (x) at (-7mm,0mm) {\nodex};
          \node [npattern] (y) at (7mm,0mm) {\nodey};
          \node [npattern] (z) at (-7mm,-15mm) {\nodez};
          \node [npattern,minimum size=5mm] (n) at (7mm,-15mm) {};
          \node [rectangle,draw=black,fill=white] (sx) [below of=x,node distance=7.5mm] {$s$};
          \draw [-] (x) to  (sx);
          \draw [->] (sx) to  (z);
          \node [rectangle,draw=black,fill=white] (0) [below of=z,node distance=8mm] {$0$};
          \draw [-] (z) to  (0);
          \node [rectangle,draw=black,fill=white] (sy) [below of=y,node distance=7.5mm] {$s$};
          \draw [-] (y) to  (sy);
          \draw [->] (sy) to  (n);
          \node [rectangle,draw=black,fill=white] (0b) [below of=n,node distance=8mm] {$0$};
          \draw [-] (n) to  (0b);
          \context
        }
        \end{scope}
      \end{tikzpicture}
}
}
	\end{center}
	
	They can be encoded into \pbpostrong{} using the standard encoding~\cite[Definition~71]{overbeek2023quasitoposes}. That is, $t_X = (\eta_X : X \regmono T(X))$ for $X \in \{L, K \}$ (with $(\eta, T)$ the partial map classifier of the category), and $R'$ is obtained by pushing out $t_K$ along $r$. Note that because vertex $w$ is fresh in $R$, it does not have any connections in $R'$ other than those shown in $R$.
	
	To prove termination, we first use tile
	\begin{center}
		\scalebox{0.9}{
%
{
\newcommand{\context}{
}%
\newcommand{\nodex}{\vertex{x}{cblue!20}}%
\newcommand{\nodey}{\vertex{y}{cgreen!20}}%
\newcommand{\nodez}{\vertex{z}{corange!20}}%
\newcommand{\nodew}{\vertex{w}{cred!20}}%
\newcommand{\nodectx}{\vertex{c}{white}}%
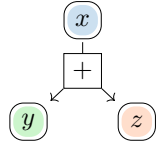
      \begin{tikzpicture}[->,node distance=12mm,n/.style={},epattern/.style={very thick},mred/.style={cred,very thick},mblue/.style={cblue,very thick},mgreen/.style={cdarkgreen,very thick},mpurple/.style={cpurple,very thick,densely dotted}]
          \node [npattern] (x) at (0mm,0mm) {\nodex};
          \node [npattern] (y) at (-8mm,-15mm) {\nodey};
          \node [npattern] (z) at (8mm,-15mm) {\nodez};
          \node [rectangle,draw=black,fill=white] (sx) [below of=x,node distance=7.5mm] {$+$};
          \draw [-] (x) to  (sx);
          \draw [->] (sx.south west) to  (y);
          \draw [->] (sx.south east) to  (z);
      \end{tikzpicture}
}
}
	\end{center}
	with weight 1 and $\ARC{\CC} = \MonoC{\CC}$. This proves $\rho$ decreasing and $\tau$ non-increasing, so that we can remove rule $\rho$.
	
	We then remove the second rule using tile
	\begin{center}
		\scalebox{0.9}{
%
{
\newcommand{\context}{
}%
\newcommand{\nodex}{\vertex{x}{cblue!20}}%
\newcommand{\nodey}{\vertex{y}{cgreen!20}}%
\newcommand{\nodez}{\vertex{z}{corange!20}}%
\newcommand{\nodew}{\vertex{w}{cred!20}}%
\newcommand{\nodectx}{\vertex{c}{white}}%
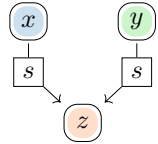
      \begin{tikzpicture}[->,node distance=12mm,n/.style={},epattern/.style={very thick},mred/.style={cred,very thick},mblue/.style={cblue,very thick},mgreen/.style={cdarkgreen,very thick},mpurple/.style={cpurple,very thick,densely dotted}]
          \node [npattern] (x) at (-8mm,0mm) {\nodex};
          \node [npattern] (y) at (8mm,0mm) {\nodey};
          \node [npattern] (z) at (0mm,-15mm) {\nodez};
          \node [rectangle,draw=black,fill=white] (sx) [below of=x,node distance=7.5mm] {$s$};
          \draw [-] (x) to  (sx);
          \draw [->] (sx.south east) to  (z);
          \node [rectangle,draw=black,fill=white] (sy) [below of=y,node distance=7.5mm] {$s$};
          \draw [-] (y) to  (sy);
          \draw [->] (sy.south west) to  (z);
      \end{tikzpicture}
}
}
	\end{center}
	with weight $1$ and $\ARC{\CC} = \MonoC{\CC}$.
	Observe that a node in the host graph with $n$ incoming $s$-edges allows for $n \cdot (n-1)$ embeddings of this tile: this observation relates our argument to that of Plump~{\cite[Example 6]{plump2018modular}}.
\end{exa}

\section{Related Work}
\label{section:related:work}

We consider two closely related approaches by Bruggink et al.~\cite{bruggink2014, bruggink2015} and the approach by Endrullis and Overbeek~\cite{gwtg2024} to be the most relevant to our method.\footnote{The recently updated arXiv version~\cite{bruggink2023proving} of \cite{bruggink2015} corrects some errors in the theory.}
These approaches use weighted type graphs $T$ to measure graphs $G$ by means of counting weighted morphisms $G \to T$ (instead of weighted morphisms $T' \to G$ for tiles $T'$). So the general idea is dual to the one in this paper. Moreover, to our knowledge, these approaches are the only systematic termination methods in the algebraic tradition based on decreasing interpretations.

The methods by Bruggink et al.\ are defined for DPO in the category of edge-labeled multigraphs. The first approach~\cite{bruggink2014}
requires that $l$ and $r$ of DPO rules $L \stackrel{l}{\leftmono} K \stackrel{r}{\mono} R$, and matches $m : L \mono G_L$, are monic. The second approach~\cite{bruggink2015} has no such restrictions.

Because our method is applicable in a much broader setting, our method will prove rules terminating that are outside the scope of the methods by Bruggink et al. Nonetheless, it is interesting to ask how the approaches relate in settings where they are all defined.

On the one hand, although Examples~5 and~6 of~\cite{bruggink2015} are within the scope of our method, our method cannot prove them terminating. The intuitive reason is that the examples terminate because of global properties, rather than local ones. On the other hand, Example~\ref{example:DPO:rule:comparison} below defines a DPO rule that falls inside the scope of all three methods, and only our method can prove it terminating. In conclusion, within the restricted setting, the methods are incomparable.

\begin{exa}
	\label{example:DPO:rule:comparison}
	Consider the following DPO rule $\rho$ in category $\FinGraph$:
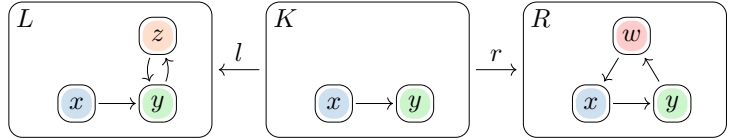
	\begin{center}
		\scalebox{0.9}{
%
{
\newcommand{\nodex}{\vertex{x}{cblue!20}}%
\newcommand{\nodey}{\vertex{y}{cgreen!20}}%
\newcommand{\nodez}{\vertex{z}{corange!20}}%
\newcommand{\nodew}{\vertex{w}{cred!20}}%
      \begin{tikzpicture}[->,node distance=12mm,n/.style={},epattern/.style={very thick},mred/.style={cred,very thick},mblue/.style={cblue,very thick},mgreen/.style={cdarkgreen,very thick},mpurple/.style={cpurple,very thick,densely dotted}]
        \graphbox{$L$}{0mm}{0mm}{30mm}{20mm}{1mm}{-5mm}{
          \node [npattern] (x) at (-6mm,-10mm) {\nodex};
          \node [npattern] (y) at (6mm,-10mm) {\nodey};
          \node [npattern] (z) at (6mm,0mm) {\nodez};
          \draw [->] (x) to (y);
          \draw [->] (y) to[bend right=20] (z);
          \draw [->] (z) to[bend right=20] (y);
        }
        \graphbox{$K$}{38mm}{0mm}{30mm}{20mm}{1mm}{-5mm}{
          \node [npattern] (x) at (-6mm,-10mm) {\nodex};
          \node [npattern] (y) at (6mm,-10mm) {\nodey};
          \draw [->] (x) to (y);
        }
        \graphbox{$R$}{76mm}{0mm}{30mm}{20mm}{1mm}{-5mm}{
          \node [npattern] (x) at (-6mm,-10mm) {\nodex};
          \node [npattern] (y) at (6mm,-10mm) {\nodey};
          \node [npattern] (z) at (0mm,0mm) {\nodew};
          \draw [->] (x) to (y);
          \draw [->] (y) to[bend left=0] (z);
          \draw [->] (z) to[bend left=0] (x);
        }
        \draw [->] (37mm,-10mm) to node [above] {$l$} ++(-6mm,0mm);
        \draw [->] (69mm,-10mm) to node [above] {$r$} ++(6mm,0mm);
      \end{tikzpicture}
}
}  
	\end{center}
	and assume matching is required to be monic. This requirement is often used in practice, because monic matching increases the expressiveness of DPO~\cite{habel2001doublerevisited}.
	
	The approach in~\cite{bruggink2014} cannot prove $\rho$ terminating. For establishing termination (on all graphs), the weighted type graph $T$ has to contain a node with a loop (called a \emph{flower node}). The flower node ensures that every graph $G$ can be mapped into $T$. Then, in particular, the technique requires a weight decrease (from $L$ to $R$) for the case that the interface $K$ is mapped onto the flower node. However, this makes $L$ and $R$ indistinguishable for the technique in~\cite{bruggink2014}.
	
	Although matches are required to be monic, the method of~\cite{bruggink2015} overapproximates for unrestricted matches by design.
	Observe that if matching is not monic, then graph $L$ of $\rho$, but with $x$ and $y$ identified, rewrites to itself, meaning $\rho$ is not terminating. As a consequence, the overapproximation of~\cite{bruggink2015} causes it to fail in proving $\rho$ terminating for the monic matching setting. (For the same reason, the method of~\cite{bruggink2015} fails on the simpler top span of Example~\ref{example:fold:edge}, which is a DPO rule, for the monic matching setting.)
	
	Rule $\rho$ can be proven terminating with our method as follows. Encode $\rho$ into \pbpostrong{} using the standard encoding~\cite[Definition~71]{overbeek2023quasitoposes}.). The resulting rule and a termination proof were given in Example~\ref{example:dpo:as:pbpo} above.
\end{exa}

Additional examples by Bruggink et al.\ that our method can prove are Example~4 of \cite{bruggink2014} ($=$ Example~2 of \cite{bruggink2015}), and Example 4 of \cite{bruggink2015}. Additional examples that our method cannot prove are Example~1 and the example of Section~3.5 of~\cite{bruggink2014}. However, unlike the earlier referenced Examples~5 and~6 of~\cite{bruggink2015}, these examples are in reach if our morphism counting technique can take into account antipatterns (Remark~\ref{remark:antipatterns} below), because they terminate because of local properties.

\begin{rem}[Antipatterns]
	\label{remark:antipatterns}
	A rule that matches an isolated node, and adds a loop, cannot be proven terminating with our method. For this, one must be able to count nodes \emph{without} loops (an \emph{antipattern}), which is currently unsupported. We believe that extending our method with support for such antipatterns is a natural first step for significantly strengthening it.
\end{rem}

In \cite{gwtg2024}, we generalize and strengthen the approach of~\cite{bruggink2015} to (a) be applicable in a wider range of categories and (b) be applicable when termination depends on matching restrictions. This enhanced weighted type graph technique can prove termination of Example~\ref{example:DPO:rule:comparison}. However, it still fails for other systems which can be proven terminating by the approach presented in the current paper (for instance Example~6 of \cite{plump2018modular}). Thus, even in settings where both methods are defined, the methods are strictly incomparable.

We discuss some additional related work.
An early systematic termination criterion for hypergraph rewriting with DPO, due to Plump, is based on the concept of forward closures~\cite{plump1995on}. Both of the examples proven terminating with forward closures, Example~3.8 (Example~\ref{example:plump:string} above) and Example~4.1 of~\cite{plump1995on}, can be handled with our method.

More recently, Plump formulated a modularity criterion for hypergraph rewriting using DPO~\cite{plump2018modular}: the union of two terminating systems is terminating if there are no sequential critical pairs. Of this paper, our method can prove three out of four examples: Examples~3 ($=$ Example 3.8 of~\cite{plump1995on}),~5 and~6. The modeling of Example 6 is Example~\ref{ex:plump:jungle:eval:hypergraph} above.
Our method cannot prove Example~4 ($=$ the already discussed Example 5 of \cite{bruggink2014}). It would be interesting to assess the strength of the modularity criterion (especially if generalized to \pbpostrong{}) combined with our method.

Bruggink et al.\ have shown that string rewriting rules are terminating on graphs iff they are terminating on cycles~\cite{bruggink2014}, making cycle rewriting techniques~\cite{zantema2014termination, sabel2017termination} applicable to graph transformation systems consisting of string rewrite rules. Similarly, in another paper~\cite{overbeek2021from} we have shown that particular \pbpostrong{} encodings of linear term rewrite rules are terminating on graphs iff they are terminating on terms.

There also exist a variety of approaches that generalize TRS methods (such as simplification orderings) to acyclic term graphs~\cite{plump1999term,plump97simplification,moser2016kruskal} and drags~\cite{dershowitz2019drags, dershowitz2023drag} (that possibly contain cycles)~\cite{dershowitz2018graph}.

\begin{rem}[Relating Subgraph Counting with TRS Termination Methods]
	An interesting question is how the technique presented in this chapter relates to the many techniques available for term rewriting systems. Zantema~\cite[Chapter 6]{terese} has roughly categorized the available termination techniques for TRSs as follows:
	\begin{itemize}
		\item syntactical methods, which inductively define a well-founded order directly on the terms;
		\item transformational methods, which define non-termination preserving transformations $\Phi$ on rewrite systems $(\Sigma, R)$, such that termination of $(\Sigma, R)$ can be proven by applying other techniques to $\Phi((\Sigma, R))$; and
		\item semantical methods, which interpret terms into some well-founded order using weight functions.
	\end{itemize}
	One of the most well-known syntactical methods for TRSs is the recursive path order (RPO) by Dershowitz~\cite{dershowitz82orderings}. RPO has been adapted for acyclic term graphs and~\cite{plump1999term,plump97simplification,moser2016kruskal} and drags~\cite{dershowitz2018graph}.
	
	Well-known transformational methods for TRSs include dependency pairs~\cite{arts2000termination} and semantic labeling~\cite{zantema1995termination}. As far as we know, no versions for graphs have been proposed.
	
	For the semantical methods, finally, a fundamental result is that a TRS is terminating iff it admits a compatible well-founded monotone algebra~\cite[Theorem 6.2.2.]{terese}. This of course also implies that finding such an algebra is hard. For TRSs, there exist well-known heuristics for finding polynomial interpretations~\cite{cherifa1987termination} and matrix interpretations~\cite{hofbauer2006termination, endrullis2006matrix}. The weighted type graph approach by Bruggink et al.~\cite{bruggink2014, bruggink2015}, discussed above, essentially corresponds to the matrix interpretation method.
	
	Our weighted subgraph counting approach is distinctly semantical, but as far as we can see, it does not correspond directly to an existing method for term rewriting systems. Because an analogous approach would not exploit the inductive term structure, we also do not believe that it would be very powerful in the TRS setting.
\end{rem}

\section{Implementation}
\label{section:implementation}

We have implemented our termination method using Scala 3~\cite{odersky2021programming}. The software package is called \texttt{graphTT}~\cite{overbeek2024graphtt}, and it includes a REPL (read-eval print loop) for exploring relative termination proofs iteratively, where in each step detailed feedback about the proof process is provided (Section~\ref{sec:repl}). Our algorithm is implemented generically for general categories, and one category is included (together with parsing methods for input files): $\FinGraph$, the category of finite edge- and vertex-labeled directed multigraphs. After describing the REPL, we give a brief description of the implementation (Section~\ref{sec:implementation:details}), which can be viewed as an example of computational (applied) category theory (see Remark~\ref{rem:computational:applied:category:theory} below).

\begin{rem}[Related Work]
	\label{rem:computational:applied:category:theory}
	\emph{Computational (applied) category theory}, initiated by Rydeheard and Burstall~\cite{rydeheard1988computational} in 1988, bridges the gap between programming and category theory. More specifically, Rydeheard and Burstall implemented a large body of category theory in ML.
	For algebraic graph transformation in particular, Minas and Schneider~\cite{minas2010graph} proposed a Java implementation in 2010. In more recent years, Brown et al.~\cite{brown2023computationalJLAMP} have implemented algebraic graph rewriting in Julia, as part of the broader AlgebraicJulia\footnote{%
		\url{https://www.algebraicjulia.org}
	}~\cite{halter2020compositional, patterson2021categorical} applied category theory project.
	The work by Brown et al.\ is especially interesting because the code is both general and performant.
\end{rem}

\subsection{REPL}
\label{sec:repl}

We describe the REPL used for exploring relative termination proofs by running through a simple example session.

When the REPL is invoked, it starts in system selection mode. One of the available commands is \texttt{help}, which shows the available commands. The output is as follows:
\begin{Verbatim}[xleftmargin=.5cm]
=== GraphTT REPL ===
>> You are in system selection mode.
>> Type 'help' to view the available commands.
graphTT> help
>> Available commands:
select [n]  : select system n for termination proving
inspect [n] : inspect option (system/tile) n in detail
systems     : list the available systems
help        : print all available commands
exit        : exit the program
\end{Verbatim}
The list of available systems at the time of writing is as follows:
\begin{Verbatim}[xleftmargin=.5cm]
graphTT> systems
>> The following systems were loaded:
(0) multiset_as_graph
(1) delete_loop_and_nonloop
(2) unfold_to_triangle
(3) folding_an_edge
(4) duplicating_bipartite_components
(5) generalized_multiset_as_graph
\end{Verbatim}
When the REPL is started, the available systems are read from relative directory \path{./src/main/resources/labeled/systems}. The included systems correspond to examples provided in this paper. The correspondences are indicated in the included files using comments (\texttt{/*...*/}). For example, the contents of file \texttt{folding\_an\_edge.pbpop} are as follows:
\begin{Verbatim}[xleftmargin=.5cm]
graphTT> inspect 2
SYSTEM: folding_an_edge
/* Example 5.2 of 
Termination of Graph Transformation Systems 
Using Weighted Subgraph Counting
*/

=== rho ===
L  { x:0 -P:0-> y:0 }
L' { x:0 -P:0-> y:0     
	x:0 -XX:0-> x:0
	y:0 -YY:0-> y:0
	x:0 -XY:0-> <-YX:0- y:0 
	x:0 -XC:0-> <-CX:0- c:0
	c:0 -CY:0-> <-YC:0- y:0 
	c:0 -CC:0-> c:0 }
K  { x:0 -P:0-> y:0 }
K' { x:0 -P:0-> y:0     
	x:0 -XX:0-> x:0
	y:0 -YY:0-> y:0
	x:0 -XY:0-> <-YX:0- y:0 
	x:0 -XC:0-> <-CX:0- c:0
	c:0 -CY:0-> <-YC:0- y:0 
	c:0 -CC:0-> c:0 }
R  { x.y:0 -P:0-> x.y:0 }
\end{Verbatim}
The file contains the system of Example~\ref{example:fold:edge}, which consists of a single \pbpostrong{} rule $\rho$ over unlabeled graphs. Elements and their labels are written as \texttt{id:label}. By convention, we use lowercase for vertex identities, uppercase for edge identities, and we label unlabeled elements with \texttt{0}. Edges require two endpoints and an indication of direction, and vertices without incident edges can be listed separately. The rule morphisms are implicitly defined and follow the convention described in Notation~\ref{notation:visual}. For instance, vertices \texttt{x} and \texttt{y} of $K$ are mapped onto vertex \texttt{x.y} of $R$: so \texttt{.} is used to compose identities.

Let us prove system \texttt{folding\_an\_edge} terminating. We do this by selecting it, after which the REPL enters into proof mode:
\begin{Verbatim}[xleftmargin=.5cm]
graphTT> select 3
>> Entering proof mode for system folding_an_edge.
>> The system consists of 1 rule.
>> The system is as follows:

[... printed system as above ...]
\end{Verbatim}
In proof mode, the available commands are as follows:
\begin{Verbatim}[xleftmargin=.5cm]
graphTT> help
>> Available commands:
use [i w c]+ :
use tile i with weight w, and count morphisms of class c, where:
- i and w are integers (and w is positive), and 
- c is a character (r: regular monos, m: monos, h: homomorphisms)
multiple tiles can be specified. 
for example, 'use 3 4 h 5 9 r' uses:
- tile 3 with weight 4 (counting homomorphisms), and
- tile 5 with weight 9 (counting regular monos)
inspect [n] : inspect option (system/tile) n in detail
back : return to system selection mode
help : print all available commands
tiles : list the available tiles
exit : exit the program
\end{Verbatim}
Tiles are loaded from relative directory \path{./src/main/resources/labeled/tiles}. At the time of writing, these tiles are available:
\begin{Verbatim}[xleftmargin=.5cm]
>> The following tiles were loaded:
(0) two_opposing_edges
(1) single_loop
(2) single_node
(3) single_nonloop_edge
(4) a_loop
(5) b_loop
\end{Verbatim}
As Example~\ref{example:fold:edge} suggests, we just need a single unlabeled non-loop edge, which is provided:
\begin{Verbatim}[xleftmargin=.5cm]
graphTT> inspect 3
TILE: single_nonloop_edge
x:0 -XY:0-> y:0
\end{Verbatim}
Example~\ref{example:fold:edge} also suggests that a weight assignment of 1 is sufficient, and that we should count (regular) monos. Note that our implementation allows choosing a class $\ARC{\CC}$ on a tile-by-tile basis (see Remark~\ref{remark:choosing:class:per:tile}). As the description for \texttt{use} suggests in the \texttt{help} output above, we can use tile 3 with weight 1 (and counting monos) by writing \texttt{use 3 1 m}.

Using tile 3 with weight 1 produces the following output (commentary on the output follows afterwards):
\begin{Verbatim}[xleftmargin=.5cm]
graphTT> use 3 1 m

=============== SYSTEM TERMINATION REPORT ===============
---------------          SUMMARY          ---------------

The system has 1 rules, named: rho
Was the sliding successful for every rule? yes
Provably decreasing rules: rho
Provably nonincreasing (but not provably decreasing) rules: 
Possibly increasing rules: 
The pruned system contains rules: 

The pruned system is empty, so the system is TERMINATING.

---------------   DETAILED RULE REPORTS   ---------------

>>>>>>>>>>>>>>> rule rho <<<<<<<<<<<<<<<
Summary:
- The sliding is SUCCESSFUL.
- The weight of Delta is 1.
- The weight of R is 0.
- Conclusion: the rule is PROVABLY DECREASING.

The details per tile for this rule now follow.

~~~ Tile single_nonloop_edge with weight 1, counting MONOS only
x:0 -XY:0-> y:0

- The tiling of R has size:             0
- Giving a weight of:                   0 * 1 = 0
- A largest valid tiling of L has size: 1
- Giving a weight of:                   1 * 1 = 1

Slide data:

# morphisms into R':           10
# of which valid:              5
# iso in R:                    0
# noniso in R:                 5
# number of ways to slide:     1

>> The pruned system is empty!
>> You have proven system folding_an_edge terminating.
>> Returning to system selection mode.
\end{Verbatim}
As can be seen:
\begin{enumerate}
	\item a summary is given on the system level, indicating whether the sliding was successful for every rule, and partitioning the rules according to whether they are provably decreasing, provably nonincreasing (but not provably decreasing), or neither (i.e., possibly increasing);
	\item a summary is given for every rule $\rho$ of the system, indicating, among others, the largest possible weight for $\Delta_\rho$ (Theorem~\ref{thm:termination:by:element:counting}), and the weight of $R$;
	\item for every rule $\rho$ and tile $T$, detailed data is given for the slide analysis. Here, 
	\begin{itemize}
		\item \texttt{morphisms into R'} indicates $\card{\homset{T}{R'}}$;
		\item \texttt{of which valid} indicates $\card{S}$ for $S = \{ t\in \homset{T}{R'} \mid \pb{t}{t_R} \in \ARC{\CC} \}$;
		\item \texttt{iso in R} indicates $\isosub{\homset{T}{R'}}{t_R} \cap S$;
		\item \texttt{noniso in R} indicates $\nisosub{\homset{T}{R'}}{t_R} \cap S$; and
		\item \texttt{ways to slide} indicates the number of possible slid tile sets into $L'$. The number of choices for right inverses do not matter here: it is the number of outcomes.
	\end{itemize}
\end{enumerate}
Moreover, the system is automatically pruned if possible. In this example the pruned system is empty, meaning the system is terminating. If the system is not empty, the process repeats, and the user can select a new set of tiles for the pruned system.

\subsection{Implementation Details}
\label{sec:implementation:details}

The implementation consists of seven packages, named \texttt{categorytheory}, \texttt{labeledgraph}, \texttt{parsing}, \texttt{repl}, \texttt{rewriting}, \texttt{termination}, and \texttt{util}. We comment on packages  \texttt{categorytheory}, \texttt{labeledgraph}, \texttt{rewriting}, and \texttt{termination}.

\subsubsection{Package \texttt{categorytheory}}

\newcommand{\trait}{\scala{trait}}
\newcommand{\interface}{\java{interface}}

Package \texttt{categorytheory} contains general purpose categorical definitions, organized in a highly modular fashion. The library is designed to be reusable for other purposes. Some of the more complex definitions include the notions of a presheaf category and a topos.

Most of the definitions in package \texttt{categorytheory} are defined as \emph{traits}. A Scala trait can be described as a more powerful version of Java's \emph{interface}. A trait can take type parameters and define (abstract or implemented) members. It can also mix in other traits, which is Scala's way of implementing multiple inheritance. The definition of a topos is given in Listing~\ref{code:topos}, for instance, where type parameters \scala{O} and \scala{A} represent the collections of objects and arrows, respectively. All members of \scala{Topos} are inherited from the \scala{Quasitopos} trait, and three default implementations are provided. The \scala{RmAdhesive} trait is empty: mixing it in constitutes a promise that the category is rm-adhesive.

\begin{code}
	\caption{\texttt{Topos.scala}}
	\label{code:topos}
	\begin{minted}{scala}
		trait Topos[O, A] extends Quasitopos[O, A] with RmAdhesive:
			final def isRegularEpic(f: A): Boolean = isEpic(f)
			final def isRegularMonic(f: A): Boolean = isMonic(f)
			final def regularEpiMonoFactorization(f: A): (A, A) =  
				epiRegularMonoFactorization(f)
	\end{minted}
\end{code}

We have also defined a shorthand type name \scala{TerminationCategory} for the intersection of the types listed in Listing~\ref{code:termination:category}. Our termination method implementation expects that the ambient category is a \scala{TerminationCategory}. Note that this intersection is strictly stronger than the requirements listed in the paper: it could be generalized in the future, and the algorithm made more fine-grained, if the use case arises.

\begin{code}
	\caption{\texttt{TerminationCategory.scala}}
	\label{code:termination:category}
	\begin{minted}{scala}
		type TerminationCategory[O, A] = Quasitopos[O, A] & RmAdhesive
	\end{minted}
\end{code}

\subsection{Package \texttt{labeledgraph}}

Package \texttt{labeledgraph} defines labeled graphs and graph homomorphisms as defined in Definition~\ref{def:graph}. It also defines \scala{LabeledGraphCategory} as an instance of the \scala{Topos} trait. Enumerations of homomorphisms are not naively implemented, but also not fully optimized. For our present purposes, this is sufficient.

\subsection{Package \texttt{rewriting}}

Package \texttt{rewriting} contains, among others, the definitions of a \pbpostrong{} rule and step. For computing the step, an immediate question is how morphism $u : K \to G_K$ is to obtained (Definition~\ref{def:pbpostrong:rewrite:step}). The lemma below demonstrates a few things of interest:
\begin{enumerate}
	\item the pullback square ${t_L \circ l} = {l' \circ t_K}$ of the rewrite rule does not need to be used in the construction;
	\item the explicit construction of $u$ can be bypassed; and
	\item no search is needed (for universal morphisms or factorizations): (a) composition of morphisms, (b) construction of pullbacks, (c) construction of pushouts, and (d) verifying that a morphism is an isomorphism are the only required operations.
\end{enumerate}

\begin{lem}[Computing a \pbpostrong{} Rewrite Step]
	\label{lemma:computing:a:pbpoplus:step}
	Let a rule $\rho$ and morphism $\alpha : G_L \to L'$ be given. By using the constructions outlined in the diagram:
	\begin{equation}
		\label{eq:computing:a:rewrite:step:diagram}
		\begin{tikzcd}[column sep=30]
			X\arrow[r, "n" description] \arrow[rd, "\mathrm{PB}(1)", phantom] \arrow[d, "j" description] & G_L \arrow[d, "\alpha" description] & G_K \arrow[l, "g_L" description] \arrow[ld, "\mathrm{PB}(2)", phantom] \arrow[d, "u'" description] & Y\arrow[l, "v" description] \arrow[d, "i" description] \arrow[ld, "\mathrm{PB}(3)", phantom] \arrow[r, "v" description] \arrow[rdd, "\mathrm{PO}(4)", phantom] & G_K \arrow[dd, "g_R" description] \\
			L \arrow[r, "t_L" description]                                                                  & L'                                  & K' \arrow[l, "l'" description]                                                                     & K \arrow[l, "t_K" description] \arrow[d, "r" description]                                                                                                                           &                                   \\
			&                                     &                                                                                                    & R \arrow[r, "w" description]                                                                                                                                                        & G_R                              
		\end{tikzcd}
	\end{equation}
	and by additionally verifying that $j$ is an isomorphism, the result of the \pbpostrong{} rewrite step induced by $\alpha$ is computed.
	Moreover, a \pbpostrong{} rewrite step diagram as given in Definition~\ref{def:pbpostrong:rewrite:step} can be reconstructed by defining the missing morphisms $m: L \to G_L$ and $u : K \to G_K$ as follows:
	\begin{itemize}
		\item ${m} = {n \circ j^{-1}}$; and
		\item ${u} = {v \circ \inverse{i}}$.
	\end{itemize}
\end{lem}
\begin{proof}
	If $j$ is an isomorphism, it is easy to verify that $\alpha$ defines a strong match for match morphism ${m = {n \circ j^{-1}}}: L \to G_L$. Moreover, morphism $i$ is known to be an isomorphism~\cite[Lemma 15]{overbeek2023quasitoposes}. That the correct pushout is computed in the last step can then be observed as follows.  First, using $\mathrm{PB}(3)$ we have $t_K = {t_K \circ i \circ \inverse{i}} = {u' \circ v \circ  \inverse{i}}$, which means that $v \circ \inverse{i}$ is the unique $u : K \to G_K$ for which $t_K = {u'\circ u}$~\cite[Lemma 15]{overbeek2023quasitoposes}. Then in diagram
	\begin{center}
		\begin{tikzcd}[column sep=40]
			K \arrow[r, "\inverse{i}" description, two heads, tail] \arrow[rr, "u" description, bend left=29] \arrow[dd, "r" description] \arrow[rdd, "\mathrm{PO}", phantom] & Y \arrow[r, "v" description] \arrow[d, "i" description, two heads, tail] \arrow[rdd, "\mathrm{PO}", phantom] & G_K \arrow[dd, "g_R" description] \\
			& K \arrow[d, "r" description]                                                                                     &                                   \\
			R \arrow[r, equals]                                                                                                                                  & R \arrow[r, "w" description]                                                                                     & G_R                              
		\end{tikzcd}
	\end{center}
	it is easy to verify that the left square is a pushout. By the dual of the pullback lemma, the pushout of the right square is also the pushout of the outer square.
\end{proof}

The implementation of the construction described in Lemma~\ref{lemma:computing:a:pbpoplus:step}, utilizing the library of \texttt{categorytheory} library, is given in Listing~\ref{code:pbpostrong:step}. Given an adherence $\alpha: G_L \to L'$, it computes the match $m : L \to G_L$ and the result of the rewrite step $G_R$. Functions \scala{maybePullback} and \scala{maybePushout} are \scala{Option} types: they can return \scala{None} in categories where not all pullbacks and pushouts are defined. 

\begin{code}
	\caption{\texttt{PbpoPlusRule.scala}}
	\label{code:pbpostrong:step}
	\begin{minted}[xleftmargin=0.028\textwidth]{scala}
		def adherenceInducedStep(alpha: A): Option[(O, A)] =
			require(codomain(alpha) == L1, "the codomain of alpha needs to be L'")
			for {
				pb0 <- maybePullback(Cospan(tL, alpha))
				if isIso(pb0.left)
				m = pb0.right o inverseOf(pb0.left).get
				pb1 <- maybePullback(Cospan(alpha, l1))
				u1 = pb1.right
				pb2 <- maybePullback(Cospan(u1, tK))
				u = pb2.left
				iso = pb2.right
				po <- maybePushout(Span(r o iso, u))
				gR = po.sink
			} yield (gR, m)
	\end{minted}
\end{code}

\begin{rem}[Decidability of the Rewrite Relation]
	Given a \pbpostrong{} rule $\rho$ in an arbitrary category $\CC$, from Lemma~\ref{lemma:computing:a:pbpoplus:step} it follows that the following conditions are sufficient for deciding whether $\pbpostrongsteprule{A}{B}{\rho}$ for objects $A, B \in \obj{\CC}$:
	\begin{enumerate}
		\item for all $X, Y \in \obj{\CC}$, the hom-set $\homset{X}{Y}$ is finite; and
		\item pushouts, pullbacks, and morphism compositions are computable.
	\end{enumerate}
	More specifically, one can use the construction given in Lemma~\ref{lemma:computing:a:pbpoplus:step}, where the conditions ensure that
	\begin{itemize}
		\item all adherence morphism candidates $\alpha \in \homset{A}{L'}$ can be exhaustively considered;
		\item verifying whether a morphism $f : X \to Y$ is an isomorphism reduces to a finite search for a morphism $f' \in \homset{Y}{X}$ that satisfies ${f \circ f'} = \id{Y}$ and ${f' \circ f} = \id{X}$;
		\item verifying whether $B$ qualifies as a pushout object in PO(4) of Diagram~\eqref{eq:computing:a:rewrite:step:diagram} reduces to a finite search for an isomorphism in $\homset{G_R}{B}$.
	\end{itemize}
	
	For the category of finite sets $\FinSet$, the sufficient conditions hold. The same is true for functor categories $\functorcat{I}{\FinSet}$ for which the index category $I$ is finite, which includes many graph categories of interest. This is because here the constructions are computed pointwise in $\FinSet$.
\end{rem}

\subsection{Package \texttt{termination}}

Package \texttt{termination}, finally, contains the implementation of the termination method, currently implemented at the level of rm-adhesive quasitoposes. Listing~\ref{code:termination} contains a code fragment that illustrates how the implementation utilizes the \texttt{categorytheory} library. When given a morphism 
$f : T \to R'$ (parameter \texttt{tiling}) rule $\rho$ (parameter \texttt{rule}), it computes set $\{ l' \circ \pb{f}{r'} \circ g \mid \text{$g$ is a right inverse for $\pb{r'}{f}$} \}$. Then, based on the selected counting class (either $\HomC{\CC}$, $\MonoC{\CC}$, or $\RegC{\CC}$, represented by parameter \texttt{countingClass}), it verifies part of the conditions of Theorem~\ref{thm:xi:is:injection}, also utilizing Proposition~\ref{prop:preserving:a:factorization:and:factorization:systems}.

\begin{code}
	\caption{\texttt{Termination.scala}}
	\label{code:termination}
	\begin{minted}[xleftmargin=0.005\textwidth]{scala}
		def slideOptions[O, A](tiling: A, rho: PbpoPlusRule[O, A], 
		countingClass: CountingClass)(implicit category: TerminationCategory[O, A]): Set[A] =
			import category.{codomain, pullback, rightInversesFor, image, ArrowOps}
			require(codomain(tiling) == rho.R1.get)
			
			/**
			* L1 <-l1- K1 -r1-> R1
			*          ^        ^
			*          |        |
			*          |   PB  tiling
			*          |        |
			*          X -----> T
			*/
		
			val `K1 <- X -> T` = pullback(Cospan(rho.r1.get, tiling))
			val `K1 <- X`: A = `K1 <- X -> T`.left
			val `X -> T`: A = `K1 <- X -> T`.right
			val `{K1 <- X <- T}`: Set[A] = 
			rightInversesFor(`X -> T`).map(`K1 <- X` o _)
			
			val inClass = countingClass.predicate
			val secondFactor = countingClass.secondFactor
			
			`{K1 <- X <- T}`.view.filter(f => inClass(rho.l1 o secondFactor(f)))
			.map(rho.l1 o _).toSet
	\end{minted}
\end{code}

\section{Conclusion and Future Work}
\label{section:conclusion}

We have introduced a termination method for graph transformation systems that can be utilized across frameworks, and which is defined in a broad array of categories. Our examples and comparisons with related work show that the method adds considerable value to the study of termination for graph transformation. We have also implemented the method in Scala, together with a REPL with which one can explore relative termination proofs iteratively.

Future work for strengthening the method includes solving the issues raised related to rule equivalence (Example~\ref{example:dpo:as:pbpo}) and antipatterns (Remark~\ref{remark:antipatterns}).
Methods for finding $\TT$, if it exists, and identifying useful sufficient conditions for the non-existence of $\TT$, would also be very useful. A possible metatheoretical direction for future research includes the question posed regarding \pbpoadhesivity{} (Remark~\ref{remark:pbpo:adhesive:question}).

\section*{Acknowledgment}

We thank the anonymous reviewers for their helpful suggestions. We also thank Nicolas Behr and Femke van Raamsdonk for some specific commentary on the PhD thesis of the first author that indirectly led to improvements in this paper. Both authors received funding from the Netherlands Organization for Scientific Research (NWO) under the Innovational Research Incentives Scheme Vidi (project.\ No.\ VI.Vidi.192.004). 

\bibliographystyle{alphaurl}
\bibliography{main}

\end{document}